\documentclass[reqno,12pt,letterpaper]{amsart}
\usepackage{amsmath,amssymb,amsthm,graphicx,mathrsfs,url}
\usepackage{amsbsy}
\usepackage[usenames,dvipsnames]{color}
\usepackage[colorlinks=true,linkcolor=Red,citecolor=Green]{hyperref}
\usepackage{floatrow}
\usepackage[usenames,dvipsnames]{xcolor}
\usepackage{tikz} 
\usepackage{slashed}
\usepackage{subcaption}
\usepackage{graphicx}
\usepackage[colorlinks=true]{hyperref}
\usepackage{color}
\usepackage{amsmath,amsfonts}
\usepackage{calrsfs}
\usepackage{amsmath,environ}
\usepackage{floatrow}
\DeclareMathAlphabet{\pazocal}{OMS}{zplm}{m}{n}

\usepackage[usenames,dvipsnames]{color}
\usepackage{amsfonts}
\usepackage{amsmath,amssymb,amsrefs}
\usepackage{graphicx}
\usepackage[mathscr]{euscript}
\usepackage[shortlabels]{enumitem}
\usepackage{bm}
\usepackage{hyperref}

\makeatletter
\def\serieslogo@{}
\makeatother
\makeatletter
\def\@setcopyright{}
\makeatother

\numberwithin{equation}{section}

\newtheorem{theorem}{Theorem}[section]
\newtheorem{corollary}[theorem]{Corollary}
\newtheorem{proposition}[theorem]{Proposition}
\newtheorem{lemma}[theorem]{Lemma}
\newtheorem{definition}[theorem]{Definition}
\theoremstyle{definition}
\newtheorem{remark}[theorem]{Remark}
\theoremstyle{remark}
\NewEnviron{equations}{%
\begin{equation}\begin{gathered}
  \BODY
\end{gathered}\end{equation}
}
\newcommand{\near}{{\operatorname{near}}}

\newcommand{\loc}{{\operatorname{loc}}}

\newcommand{\var}{{\vartheta}}
\newcommand{\edge}{{\operatorname{edge}}}
\newcommand{\bulk}{{\operatorname{bulk}}}
\newcommand{\Sf}{{\operatorname{Sf}}}

\newcommand{\systeme}[1]{\left\{ \begin{matrix} #1 \end{matrix} \right.}
\newcommand{\BBB}{\mathcal{B}}

\newcommand{\R}{\mathbb{R}}
\newcommand{\PPP}{\mathcal{P}}
\newcommand{\gap}{{\operatorname{gap}}}
\newcommand{\CCC}{\mathcal{C}}
\newcommand{\C}{\mathbb{C}}
\newcommand{\Z}{\mathbb{Z}}
\newcommand{\w}{{\omega}}

\newcommand{\GG}{{\mathcal{G}}}

\newcommand{\brill}{{\mathcal{B}}}
\newcommand{\matrice}[1]{\left[ \begin{matrix}
#1
\end{matrix} \right]}

\newcommand{\bk}{{\bf k}}

\newcommand{\RRR}{{\mathcal{R}}}
\newcommand{\CC}{{\CCC}}
\newcommand{\PP}{{\PPP}}

\newcommand{\Dom}{{\rm{Dom}}}

\newcommand{\bq}{{\bf q}}
\newcommand{\az}{\alpha}
\newcommand{\bK}{{\bf K}}
\newcommand{\bv}{{\bf v}}

\newcommand{\Ll}{{\mathscr{L}}}
\newcommand{\Hh}{{\mathscr{H}}}
\newcommand{\bl}{\ensuremath{\boldsymbol\ell}}
\newcommand{\br}{\ensuremath{\boldsymbol{r}}}
\newcommand{\bx}{{\bf x}}

\newcommand{\by}{{\bf y}}

\newcommand{\Id}{{\rm Id}}

\newcommand{\tvar}{{\widetilde{\var}}}

\newcommand{\ove}[1]{{\overline{#1}}}
\newcommand{\HH}{\pazocal{H}}

\newcommand{\far}{{\operatorname{far}}}
\newcommand{\blr}[1]{\left\langle #1 \right\rangle}
\newcommand{\Dirac}{{\slashed{\mathfrak{D}}}}

\newcommand{\ktilde}{{\bm{\mathfrak{K}}}}

\newcommand{\kpar}{{k_{\parallel}}}

\newcommand{\tk}{{\bm{\mathfrak{K}}}}
\newcommand{\E}{{\mathfrak{E}}}

\newcommand{\D}{\partial}
\usepackage{slashed}
\newcommand{\Di}{{\slashed{D}}{}}

\newcommand{\DiKs}{\slashed{D}{}^\bKs}
\newcommand{\DiK}{\slashed{D}{}^\bK}
\newcommand{\DiKp}{\slashed{D}{}^\bKp}

\newcommand{\tDir}{\widetilde{\slashed{D}}}
\newcommand{\tDiKs}{\tDir{}^\bKs}
\newcommand{\tDiK}{\tDir{}^\bK}
\newcommand{\tDiKp}{\tDir{}^\bKp}

\newcommand\bKs {{\bK_\star}}
\newcommand{\epsi}{\varepsilon}
\newcommand{\de}{ \ \mathrel{\stackrel{\makebox[0pt]{\mbox{\normalfont\tiny def}}}{=}} \ }

\newcommand{\tv}{{\bm{\mathfrak{v}}}}
\newcommand{\vtilde}{{\bm{\mathfrak{v}}}}
\newcommand{\TT}{\pazocal{T}}
\newcommand{\dist}{{\operatorname{dist}}} 

\newcommand{\hF}{{\widehat{F}}}
\newcommand{\vF}{{\rm v}_{\rm F}}

\newcommand{\eps}{\varepsilon}
\newcommand{\tHd}{\widetilde{H}{}^\delta}

\newcommand{\hPhi}{{\widehat{\Phi}}}
\newcommand{\ta}{\widetilde{a}}
\newcommand{\bt}{\widetilde{b}}

\newcommand{\nit}{\noindent}
\newcommand{\nn}{\nonumber}
\newcommand{\sgn}{sgn}

\newcommand{\dive}{{\operatorname{div}}}

\newcommand{\bKp}{{ {\bK^\prime}}}
\newcommand{\OO}{{\mathscr{O}}}
\newcommand{\lr}[1]{\langle #1 \rangle}

\newcommand{\hG}{{\widehat{G}}}

\begin{document}

\title[Edge states and the  valley Hall effect]
{Edge states and the  valley Hall effect}

 \author{A. Drouot}
\address{Department of  Mathematics, Columbia University, New York, NY, USA}
\email{alexis.drouot@gmail.com}

 \author{M. I. Weinstein}
 \address{Department of Applied Physics and Applied Mathematics and Department of Mathematics, Columbia University, New York, NY, USA}
 \email{miw2103@columbia.edu}

\begin{abstract}  

We study energy propagation along line-defects (edges) in two dimensional continuous, energy preserving periodic media.
The unperturbed medium (bulk) is modeled by a honeycomb Schroedinger operator, which is periodic with respect to the triangular lattice, invariant under parity, $\PP$,  and complex-conjugation, $\CC$. 
A honeycomb operator has {\it Dirac points} in its band structure: two dispersion surfaces touch conically at an energy level, $E_D$ \cite{FW:12,FLW-CPAM:17}. Periodic perturbations which  break $\PP$ or $\CC$  open a gap in the essential spectrum about energy $E_D$. Such operators model an insulator near energy $E_D$.

Our edge operator is a small perturbation of the bulk and models a transition (via a domain wall)
between distinct periodic, $\PP$ or $\CC$ breaking  perturbations. The edge operator permits energy transport along the line-defect. The associated  energy channels  are called {\it edge states}. They are time-harmonic solutions of the underlying wave equation, which are localized near and propagating along the line-defect. 
They are of great scientific interest due to their remarkable stability, and are a key property of topological insulators.

We completely characterize the edge state spectrum within the bulk spectral gap about $E_D$.
At the center of our analysis is an expansion of the edge operator resolvent  for energies 
near $E_D$.  The leading term features
 the resolvent of  an  {\it effective Dirac operator}. Edge state eigenvalues are poles of the resolvent, which bifurcate from the Dirac point. The corresponding eigenstates have the multiscale structure identified in \cite{FLW-2d_edge:16}.

We extend earlier work on  zigzag-type edges \cite{Drouot:18b} to all rational edges. We 
 elucidate the role in edge state formation played by the type of symmetry-breaking and the orientation of the edge.  We prove the  resolvent expansion by a new direct and transparent strategy. 
 Our results also provide a rigorous explanation of the numerical observations in \cite{FLW-2d_materials:15,LWZ:18};
    see also the photonic experimental study in \cite{Rechtsman-etal:18}. 
Finally we discuss implications  for the  {\it Valley Hall Effect}, which concerns quantum Hall-like energy transport in honeycomb structures.
% in the absence of a  magnetic field.
%
\end{abstract}

\maketitle

\section{Introduction}\label{introduction}
 Propagation of energy along an interface between different bulk media
   is a ubiquitous and important phenomenon in physics. In two-dimensional systems, the interface is a line defect and  the basic modes of propagation 
   are called {\it edge states}. These are time-harmonic solutions of the underlying wave equation, 
   localized near and propagating plane wave-like along the interface.
  % localized transversely to and  plane-wave like (propagating) parallel to the line-defect.
  %
  The bulk and defect models we study  are closely related to two-dimensional materials, one-atom-thick monolayers extending in-plane to the macro-scale.
    A paradigm is graphene, a two-dimensional honeycomb arrangement of carbon atoms which  is the most conductive known material, both electrically and thermally \cite{RMP-Graphene:09}.

  When suitably perturbed, graphene and related materials admit  edge states which are spectacularly robust against strong spatially localized perturbations.
  Many aspects of this stability can be understood using notions of topology  in terms of the Dirac points   and associated Floquet--Bloch modes, 
   of the  bulk honeycomb operator. 
   These are  conical singularities in the band spectrum; see \S\ref{bulk}.
 %  These are  conical singularities at quasimomentum / energy pairs $(\bKs,E_D)$, where $\bKs$ %varies over the vertices 
  % of the Brillouin zone; see \S\ref{bulk}.
  
  These propagation phenomena arise for general energy-preserving wave equations satisfying certain periodicity and symmetry assumptions. This has inspired investigations of  fabricated media, dubbed {\it artificial graphene}, in electronic physics, optics and photonics, acoustics and mechanics \cite{artificial-graphene:11,ozawa_etal:18,Marquardt:17,irvine:15}.
  The great interest in {\it topologically protected edge states} lies in applications of robust energy transport to technological settings.

 Motivated by graphene and its artificial analogs, we consider  continuum Schr\"odinger operators 
  which interpolate across a line defect between two weakly deformed honeycomb structures.  The types of  perturbed honeycomb operators or {\it edge operators} we consider  were introduced in \cite{FLW-2d_materials:15,FLW-2d_edge:16} to capture the essential features of theoretical and experimental work \cite{HR:07,RH:08,Soljacic-etal:08}. 

We consider two classes of edge operators: 
%\begin{itemize}
%\item[(i)] 

\nit (i)\ self-adjoint deformations of the bulk operator which break parity-inversion symmetry ($\mathcal{P}$) but preserve time-reversal invariance ($\mathcal{C}$);\small

%\item[(ii)]  
\nit (ii)\ self-adjoint deformations of the bulk operator which break $\CCC$ and preserve $\PPP$.
%\end{itemize}

In case (i), the model is a real-valued perturbation of a bulk honeycomb Schr\"odinger operator. In case (ii), the model is a divergence form elliptic operator, modeling a {\it bi-anisotropic} perturbation. Such models are of interest in the field of metamaterials;
   see, for example,  \cite{HR:07,RH:08,Shvets-PTI:13} in the physics literature, and the mathematical study \cite{LWZ:18}. Our methods  also apply to honeycomb magnetic
    Schr\"odinger operators \cite{Drouot:18b,Drouot:19}. 
  
Our (unperturbed) bulk Hamiltonian $H^0$ is periodic with respect to the equilateral triangular  lattice $\Lambda$ and satisfies the symmetry properties of a honeycomb potential; see  \eqref{symm-def}. A {\it rational edge} is a line defect which is  parallel to a fixed  $\tv_1\in\Lambda$. We construct an edge Hamiltonian $H_{\edge}^\delta$, which is a $\OO(\delta)$-perturbation of $H^0$. It is invariant under translations in $\Z\tv_1 $ but not invariant in other directions of $\Lambda$; see \S\ref{Hdelta}.

  {\it Edge states} are  time-harmonic solutions $e^{-iEt}\Psi(\bx)$ of the Schr\"odinger equation $i\D_t\psi=H_{\edge}\psi$. The spatial profile, $\Psi$, is localized in directions which are transverse to $\R\tv_1$, plane-wave like (propagating) parallel to $\R\tv_1$, {\it i.e.} $\Psi(\bx+\tv_1)=e^{i\kpar}\Psi(\bx)$. Here, $\kpar\in[0,2\pi]$ is referred to as the propagation constant, 
  parallel quasimomentum or parallel wave number.
%
 %with phase velocity $E/\kpar$. Here $\kpar$ is the parallel quasimomentum: the characteristic frequency of $\Psi$ in the $\tv_1$-direction. 
 Thus, $\big(\Psi,E\big)$ solves an eigenvalue problem:
 \begin{align}
 H_{\edge}^\delta\Psi\ & =\ E\ \Psi\nn \\
 \label{kp-evp}\qquad \Psi(\bx+\tv_1)\ &  =\ e^{i\kpar}\ \Psi(\bx)\qquad \ \textrm{(propagation parallel to $\R\tv_1$)}\\ 
\Psi(\bx)\ & \to\ 0\  \textrm{as}\  \systeme{ |\bx|\to\infty\\ \bx\cdot\tv_1=0} \quad \textrm{(localization transverse to $\R\tv_1$)}.\nn
\end{align}
We will reformulate \eqref{kp-evp} in a function space $\Ll^2_\kpar$, $0\le\kpar\le2\pi$, which incorporates the boundary conditions of   \eqref{kp-evp}; see \S\ref{sec:1.3} and \S\ref{FB-cyl}. 

There are two types of rational edge orientations:  zigzag-type  and  armchair-type;  see \S\ref{ZZ-AC}.
Edge state diagrams are plots of
  $E(\kpar)$ vs. $\kpar$. The global character of these diagrams depends strongly on the type of edge and on the manner in which the symmetries are broken by the perturbation; see Figures \ref{CnoP} and \ref{PnoC}. 

    \begin{figure}
  \includegraphics{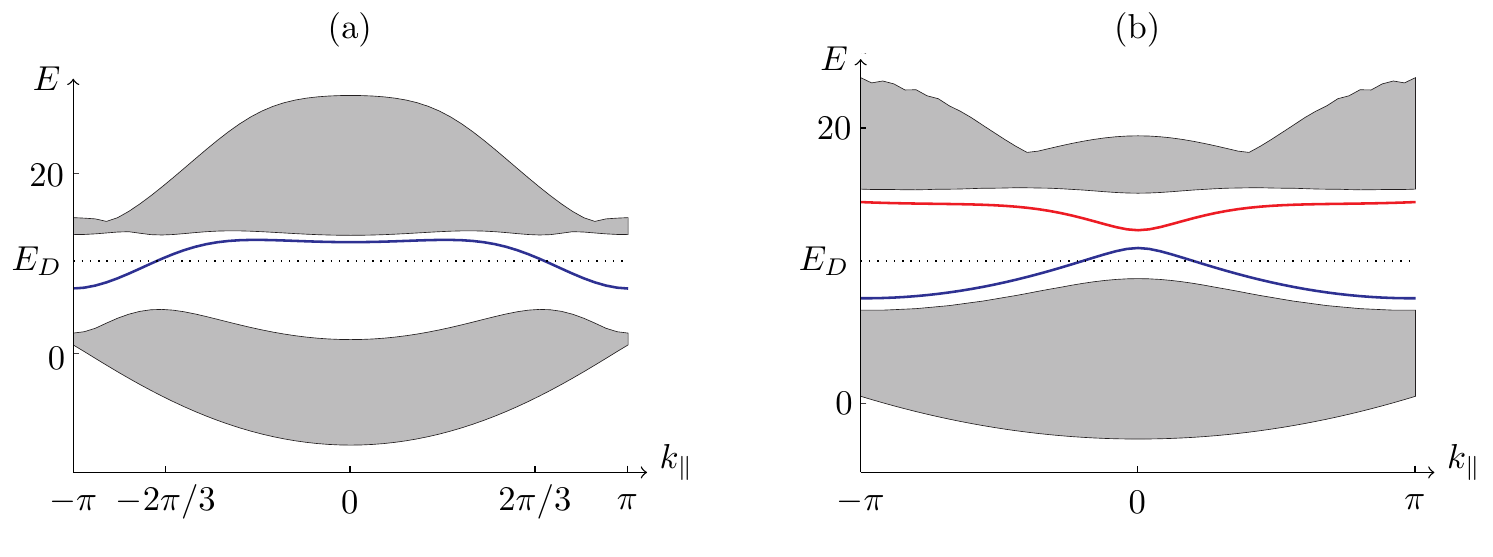}
  \caption{Numerical simulations of edge state curves (blue/red) for $\PP$-breaking, $\CC$-preserving deformations ($H^\delta_{\edge} = H^\delta$) and for (a) a zigzag edge; (b) an armchair edge. Here $\delta = 3$ and the bulk spectrum is in grey. 
   }\label{CnoP} %fig:1a
\end{figure}

In this paper we advance the spectral analysis 
of $H^\delta_{\edge}$  initiated in \cite{FLW-2d_materials:15,FLW-2d_edge:16}
 and continued in \cite{Drouot:18b,Drouot:19}.
   Edge states which bifurcate from Dirac points were first constructed for zigzag-type edges in \cite{FLW-2d_edge:16,Drouot:18b} and in a related 1D model in \cite{FLW-PNAS:14,FLW-MAMS:17}. See  also \cite{DFW:18,Drouot:18a} for extensions and refinements of \cite{FLW-PNAS:14,FLW-MAMS:17} in the context of a larger family of dislocation operators.  
 In \cite{FLW-2d_edge:16}  a Schur complement  / Lyapunov--Schmidt strategy was used. The paper  \cite{Drouot:18b} characterized all zigzag-type edge states via a resolvent expansion \cite{DFW:18}. This characterization implies that  $\CCC$-breaking induced edge states are topologically protected \cite{Drouot:18b,Drouot:19}.
   %
   %\footnote{\tr{Clarified previous results}}
   
We summarize the consequences of  Theorem \ref{res-exp}, and Corollaries \ref{KKp-eigsZZ}-\ref{KKp-eigsAC}:
  \begin{itemize}
  \item We provide a complete and detailed description of the spectrum of $H^\delta_{\edge}$
    in a neighborhood of  the Dirac energy, $E_D$,  for all rational edges and small $\delta$.
   \item   Our proof extends results of \cite{Drouot:18b}  to all rational edges, and thus encompasses the more subtle case of armchair-type edges. Our method of proof unifies the approaches of  \cite{FLW-2d_edge:16} and \cite{Drouot:18b}. Motivated by 
   \cite{FLW-PNAS:14,FLW-MAMS:17,FLW-2d_edge:16}  we use a
    Lyapunov-Schmidt / Schur reduction strategy to obtain resolvent expansions.

 \item
 We interpret the robustness of edge states, following the analysis of \cite{Drouot:18b,Drouot:19}.
\item  
We discuss implications of our results for the  {\it Valley Hall Effect}, which concerns quantum Hall-like energy transport in honeycomb structures in the absence of a  magnetic field; see \S\ref{vhall}.
   \end{itemize}

A brief summary of our results is in \S\ref{sec:1.3}. The detailed theorems appear in \S\ref{sec:5a}.

\begin{figure}
        \includegraphics{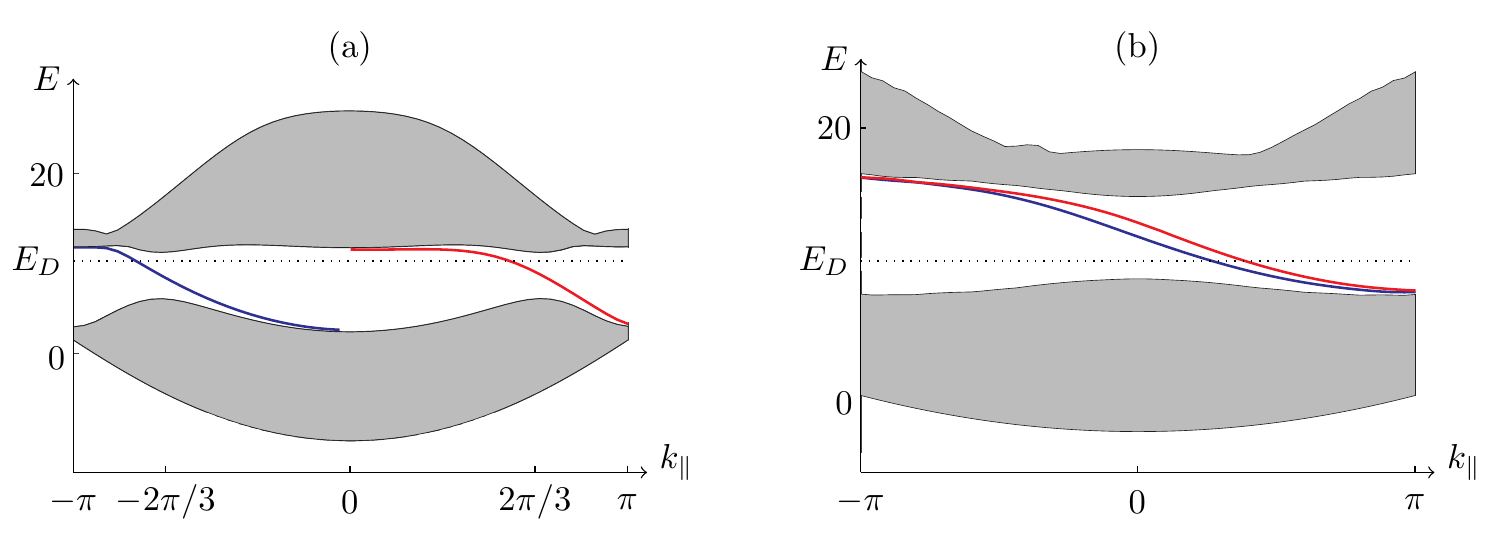}
    \caption{Numerical simulations of edge state curves (blue/red) for $\CC$-breaking, $\PP$-preserving deformations ($H^\delta_{\edge} = \tHd$) and for (a) a zigzag edge; (b) an armchair edge. The edge state curves traverse the bulk spectral gap. The spectral flow (signed count of eigenvalues crossing the gap) is $2$. It equals the difference of Chern numbers of low-lying eigenbundles at either side of the edge; see \cite{Drouot:18b,Drouot:19} and \S\ref{topology}.}\label{PnoC} %fig:2a
\end{figure}

\subsection{Honeycomb operators and Dirac points}\label{bulk}
We begin with a periodic self-adjoint Schr\"odinger operator: 
\begin{equation}
 H^0 \de -\Delta+V(\bx)\ \quad \textrm{ acting on $L^2=L^2(\R^2)$},\label{H0}
 \end{equation}
 where $V$ is real-valued and periodic with respect to the equilateral lattice $\Lambda=\Z\bv_1\oplus\Z\bv_2$. The corresponding dual lattice is 
\begin{equation}
\Lambda^*=2\pi \Z\bk_1\oplus 2\pi\Z\bk_2 \ \ \ \ \text{where} \ \ \ \ \bk_m \cdot \bv_n = \delta_{mn},\ \ m,n=1,2;
\nn\end{equation} 
see \S\ref{notation}.  The Brillouin zone $\brill\subset \R^2$ is a choice of fundamental cell of $\R^2/\Lambda^*\cong\mathbb{T}^2$, the regular hexagon in Figure \ref{fig:3}.

 For $\bk\in\R^2$, we let $L^2_\bk$ denote the space of $\bk$- pseudoperiodic functions:
 \begin{equation}
L^2_\bk \de \left\{ u \in L^2_{\rm loc}(\R^2) : \ u(\bx+\bv) = e^{i\bk\cdot\bv}\ u(\bx),\ \bv\in\Lambda \right\} .
\end{equation}
  Note that $L^2_{\bk+\bq}=L^2_\bk$ for $\bq\in\Lambda^*$, and hence $L^2_\bk$ is $\Lambda^*$-periodic in $\bk$.
For $\bk \in\R^2$, we let $H^0_\bk$ be equal to  $ H^0$ acting on $L^2_\bk$. This operator 
has  discrete spectrum  denoted $E_1(\bk)\le \dots\le E_b(\bk)\le\dots$, listed with multiplicity. The {\it dispersion relations} of $H^0$ are the eigenvalue maps $\bk \in \R^2 \mapsto E_b(\bk)$. These are $\Lambda^*$-periodic and Lipschitz continuous functions of $\bk$ \cite{avron-simon:78,FW:14}.  The spectrum of $H^0$ acting on $L^2=L^2(\R^2)$ is the union of the intervals (spectral bands) $E_b(\brill)$:
%\begin{equation}
$\sigma_{L^2}\left(H^0\right) = \bigcup_{b=1}^\infty E_b(\brill)$.
%\end{equation}
The collection of dispersion relations and corresponding eigenmodes form the {\it band structure} of $H^0$.
 
\begin{figure}
         \includegraphics{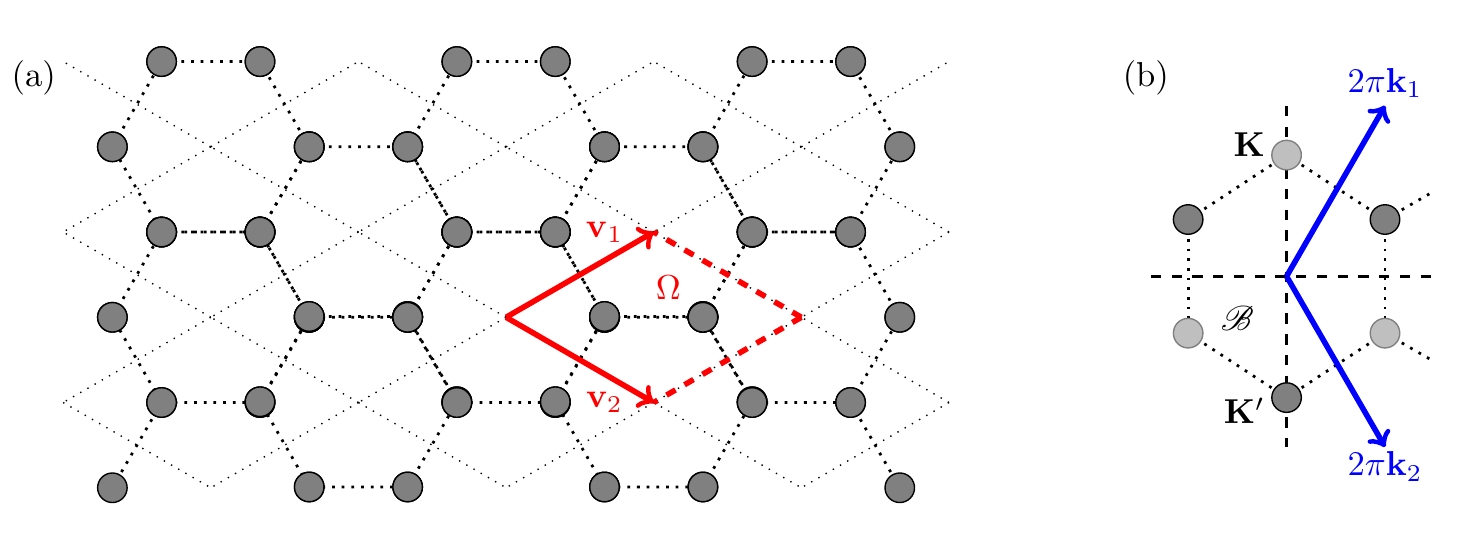}
    \caption{(a) Equilateral triangular lattice $\Lambda = \Z\bv_1 \oplus \Z\bv_2$. The circles make up the honeycomb structure --
    the union of two interpenetrating triangular sublattices. (b) Dual lattice $\Lambda^* = 2\pi \Z\bk_1\oplus 2\pi\Z\bk_2$, with Brillouin zone $\BBB$ and independent high-symmetry quasimomenta  $\bK+\Lambda^*$ and $\bK^\prime+\Lambda^*$. 
   }\label{fig:3}
\end{figure}

The function $V \in C^\infty(\R^2)$ is a {\it honeycomb  potential} if  $V$ is
 real-valued,  $\Lambda$-periodic, even and $2\pi/3$--rotationally invariant; see \cite[Definition 2.1]{FW:12}: 
\begin{equations}\label{symm-def}
 [\CCC,V(\bx)]=0,\quad [\PP ,V(\bx)]=0,\quad  
  [\RRR,V(\bx)]=0, \ \ \ \  \text{ where} \\
\CC[f](\bx) \de \overline{f(\bx)},\quad \PP[f](\bx) \de f(-\bx),\quad \RRR[f](\bx) \de f(R^*\bx)\ ,
\end{equations}
and $R$ denotes the $2\pi/3$--rotation in the plane.  The single electron model of graphene corresponds to $V$ equal to a sum of ``atomic potential wells'' over the set of  honeycomb vertices; see \cite[Section 2.3]{FW:12}.

 The vertices of the regular hexagon, $\brill$, are the points labeled $\bK$, $\bKp$ in Figure \ref{fig:3} and their rotations about the origin by $2\pi/3$. They play a distinguished role in the spectral theory of $H^0$: the commutator $[H^0_\bk,\RRR]$ vanishes if and only if $\bk$ is equal to $\bK$ or $\bKp$ modulo $\Lambda^*$.  Thus, in addition to $\PP\CC$-symmetry, $H^0_\bK$ and $H^0_{\bKp}$ are $\RRR$-invariant. The vertices $\bK$ and $\bKp$, and their dual lattice translates,
  are often called {\it high-symmetry quasimomenta}. Due to this extra rotational symmetry,  
   $L^2_\bK$ decomposes as an orthogonal sum over the eigenspaces of $\RRR$: 
   \[L^2_\bK=L^2_{\bK,1}\oplus L^2_{\bK,\tau}\oplus L^2_{\bK,\bar\tau}.\] 
Here, $1, \tau$ and $\bar\tau$ are the three cube roots of unity ($\tau=e^{2\pi i/3}$) and $L^2_{\bK,\omega}$ is the subspace of $L^2_\bK$ consisting of functions for which $\RRR f=\omega f$.
  
  Using the above observations, it was proved in \cite{FW:12,FLW-MAMS:17}  that
for generic honeycomb potentials $V$,  the band structure of  $H^0 = -\Delta + V$  has  Dirac points at  
the vertices of $\brill$;  see also \cite{C91,G09,LWZ:18,berkolaiko-comech:18,Aea:18}.
This means that there exist  $E_D$, $\vF>0$ and $b_*\ge1$ such that for $\bKs=\bK, \bKp$, the operator $H^0_\bKs$  has a double eigenvalue at energy $E_D$, at which two dispersion surfaces touch conically:
\begin{align}
\label{cones}
\begin{split}
E_{b_\star+1}(\bk)  &=  E_D + \vF\ | \bk-\bKs | \cdot  \big( 1\ +\ o(| \bk-\bKs | ) \big), \\
E_{b_\star}(\bk)  &=  E_D - \vF \ | \bk-\bKs | \cdot  \big( 1\ +\ o(| \bk-\bKs | ) \big), 
\quad \vF>0, \quad \text{$\bk$ near $\bKs$.}
\end{split}
\end{align}

  For spatially localized initial conditions the Schr\"odinger evolution disperses: $e^{-iH^0t}f$ spreads and decays as $t$ increases. When $f$ is spectrally concentrated in energy / quasi-momentum about a Dirac point, the effective evolution on large, finite time scales follows a time-dependent Dirac equation \cite{FW:14}. This explains the relativistic behavior of quasi-particles (wavepackets) in graphene \cite{RMP-Graphene:09}.  Dirac points persist under small $\Lambda$-periodic $\PP \CCC$-invariant perturbations, see \cite[Section 9]{FW:12}; in this case a tilted Dirac equation governs the character of its wave-packet evolution.

  This paper further explores how perturbations of $H^0$ affect the dynamics of wave-packets which are spectrally concentrated near $E_D$. If  $H^0$ is perturbed to $H=H^0+ Q$, where $Q$ is smooth, 
  real-valued and spatially localized, then the essential spectrum of $H$ is equal to that of $H^0$; this is Weyl's stability theorem \cite{RS4}. Since $E_D\in\sigma_{\rm ess}(H)$, we expect the dispersive character of the dynamics near energy $E_D$ to persist under such perturbations.
In contrast, a class of $\Lambda-$ periodic perturbations of  $H^0$ at spatial infinity which break $\PP\mathcal{C}$--invariance destabilizes Dirac points; see \cite[Remark 9.2]{FW:12}. Such perturbations may open a gap in the essential spectrum about $E_D$ and may produce defect state energies  within this gap. In the following sections we leverage this instability to construct Hamiltonians with edge states with energies near $E_D$.%
 
 \subsection{The edge state eigenvalue problem for rational edges}\label{general}
 
 We are interested in Hamiltonians $H_{\edge}^\delta$ that are line-defect perturbations of the honeycomb  Schr\"odinger operator $H^0$. We  briefly describe the construction of $H_{\edge}^\delta$ and refer to \S\ref{Hdelta} for details. 
 
 We start with the bulk honeycomb operator $H^0$ and we perform small (size $\delta$) $\Lambda-$ periodic deformations of $H^0$ at infinity,
  on either side of the line $\R\tv_1$. This yields perturbed bulk operators $H^\delta_{\bulk,\pm}$ at infinity, which are $\Lambda-$ periodic. In particular,  $H^\delta_{\bulk,\pm}$ commutes with translations in $\Z\tv_1$.
The operators $H^\delta_{\bulk,+}$ and $H^\delta_{\bulk,-}$ act on the spaces: 
 \begin{equation}\label{eq:1b}
 \Ll^2_\kpar \de \left\{ f \in L^2_\loc(\R^2) \ : \ f(\bx+\tv_1) = e^{i\kpar} f(\bx), \ \int_{\R^2/\Z\tv_1} |f(\bx)|^2 \ d\bx < \infty  \right\}, \ 
 \end{equation}
where $0\le\kpar\le2\pi$,   share a common spectral gap   of width $O(\delta)$ about energy $E_D$.  

We introduce an edge operator $H^\delta_{\edge}$, which  interpolates slowly and transversely  to the edge $\R\tv_1$ (length-scale $\delta^{-1}$), between $H^\delta_{\bulk,+}$ and $H^\delta_{\bulk,-}$. The interpolation is implemented via a {\it domain wall function};
  see \S\ref{edge-op}.   $H^\delta_{\edge}$ has only restricted periodicity; it commutes only with translations in $\Z \tv_1$.
  
We consider two types of asymptotic bulk operators, denoted $H^\delta_{\bulk,\pm}=H^\delta_\pm$  and 
$H^\delta_{\bulk,\pm}=\tHd_\pm$; see \cite{HR:07,RH:08,Shvets-PTI:13,LWZ:18}:
 \begin{itemize}
 \item[(i)] $\CCC$ is preserved and $\PPP$  is broken:
 \begin{equation}
  H^\delta_\pm \de -\Delta+V(\bx) \pm \delta \cdot W(\bx), \qquad \  W \textrm{ is odd and $\Lambda$-periodic}.
 \label{CCnoPP} \end{equation}
 \item[(ii)] $\PPP$ is preserved and $\CCC$  is broken;
  \begin{equations}\label{PPnoCC}
  \tHd_\pm  \de - \Delta + V(\bx) \pm \delta \cdot \dive\big( A(\bx) \cdot \nabla \big),\\
     A(\bx)  =\ i\matrice{0 & -a(\bx) \\ a(\bx) & 0}= a(\bx)\sigma_2\ , \qquad \textrm{ even and $\Lambda$-periodic}.
      \end{equations}
 \end{itemize}
The corresponding edge operators are denoted: $H_{\edge}^\delta=H^\delta$ in case (i) and $H_{\edge}^\delta=\tHd$ in case (ii). The global character of their edge state curves ($\kpar$ versus $E$) are quite different. This has dynamical and topological consequences,  see \S\ref{sec:1.3}-\ref{topology} and \S\ref{topology2}. 

\subsection{Summary of main results}\label{sec:1.3} Our main results are Theorem \ref{res-exp} and Corollaries \ref{KKp-eigsZZ} and \ref{KKp-eigsAC}. They fully describe the {\it $\kpar$-edge states} $(\Psi,E)$ of $H^\delta_{\edge}$, i.e. the solutions of 
             \begin{equation}
    H^\delta_{\edge}\Psi\ =\ E\ \Psi,\qquad \Psi\in \Ll^2_\kpar, \ \ \ \ H^\delta_{\edge} = H^\delta \text{ or } \tHd,
    \label{kp-evp1}
    \end{equation}    
    for $\kpar$ near $\bK \cdot \tv_1$ or $\bKp \cdot \tv_1$ and energies near $E_D$. 
Our results cover all rational edges (going beyond \cite{Drouot:18b}) for both $\PP$-- and $\CC$--breaking deformations of the bulk honeycomb operator, $H^0$. Hence this work gives a complete picture of the role in edge state formation and stability played by a) the manner in which symmetries are broken by the edge operator and b) the orientation of the edge, $\R\tv_1$.

Theorem \ref{res-exp} presents an expansion of  the $\Ll^2_0$-resolvent of $H_\edge^\delta$
for small $\delta$  and energies near $E_D$; note from \eqref{eq:1b} that
 $\Ll^2_0=L^2(\R^2/\Z\tv_1)$. 
For simplicity, we specialize  here to an armchair-type edge, $\R\tv_1$, and a $\CC$-preserving deformation ($H^\delta_{\edge} = H^\delta$), with $\kpar = \bK \cdot \tv_1 = \bKp \cdot \tv_1 = 0$. This case contains 
all essential characteristics and technical hurdles encountered in the general rational edge case.

The leading order term in the resolvent expansion is given in terms of the resolvent of an
effective Dirac operator.  Specifically, there exist:
\begin{itemize}
\item explicit bounded operators $J_\delta : \Ll^2_0 \rightarrow L^2(\R,\C^4)$,\ 
$J^*_\delta : L^2(\R,\C^4)\rightarrow \Ll^2_0 $;
\item a $\delta$-independent Dirac operator $\Dirac$ acting on $L^2(\R,\C^4)$;
\end{itemize}
such that if $z \not\in \sigma_{L^2}(\Dirac)$ and $\delta \rightarrow 0$, then 
\begin{equation}\label{eq:1h}
\left. \left( H^\delta-E_D - \delta z \right)^{-1}\right|_{\Ll^2_0} \ \sim \ \dfrac{1}{\delta} J_\delta^* \cdot \left( \Dirac - z\right)^{-1} \cdot J_\delta.
\end{equation}
The equation \eqref{eq:1h}  shows that $J_\delta$ approximately intertwines the operators $\Dirac$ and $H^\delta$ after the energy recentering and rescaling: $E=E_D +\delta z$. 

Corollary \ref{KKp-eigsAC}, a consequence of Theorem \ref{res-exp}, provides a bijection between the eigenvalues and spectral projections of $\Dirac$ and those of $H^\delta$ for energies near $E_D$. The operator $\Dirac$ has a block-diagonal structure:
\begin{equation}\label{eq:1q}
\Dirac = \matrice{\DiK & 0 \\ 0 & \DiKp}\ .
\end{equation}
 The operators $\DiK$ and $\DiKp$ are Dirac operators with spatially varying mass terms.
 They  act in $L^2(\R;\C^2)$, and control, for $\delta$ small,  the bifurcation of edge states from the Dirac points $(\bK,E_D)$ and $(\bKp,E_D)$. By results presented
    in \S\ref{eff-Dirac}, their point spectra are the same.  Hence, the point spectrum of $\Dirac$ consists of double eigenvalues; this explains why armchair-type edge state curves come in pairs; see Figure \ref{CnoP}, the numerical observations in \cite{FLW-2d_materials:15,LWZ:18} and the photonic results \cite{Rechtsman-etal:18}.
    
    \begin{remark}\label{0mode}
  Recall that $H^\delta_\edge$ is constructed as an adiabatic interpolation between asymptotic perturbed bulk operators  $H^\delta_{\bulk,\pm}$ via a {\it domain wall function};
  see \S\ref{edge-op}. It follows that the emergent Dirac operators have mass terms
  with opposite sign at $\pm\infty$. Thus, $\DiK$ and $\DiKp$ have a simple
     $L^2(\R;\C^2)-$ eigenvalue at zero energy; see Proposition \ref{cor-D-spec}.
The persistence of this
     zero mode of $\DiKs$ against localized perturbations guarantees that the spectra of $\DiK$ and $\DiKp$  seed at least one bifurcating edge state curve of $H^\delta_\edge$; see \cite{FLW-2d_edge:16}.
      See \cite{B19a,B19b,B19c} for topological interpretations. 
     \end{remark}

\subsection{Dynamical perspective} Theorem \ref{res-exp} and Corollaries \ref{KKp-eigsZZ}, \ref{KKp-eigsAC} apply to all rational edges, both $\PP$- and $\CC$-breaking perturbations; and quasimomenta $\kpar = \bKs \cdot \tv_1 + \delta \mu$ with $\bKs = \bK$ or $\bKp$ and $\mu = O(1)$. They involve $4 \times 4$ Dirac operator $\Dirac(\mu)$ with the block-diagonal structure \eqref{eq:1q} made up of two $2 \times 2$ Dirac operators $\DiK(\mu)$ and $\DiKp(\mu)$. Their spectra as functions of $\mu$ are represented on the left panels of Figure \ref{fig:7}. We derive the operators $\DiK(\mu)$ and $\DiKp(\mu)$ 
 and their spectra in \S\ref{eff-Dirac}.

 An analogous discussion to that given above for the case $\mu=0$ shows that the spectra of $\DiKs(\mu)$ determine the $\Ll^2_{\bKs \cdot \tv_1+\delta \mu}$-spectrum of $H^\delta_{\edge}$ near $E_D$.  The numerical simulations displayed in Figures \ref{CnoP} and \ref{PnoC} illustrate our main results and their dynamical implications:
\footnote{We are grateful to Y. Zhu and P. Hu for providing  Figures \ref{CnoP} and  \ref{PnoC}. The actual simulations
contain additional numerically-induced boundary modes, which are localized at the artificial computational boundary, which
is distant from the line defect. These spurious modes are not displayed in Figures \ref{CnoP}  and \ref{PnoC}.
The displayed edge state curves in Figures \ref{CnoP}  and \ref{PnoC} do not change as the size of the computational domain is increased.  Figures which display such modes are presented  in \cite{LWZ:18}. }
\begin{itemize}
\item  Figure \ref{CnoP}   displays edge state curves for $H^\delta_{\rm edge}$, which is  $\CCC-$ invariant. 
The two panels display the $(\kpar,E)$-plots of solutions of \eqref{kp-evp} for (a) the zigzag edge and (b) the armchair edge. By $\CCC-$ invariance,  the edge state curves are symmetric about $\kpar = 0$. Thus, about a fixed energy near $E_D$, there are wave-packets solutions of the time-dependent Schr\"odinger equation which travel in either direction along the edge.
Hence, wave-packets designed to propagate unidirectionally along the edge will scatter off localized imperfections and excite waves propagating in the opposite direction.

\item Figure \ref{PnoC}  displays edge state curves when the deformation breaks $\CCC-$ symmetry. Both edge state curves traverse the bulk spectral gap downward. Correspondingly, wave-packet solutions of the time-dependent Schr\"odinger equation travel unidirectionally (with negative group velocity). 
Wave-packets designed to propagate unidirectionally, when encountering localized imperfections, are not expected to backscatter: there is no accessible energy channel for propagation in the opposite direction. This argument does not however rule out some scattering into the bulk.
\end{itemize}

\subsection{A topological perspective}\label{topology}
In this section we explain the sense in which
(a) the family of edge states of Figure \ref{CnoP} is not topologically protected, and 
 (b) the family displayed in Figure \ref{PnoC} is topologically protected.
These conclusions follow from the detailed arguments in \cite{Drouot:18b,Drouot:19},
 applied to the present context.

Consider the  family of Fredholm operators 
\begin{equation}\label{eq:1r}
\kpar\to \left. H^\delta_{\edge}\right|_{\Ll^2_\kpar},  \ \ \ \ \kpar \in \R .
\end{equation}
This family depends periodically on $\kpar$. 
For each $\kpar$, there is a gap in the essential spectrum about
 energy $E_D$. 
 The {\it spectral flow} $\Sf\big(H^\delta_\edge,E_D\big)$
 of the family \eqref{eq:1r} is the signed number of eigenvalues crossing the energy level $E_D$ as $\kpar$ runs through $\R/2\pi\Z$; see \cite{Waterstraat:16} for an introduction. The count is $+1$ if the eigenvalue traverses $E_D$ downwards and $-1$ if it traverses $E_D$ upwards. The spectral flow is an integer-valued  topological invariant: it remains unchanged even against large compact operator perturbations of $H^\delta_{\edge}$ -- and more generally against gap-preserving deformations. 
  
For $\delta$ small,  the operator $\left. H^\delta_{\edge}\right|_{\Ll^2_\kpar}$ has no point spectrum near $E_D$ unless $\kpar$ is  near $\bK\cdot\tv_1$ or $\bKp\cdot\tv_1$. In each of these cases, the spectral characteristics are encoded in  the Dirac operators $\DiK(\mu)$ and $\DiKp(\mu)$. At the level of spectral flow, we have
 \begin{equation}
\Sf\left(H^\delta_\edge,E_D\right) = \Sf\left(\DiK,0\right) + \Sf\big(\DiKp,0\big);
\end{equation}
see \cite[\S1.7 and \S7]{Drouot:18b}. 
Furthermore, we have $\Sf(\DiKs,0)\in\{-1,+1\}$, where the sign is given explicitly in terms of  parameters in $\DiKs$.
This is a consequence of ODE arguments;
 see \cite{DFW:18,Drouot:18b} and the discussion in \S\ref{eff-Dirac}. See also \cite{B19a,B19b,B19c} for an in-depth topological and transport study of models built from these Dirac operators.

For $H^\delta$, the spectral flow vanishes because edge state energy curves are symmetric about $0$; see Figure \ref{CnoP}. One can remove such edge states via specifically designed compact operator perturbations; the family of edge states is not topologically protected.
 %This indicates non-topologically protected transport. In terms of physical mechanisms, specifically designed wave-packets %that ideally propagate unidirectionally must eventually scatter and couple into waves propagating in the opposite direction.
%

 The global character of the edge state energy curves for $\tHd$ is quite different; 
see Figure  \ref{PnoC}. In this case, the spectral flow is equal to $+2$ or $-2$; two families of edge states  persist across the bulk spectral gap and are stable against, for example, compact perturbations.
The family of edge states is topologically protected. We refer to \cite{Drouot:18b} for proofs and to \cite{Drouot:19} for extensions beyond gap preserving perturbations.

The existence (or absence) of topologically protected edge states relates to a principle called the {\it bulk-edge correspondence}; see, {\it e.g.}
 \cite{H93,EG02,KRS02,KS04a,KS04b,EGS05,T14,D:19,BKR17,BR18,B19}.  Since the gap at energy $E_D$ is open, we may associate to $H^\delta$ 
 two smooth vector bundles with base given by the torus $\R^2/\Lambda^*$. The fibers are
 eigenspaces   of $H_{+,\bk}^\delta$ and of  $H_{-,\bk}^\delta$ ($\bk\in\R^2/\Lambda^*$) associated with the bulk spectral band below $E_D$. The Chern number is an integer obtained by integrating a bundle curvature  over the base $\R^2/\Lambda^*$. The bulk-edge correspondence = anticipates  that the difference of Chern numbers associated to $H^\delta_+$ and $H^\delta_-$ is equal to $\Sf\left(H^\delta_\edge,E_D\right)$. The analogous prediction holds for $\tHd$.
 
For $H^\delta$, both Chern numbers vanish, by $\CCC-$ invariance.
Such structures are called topologically trivial. For $\tHd$, the difference of Chern numbers equals $\pm 2$. See  \cite{Drouot:19}, where the calculation is reduced to a standard two-band model.  Therefore, in both types of deformation, the spectral flow (boundary index) is equal to the difference of Chern numbers (bulk indices), a \textit{quantitative} bulk-edge correspondence for Hamiltonians $\tHd$ \cite{Drouot:19}.
From a dynamical point of view, $\tHd$ represents a configuration of materials that is insulating in the bulk but exhibits topologically stable transport along its line defect. It is a non-trivial example of topological insulator in a continuum PDE setting.

\subsection{Organization of the paper}\label{org} In \S\ref{prelim}, we review the Floquet-Bloch theory
 of $\Lambda$- and $\Z\tv_1$-periodic operators on $L^2(\R^2)$.  In \S\ref{Hdelta}
  we construct our class of edge operators $H^\delta_{\edge}$. We fix a rational edge direction and interpolate, 
 via a domain wall, between slightly deformed bulk operators $H^\delta_{\bulk,\pm}$ of the form \eqref{CCnoPP} or \eqref{PPnoCC}.  There are three key hypotheses: 
\begin{itemize}
\item[(H1)] The unperturbed bulk honeycomb operator $H^0=-\Delta+V$ has Dirac points;
\item[(H2)] $H^0$ satisfies the spectral no-fold condition (stated physically, $H^0$ is {\it semi-metallic} at energy $E_D$);
\item[(H3)]  The deformed bulk operators $H^\delta_{\bulk,+}$ and $H^\delta_{\bulk,+}$ satisfy a (generic) non-degeneracy condition.
\end{itemize}  
  
%\footnote{\tg{I shortened this part. I thought it emphasize too much the properties of the Dirac %operators, somewhat obscuring the core of the strategy}}   
%
In \S\ref{ms-Hd} a multiscale expansion \cite{FLW-2d_materials:15,FLW-2d_edge:16,LWZ:18}
 is used to construct approximate edge states as slowly varying linear combinations of Dirac point (energy-degenerate) Floquet--Bloch modes.
  The slowly varying mode-amplitudes are governed by a system of Dirac equations, whose spectral properties are summarized in \S\ref{D-eff-spec}.
This approximate construction implies the existence of genuine edge states with energies in the gap
 about the energy $E_D$; see \S\ref{app2gen}. 
Additionally, the analysis suggests  a classification of rational edges in two types: zigzag-type and  armchair-type. These depend on whether there is coupling between spectral components  near the quasi-momentum sublattices $\bK+\Lambda^*$ and $\bKp+\Lambda^*$; see \S\ref{ZZ-AC}.

 In \S\ref{main-results}  we present our main results: Theorem \ref{res-exp}, a resolvent expansion for energies
  near $E_D$, and Corollaries \ref{KKp-eigsZZ}, \ref{KKp-eigsAC}, a precise characterization of the point spectrum of $H^\delta_{\edge}$ near energy $E_D$. In \S\ref{topology2}-\ref{sec:5.3}, using topological arguments presented in \cite{Drouot:18b,Drouot:19}, we discuss the global properties of edge state curves which bifurcate from the Dirac point energy -- in the $\CC-$ invariant and $\CC-$ breaking cases. 
  The mathematical core of this paper is the proof of the resolvent expansion for general rational edges, Theorem \ref{res-exp}, via a new strategy.

\subsection{Notation and Conventions}\label{notation}
\begin{itemize}
\item We let $\Lambda=\Z\bv_1\oplus\Z\bv_2$ be the equilateral lattice and $\Lambda^*=2\pi\Z\bk_1\oplus2\pi\Z\bk_2$ be its dual:\\
 \begin{equations}\nn \bv_1=a\matrice{\sqrt3/2\\ 1/2},\quad 
  \bv_2=a\matrice{\sqrt3/2\\ -1/2} \ \ \text{ and } \ \ 
  \bk_1= a\matrice{1/2\\ \sqrt3/2},\quad 
  \bk_2=a\matrice{1/2\\ -\sqrt3/2},
  \end{equations}
   where $a^2=2/\sqrt{3}$ so that $|\bv_1\wedge\bv_2|=1$  and $\bk_j\cdot\bv_l=\delta_{jl}$.
   A fundamental period cell and Brillouin zone, $\brill$, are depicted in Figure \ref{fig:3}.
   \item High symmetry quasimomenta: 
   \begin{equation}
   \bK=\frac{2\pi}{3}(\bk_1-\bk_2),\qquad \bKp=-\bK;\label{K-hs}
   \end{equation}
see Figure \ref{fig:3} and \S\ref{HLP}.
   \item We will use $L^2$-based spaces: $L^2 = L^2(\R^2)$; quasi-periodic functions w.r.t. $\Lambda$: $L^2_\bk = L^2_\bk(\R^2 /\Lambda)$; quasi-periodic functions w.r.t. $\Z \tv_1$: $\Ll_\kpar^2 = \Ll^2_\kpar(\Sigma)$, where $\Sigma = \R^2/\Z\tv_1$. We define analogously Sobolev spaces $H^s$, $H^s_\bk$ and $\mathscr{H}^s_\kpar$. 
\item The Pauli matrices are
\begin{equation}\label{eq:1s}
\sigma_1 = \matrice{0 & 1 \\ 1 & 0}, \ \ \sigma_2 = \matrice{0 & -i \\ i & 0}, \ \  \sigma_3 = \matrice{1 & 0 \\ 0 & -1}.
\end{equation}
They satisfy $\sigma_j^2 = \Id$ and $\sigma_j \sigma_m = -\sigma_m \sigma_j$ for $j \neq m$.
\item $\CC$, $\PP$ and $\RRR$ denotes respectively the complex conjugation, parity inversion and $2\pi/3$ rotation; see \eqref{symm-def}. We shall refer to $\CC-$, $\PP-$ and $\PP\CC-$ invariant operators. 
\item If $\HH$ is a Hilbert space and $A$ is selfadjoint  operator on $\HH$, we denote the spectrum of $A$ by $\sigma_\HH(A)$.
\item If $\HH$ and $\HH'$ are Hilbert space and $\psi \in \HH$, we write $|\psi|_\HH$ for the norm of $\HH$; if $A : \HH \rightarrow \HH'$ is a bounded operator, the operator norm of $A$ is
\begin{equation}
\| A \|_{\HH \rightarrow \HH'} \de \sup_{|\psi|_\HH=1} |A\psi|_{\HH'}.
\end{equation}
If $\HH = \HH'$, we simply write $\| A \|_{\HH} = \| A \|_{\HH \rightarrow \HH}$.
\item If $\psi_\epsi \in \HH$ -- resp. $A_\epsi : \HH \rightarrow \HH$ is a linear operator --  and $f : \R \setminus \{0\} \rightarrow \R$, we write $\psi_\epsi = O_{\HH}\big(f(\epsi)\big)$ -- resp. $A_\epsi = \OO_{\HH \rightarrow \HH'}\big(f(\epsi)\big)$ -- when there exists $C > 0$ such that $|\psi_\epsi|_\HH \leq Cf(\epsi)$  -- resp. $\|A_\epsi\|_{\HH \rightarrow \HH'} \leq C f(\epsi)$ -- for $\epsi \in (0,1]$. If $\HH = \HH'$, we simply write $A_\epsi = \OO_\HH\big(f(\epsi)\big)$.
\end{itemize}

\noindent{\bf Acknowledgments.} This research was supported in part by 
National Science  Foundation grants DMS-1800086 (AD),  DMS-1412560, DMS-1620418 (MIW) and Simons Foundation Math + X Investigator Award \#376319 (MIW and AD). AD thanks MSRI for its hospitality and support (through DMS-1440140) during Fall, 2019.  The authors would like to thank Charles Fefferman, James Lee-Thorp, Mikael Rechtsman, Xu Yang and Yi Zhu for many stimulating discussions. 
We thank Yi Zhu and Pipi Hu for providing corroborating numerical simulations based on 
techniques developed in \cite{GYZ:19}.

\section{Preliminaries}\label{prelim}
In this section we review Floquet--Bloch theory for periodic operators. Basic references are \cite[Chapter XIII]{RS4} and \cite{Kuchment:16}.
We discuss decompositions of $L^2=L^2(\R^2)$ and $\Ll^2_\kpar=\Ll^2_\kpar(\R^2/\Z\tv_1)$:
\begin{equations} 
L^2 =  \int^\oplus_{\bk\in\brill} L^2_\bk(\R^2/\Lambda)  \ d\bk
\ \ \ \text{and} \ \ \ 
  \Ll^2_\kpar = \int^\oplus_{[-\pi,\pi]} L^2_{\kpar\ktilde_1+t \ktilde_2}(\R^2/\Lambda) \ dt.
\end{equations} 
These spaces arise in the study of  Hamiltonians  which are  invariant with respect to translations in $\Lambda=\Z\tv_1\oplus\Z\tv_2$ and $\Z\tv_1$, respectively.

\subsection{General Floquet--Bloch theory of operators acting in $L^2$}\label{Floquet}

Let $\Lambda$ be the equilateral lattice $\Z\bv_1\oplus\Z\bv_2\subset \R^2$ with fundamental cell $\Omega$. Let $\Lambda^*= 2\pi\Z\bk_1\oplus 2\pi\Z\bk_2\subset \R^2$ be the lattice dual to $\Lambda$, and $\mathcal{B}$ be the Brillouin zone of $\R^2/\Lambda^*$; see \S\ref{notation}.

We introduce the space $L^2_\bk = L^2_\bk \big(\R^2/\Lambda\big)$ of $\bk$-pseudoperiodic functions 
\begin{equation}\label{L2k}
L^2_\bk \de \left\{ F \in L^2_{\rm loc}(\R^2) : \ F(\bx+\bv) = e^{i\bk\cdot\bv} F(\bx), \ \bv\in\Lambda \right\} 
\end{equation}
equipped with the Hermitian inner product and norm:
\begin{equation}
 \left\langle F,G\right\rangle_{L^2_\bk} \de \int_\Omega \ove{F(\bx)}G(\bx)d\bx,\ \ \ \   
| F |_{L^2_\bk}^2 \de \int_\Omega |F(\bx)|^2 d\bx.
\end{equation}

Given a complex-valued function $f \in C_0^\infty(\R^2)$,  we associate the Gelfand--Bloch transform, the function of $(\bk,\bx)$ defined by:
\begin{equation}\label{ftil}
\widetilde{f}(\bk,\bx) \de  e^{i\bk\cdot \bx}\ \sum_{\bk^\prime \in \Lambda^*} \widehat{f}(\bk-\bk^\prime) e^{- i \bk^\prime\cdot \bx }, \ \ \ \ \widehat{f}(\xi) \de  \int_{\R^2} e^{-i\xi\cdot\by} f(\by)  d\by.
\end{equation}
For each $\bk\in\brill$, $\widetilde{f}(\bk,\bx)\in L^2_\bk$, and by the Fourier inversion formula, $f$ has a $L^2_\bk$-decomposition:
\begin{equation}\label{fL2k}
f(\bx) = \frac{1}{(2\pi)^2}\int_{\mathcal B} \widetilde{f}(\bk,\bx)  d\bk .
\end{equation}
The factor $(2\pi)^2=|2\pi\bk_1\wedge 2\pi\bk_2|$ is equal to the area of $\mathcal{B}$. 
The decomposition \eqref{fL2k} extends by density to all $f \in L^2$ using the Plancherel-like identity:
\begin{equation}
|f|_{L^2}^2= \dfrac{1}{(2\pi)^2} \int_\brill \big|\widetilde{f}(\cdot,\bk)\big|_{L^2_\bk}^2 \ d\bk.
\label{isometry}
\end{equation}

Let $T$ be an operator which is $\Lambda$-periodic, i.e. $(Tf)(\cdot+\bv)=T \big(f(\cdot+\bv)\big)$ for $\bv\in\Lambda$ and  $f$ in the domain of $T$.  Then $T$ maps $L^2_\bk$ to itself and we denote the resulting operator by $T_\bk$. We can study the action of $T$ on $L^2$ fiberwise. We write:
\begin{equation}\label{eq:1y}
 T =  \frac{1}{(2\pi)^2}\int_{\mathcal{B}}^\oplus T_\bk\   d\bk,\qquad\textrm{meaning}\quad  Tf(\bx) = \frac{1}{(2\pi)^2}\int_{\mathcal{B}} \big(T_\bk \widetilde{f}(\bk,\cdot)\big)(\bx)  \ d\bk.
\end{equation} 
Thanks to \eqref{isometry} and \eqref{eq:1y}, we can estimate the operator norm of $T$ from those of $T_\bk$:
\begin{equations}\label{eq:2j}
|Tf|_{L^2}^2  = \dfrac{1}{(2\pi)^2} \int_\brill \big| T_\bk\widetilde{f}(\cdot,\bk) \big|_{L^2_\bk}^2\ d\bk
 \\
 \leq \dfrac{1}{(2\pi)^2} \int_\brill \|T_\bk\|_{L^2_\bk}^2 \cdot |\widetilde{f}(\cdot,\bk)|_{L^2_\bk}^2\ d\bk \leq  \sup_{\bk\in\brill}\|T_\bk\|_{L^2_\bk}^2\cdot |f|_{L^2}^2.
\end{equations}
 By \eqref{eq:1y} the kernel $T(\bx,\by)$ of $T$ is  expressed in  terms of the kernels $T_\bk(\bx,\by)$ of $T_\bk$:
\begin{equation}\label{eq:4f}
T(\bx,\by) = \dfrac{1}{(2\pi)^2} \int_\mathcal{B} T_\bk(\bx,\by) d\bk.
\end{equation}

\subsection{Floquet--Bloch theory of operators acting on $\Ll^2_\kpar$}\label{FB-cyl}

Let $\tv_1, \tv_2$ be any vectors in $\Lambda$ such that $\Lambda=\Z\tv_1\oplus\Z\tv_2$ and define the cylinder $\Sigma=\R^2/\Z\tv_1$. For $\kpar \in [-\pi,\pi]$ we
 introduce the space
\begin{equation}\label{eq:9e}
\Ll^2_\kpar =  \Ll^2_\kpar(\Sigma) \de \left\{ f\in L^2_{\rm loc}(\R^2) :\ e^{-i\kpar\ktilde_1\cdot\bx } f(\bx)\in L^2\left(\Sigma\right) \right\};
\end{equation}
see also \eqref{eq:1b}. 
The Sobolev spaces $\Hh_\kpar^s$ are analogously defined.

In analogy with \eqref{fL2k}, functions in $\Ll^2_\kpar$ have a representation in terms of elements of $L^2_\bk$, where $\bk$ varies over a one dimensional Brillouin zone, the quasimomentum segment:  $\kpar \tk_1 + t \tk_2$, $|t|\le\pi$:
\begin{equation}
f(\bx) = \dfrac{1}{2\pi}\int_{-\pi}^\pi \widetilde{f}(\kpar\ktilde_1+t\ktilde_2,\bx) dt.
\label{slice-rep}
\end{equation}

The decomposition \eqref{slice-rep} implies relations analogous to \eqref{eq:1y}, \eqref{eq:2j} and \eqref{eq:4f}. If $T$ acting on $\Ll^2_\kpar$ commutes with $\Z \tv_1$-translations, then $T$ acts on $L^2_{\kpar \tk_1 + t\tk_2}$ for every $t \in \R$. The resulting operator is denoted $T_{t}$ and we have
\begin{equation}\label{TSig}
T =  \frac{1}{2\pi}\int_{[-\pi,\pi]}^\oplus T_t \ dt
\ \ \text{ meaning } \ \ 
Tf(\bx) = \frac{1}{2\pi}\int_{-\pi}^\pi \big(T_t \widetilde{f}(\kpar\ktilde_1+t\ktilde_2,\cdot)\big)(\bx)  dt.
\end{equation} 
The analog of \eqref{eq:2j} is
\begin{equation}\label{eq:2p}
\|Tf\|_{\Ll_\kpar^2}^2 \le \sup_{|t|\le\pi} \left\| T_t \right\|_{L^2_{\kpar\ktilde_1+t\ktilde_2}}^2\cdot |f|_{\Ll^2_\kpar}^2,
\end{equation}
and the analog of the representation \eqref{eq:4f} is
\begin{equation}\label{eq:9i}
T(\bx,\by) = \dfrac{1}{(2\pi)^2} \int_{-\pi}^\pi T_{t}(\bx,\by)  dt.
\end{equation}
The decomposition \eqref{TSig} allows us to build up  the $\Ll^2_\kpar$-spectrum from the collection of fiber spectra:
\begin{equation}\label{eq:2k}
\sigma_{\Ll^2_\kpar}(T)  = 
\bigcup_{t\in [-\pi,\pi]} \left\{\sigma_{L^2_\bk}(T_\bk) : \ \bk = \kpar \tk_1 + t \tk_2 \right\}\ .
 \end{equation}

\section{Honeycomb medium with line defect}\label{Hdelta}

Let $\Lambda = \Z \bv_1 \oplus \Z\bv_2$ be the equilateral lattice, with area-normalized fundamental cell;
 see \S\ref{notation}.
Let $\tv_1=a_1\bv_1+b_1\bv_2\in\Lambda$, where $a_1$ and $b_1$ are relatively prime, be the direction of the edge. 
There exist $a_2, b_2$ such that $a_1b_2-a_2b_1=1$ and we set $\tv_2=a_2\bv_1+b_2\bv_2$; this implies  $\Lambda=\Z\tv_1\oplus\Z\tv_2$. The dual lattice, $\Lambda^*$ is given by $\Lambda^*=\Z\ktilde_1\oplus\Z\ktilde_2$, 
 where $\ktilde_1=b_2\bk_1-a_2\bk_2$ and $\ktilde_2=-b_1\bk_1+a_1\bk_2$. See Figure \ref{fig:3}, showing the zigzag-edge direction $\tv_1 = \bv_1$ .

We next discuss the honeycomb Schr\"odinger operator $H^0$ and then build up our edge (line-defect)  Hamiltonians $H_\edge^\delta$ in several steps.  We adopt the framework of \cite{FLW-2d_edge:16}.

\subsection{Honeycomb Schr\"odinger operators and Dirac points}\label{HLP}  
Let $H^0=-\Delta+V$, where 
$V$ is a honeycomb lattice potential as defined in  \S\ref{bulk}, {\it i.e.} $V$ is $C^\infty(\R^2)$, real-valued, $\Lambda$-periodic, and invariant under $\RRR$ ($2\pi/3$-rotation) and $\PP$ (spatial inversion). 
Recall that the vertices of the Brillouin zone, $\brill$, are  high-symmetry quasimomenta  generated via $2\pi/3$ rotation of $\bK$ and $\bKp$; see \eqref{K-hs}.

We say that $-\Delta+V$ has a Dirac point at $(E_D,\bk_D)$ if (a) $E_D$ is a $L^2_{\bk_D}$-eigenvalue of $-\Delta+V$ of multiplicity $2$; and (b) there exists $b_*\ge1$  such that the dispersion surfaces $\bk\mapsto E_{b_*}(\bk), \ E_{b_*+1}(\bk)$  touch in isotropic cones: for some $\vF > 0$,
\begin{equations}
\label{cones1}
\begin{split}
E_{b_\star+1}(\bk)  =  E_D + \vF\ | \bk-\bk_D | \cdot  \big( 1\ +\ o(| \bk-\bk_D | ) \big), \\
E_{b_\star}(\bk)  =  E_D - \vF \ | \bk-\bk_D | \cdot  \big( 1\ +\ o(| \bk-\bk_D | ) \big).
\end{split}
\end{equations}

\begin{theorem}\label{FW12}\cite{FW:12,FLW-MAMS:17,LWZ:18}
For a generic choice\footnote{Generic has the following precise meaning. Assume $
\int_{\R^2/\Lambda}e^{-i(\bk_1+\bk_2)\cdot\bx}V(\bx)d\bx\ne0$. Then, for all real $\eps$, except possibly for a discrete set which includes $\eps=0$, the operators $H^{(\eps)}=-\Delta+\eps V$ have Dirac points.} of honeycomb potential $V$, the operator $-\Delta+V$ has Dirac points at $(\bK,E_D)$ and $(\bKp,E_D)$.
\end{theorem}

We elaborate further in order to derive some useful relations. The kernel of $H^0_\bK-E_D$ is spanned by two orthonormalized Bloch eigenmodes $\Phi_1^\bK, \ \Phi_2^\bK \in L^2_\bK$, which satisfy
\begin{equation}\label{eq:1d}
\RRR \Phi_1^\bK = \tau \ \Phi_1^\bK,\ \ \ \CC\PP\Phi_1^\bK = \Phi_2^\bK, \ \ \ \ \RRR \Phi_2^\bK = \ove{\tau} \ \Phi_2^\bK\quad (\tau=e^{2\pi i/3}).
\end{equation}
Thus, $\Phi_1\in L^2_{\bK,\tau}$ and $\Phi_1\in L^2_{\bK,\bar\tau}$.\ 
We set $\Phi_1^\bKp = \PPP \Phi_1^\bK$ and $\Phi_2^\bKp = \PPP \Phi_2^\bK$. Since $\PP H^0_\bK = H^0_\bKp \PP$, the pair $\{\Phi_1^\bKp, \Phi_2^\bKp\}$ is a basis of $\ker\left(H^0_\bKp-E_D\right)$. Furthermore, by \eqref{eq:1d}:
\begin{equation}\label{eq:1g}
\RRR \Phi_1^\bKp = \tau \ \Phi_1^\bKp,\ \ \ \CC\PP\Phi_1^\bKp = \Phi_2^\bKp, \ \ \ \ \RRR \Phi_2^\bKp = \ove{\tau} \ \Phi_2^\bKp;
\end{equation}
$\Phi_1^\bKp\in L^2_{\bKp,\tau}$ and $\Phi_2^\bKp\in L^2_{\bKp,\bar\tau}$.
Furthermore, we may take $\Phi^{\bKs+\bq}=\Phi^\bKs$ for all $\bq\in\Lambda^*$ for $\bKs = \bK, \bKp$.

Using properties of $\RRR$ acting on the vectors $\Phi_j^\bKs$ ($j=1,2$, $\bKs=\bK, \bKp$),  we have
\begin{align}
\left\langle \Phi_1^\bKs,\nabla\Phi_1^\bKs\right\rangle\ &=\ 
\left\langle \Phi_2^\bKs,\nabla\Phi_2^\bKs\right\rangle\ 
=\ \begin{pmatrix} 0\\ 0\end{pmatrix},\quad a=1,2 \label{1nab1}
\end{align}
and that $\left\langle \Phi_1^\bK,-2i\nabla\Phi_2^\bK\right\rangle$ 
 is proportional to $(1, i)^\top$, an eigenvector of the $2\pi/3$ rotation matrix, $R$.
 Using that any scalar multiple of $\Phi_j^\bK$ and $\Phi_j^\bKp$ satisfies \eqref{eq:1d}, respectively \eqref{eq:1g},
we may replace $\Phi^\bKs_1$ by $e^{i\varphi}\Phi^\bKs_1$.
For an appropriate choice of phase, $\varphi$, 
we may take the proportionality factor (the {\it Fermi velocity}) to be strictly positive. Hence, 
\begin{equation}
\left\langle \Phi_1^\bK,-2i\nabla\Phi_2^\bK\right\rangle\ =\ \vF\ \begin{pmatrix} 1\\ i\end{pmatrix}\quad
 \textrm{and}\quad 
 \left\langle \Phi_1^\bKp,-2i\nabla\Phi_2^\bKp\right\rangle\ =\ -\vF\ \begin{pmatrix} 1\\ i\end{pmatrix},\  \vF>0\ .
\label{vF}\end{equation}

Going forward we make the following hypothesis on the bulk potential $V$:
\begin{equation}
{\bf (H1)}\ \ H^0\  \textrm{is a generic honeycomb operator so that Theorem \ref{FW12} applies with \eqref{vF}}.
\label{H1}\end{equation}

\subsection{The spectral no-fold condition; $H^0$ models a {\it semi-metal}
 at energy $E_D$.}\label{nf}

A second key hypothesis on $H^0 = -\Delta+V$ states that if $E_D$ is the energy of a Dirac point,
then  the two touching dispersion surfaces at energy $E_D$ do so only at  high symmetry quasimomenta: 

 \nit {\bf (H2)\ The spectral no-fold condition}
 
\nit  Let $\bKs=\bK$ or $\bK_\star=\bKp$. Then,
  for $b=b_*$ and $b_{*}+1$, 
\begin{equation}
{\rm if}\  E_b(\bKs+t\ktilde_2)\ =\ E_D,\ \ {\rm then}\ \ \bKs+t\ktilde_2\in \bK+\Lambda^*\ \  {\rm or}\ \  \bKs+t\ktilde_2\in\bKp+\Lambda^*.
%\systeme{E_j(\bk) = E_D, \\ \bk \in \bKs + \R \ktilde_2} \ \Rightarrow \ j \in \{b_\star, b_\star+1\} \text{ and } %\bk \in \{\bK,\bKp\}.
\label{H2}\end{equation}

An assumption similar to (H2) was introduced in \cite{FLW-2d_edge:16} and is implicit in the physics literature \cite{HR:07,RH:08}. The formulation in \cite{FLW-2d_edge:16} is tailored for zigzag-type rational edges and the generalization (H2) allows for armchair type edges as well; all arbitrary rational edges are now covered. 

Since the density of electronic states is zero at energy $E_D$, graphene is
often called a {\it semi-metal}. The condition (H2) holds in the (explicitly solvable) tight binding model of graphene \cite{RMP-Graphene:09} and by  \cite[Corollary 6.4]{FLW-CPAM:17} the condition (H2) holds in the strong binding regime.
Honeycomb structures that satisfy (H1) but fail to satisfy (H2)  can be thought as {\it metallic} at energy $E_D$. 
 An example where the no-fold condition fails is the case where $V$ is a small amplitude honeycomb potential 
and $\tv_1=\bv_1+\bv_2$, the armchair edge direction; see \cite[\S8]{FLW-2d_edge:16}.

In the following section we shall introduce edge perturbations of $H^0$, which destabilize Dirac points.
 Thus,  for $\bk$
near $\bKs$, there is a spectral gap about energy $E_D$. If (H2) fails, then we expect a wave-packet comprised of Floquet-Bloch modes near $\bKs$ to resonate with Floquet-Bloch modes of energies near $E_D$, but with quasi-momenta
 away from $\bKs$. Such wave-packets are conjectured to slowly radiate their energy into the bulk \cite[\S1.4]{FLW-2d_edge:16}.

\subsection{ Deformed bulk Hamiltonians}\label{dfrm-H0}
Henceforth, we assume that $H^0=-\Delta+V$ satisfies hypotheses (H1) and (H2) of \S\ref{HLP} and \ref{nf}.

In this section we define $\Lambda$-periodic perturbed bulk operators, perturbations of $H^0$: $H^\delta_{{\rm bulk},\pm}$: $H_\pm^\delta$ and  $\tHd_\pm$. Far from the line-defect, $H_\pm^\delta$ corresponds to $\mathcal{P}$-breaking and $\tHd_\pm$ to  $\mathcal{C}$-breaking deformations of $H^0$.

\subsubsection{$\PPP$-breaking and $\CCC$-invariant model}\label{Pbreak}

Let $W \in C^\infty(\R^2)$ be real-valued, $\Lambda$-periodic and odd. We set
\begin{equation}
H^\delta_\pm\ \de\ H^0\pm\delta W(\bx) = -\Delta + V(\bx) \pm \delta W(\bx).
\label{Hpm-bulk}
\end{equation}
We observe that $[\CCC,H^\delta_\pm] = 0$ while $[\PPP, H^\delta_\pm] \neq 0$. Since $\PP H^\delta_+=H^\delta_-\PP$, the operators $H^\delta_+$ and  $H^\delta_-$ have the same spectrum.

In addition to (H1) and (H2) we make the following assumption on $W$:
\begin{enumerate}
\item[{\bf (H3)}] {\bf Non-degeneracy:}  $W$ is smooth, real-valued, odd, $\Lambda$-periodic  and 
\begin{equation}
\var^\bKs =\blr{\Phi_1^\bKs, W(\bx) \Phi_1^\bKs} \neq 0,\ \ \bKs\ =\ \bK, \bKp.
\label{H3}
\end{equation}
Note that since $W(-\bx)=-W(\bx)$, we have $\var^\bKp=-\var^\bK$. Furthermore, $\vartheta^\bKs$ is invariant
under $\Phi_1\mapsto e^{i\theta}\Phi_1$, $\theta\in\R$.
\end{enumerate}

Conditions (H1),  (H2) and (H3) imply that $H^\delta$ has a gap in its  $\Ll^2_{\bKs \cdot \tv_1}$- essential spectrum,  centered at $E_D$ for $\bKs \in \{\bK, \bKp\}$. 
Indeed $H^\delta_+$ and $H^\delta_-$ have a common gap in their 
essential spectra;  
\begin{equation}\label{gap-d}
\left(\textrm{$\Ll^2_{\bKs \cdot \tv_1}$ essential spectrum of $H^\delta_\pm$}\right)  \cap \GG_\delta\ =\ \emptyset,\quad \GG_\delta \de (a_\delta,b_\delta),
\end{equation}
where, 
\begin{equation}
a_\delta =  E_D-\delta \cdot \var_\gap + O(\delta^2),\quad b_\delta =  E_D-\delta \cdot \var_\gap + O(\delta^2),\ \
\var_\gap = \big|\var^\bK\big|.
\label{C1gap-d}
\end{equation}
 This property is implicitly contained in \cite[Section 7.1]{FLW-2d_edge:16}; see \cite[Lemma 4.1 and 4.3]{Drouot:18b}, which imply \eqref{gap-d} via \eqref{eq:2k}.

\subsubsection{$\CCC$--breaking and $\PPP$--invariant model} Let $a \in C^\infty\big(\R^2)$ be real-valued, even and $\Lambda$-periodic. We set
\begin{equation}
\tHd_\pm = - \Delta + V(\bx) \pm  \delta \cdot \dive\big(  A(\bx)  \cdot \nabla \big), \ \ \ \ A(\bx) = i\matrice{0 & -a(\bx) \\ a(\bx) & 0 }\ =\ ia(\bx)\sigma_2.
\end{equation}
We have  $[\PPP,\tHd_\pm] = 0$ while $[\CCC, \tHd] \neq 0$. Furthermore, since $\CCC \tHd_+ = \tHd_- \CCC$,  the operators $\tHd_\pm$ have the same spectrum. In analogy with our discussion for $H^\delta_\pm$, if $H^0$ satisfies (H1), (H2) and the non-degeneracy assumption
\begin{equation}
{\bf (\widetilde{H3})}\qquad   \tvar^\bK \de \blr{\Phi_1^\bK, \dive\big(  A(\bx)  \cdot \nabla \Phi_1^\bK \big)} \neq 0
\end{equation}
then the operators $\tHd_+$ and $\tHd_-$ have a (common)  $\Ll^2_{\bKs \cdot \tv_1}$ gap in their essential spectra;
\begin{equation}\label{tgap-d}
\left(\textrm{$\Ll^2_{\bKs \cdot \tv_1}$ essential spectrum of $\tHd_\pm$}\right)  \cap \tilde{\GG}_\delta\ =\ \emptyset,\quad \tilde{\GG}_\delta \equiv (\ta_\delta,\bt_\delta)\ ,\ \textrm{where}
\end{equation}
\begin{equation}\label{Cgap-d}
\ta_\delta = E_D-\delta \cdot \tvar_\gap + O(\delta^2), \ \ \bt_\delta = E_D + \delta \cdot \tvar_\gap + O(\delta^2), \ \ \ \ \tvar_\gap \de |\tvar^\bK\big|.
 \end{equation}

\subsection{Edge operators and the edge state eigenvalue problem }\label{edge-op}

A {\it domain-wall function} is a real-valued function $\kappa \in C^\infty(\R)$, such that:
\begin{equation}
\lim_{s \rightarrow \pm \infty} \kappa(s) = \pm 1 \ \ \ {\rm and} \ \ \ \kappa^\prime \in L^\infty.
\label{dw-fn}
\end{equation}

We introduce and study Hamiltonians $H^\delta_{\edge}$ which slowly interpolate transversely to the rational line-defect $\R\tv_1$,  between $H^\delta_\pm$ as $\ktilde_2\cdot\bx \rightarrow \pm \infty$. This is realized  via a scaled domain wall function.  
We define the edge Hamiltonian, $H^\delta_{\edge}$,  with $\PPP$--breaking and $\CCC$--invariant bulk as
 \begin{equation}
  H^\delta \ =\   -\Delta  +  V(\bx)  +  \delta \cdot \kappa(\delta\ktilde_2\cdot \bx)W(\bx),\label{H_delta}
  \end{equation} 
and define the edge Hamiltonian, $H^\delta_{\edge}$, with $\CCC$--breaking and $\PPP$--invariant bulk as 
\begin{equation}
\tHd \ =\  - \Delta + V(\bx) + \delta \cdot \dive\big( \kappa(\delta \tk_2 \cdot \bx) A(\bx) \nabla \big)\label{tH_delta}
\end{equation}
Since $\tk_2 \cdot \tv_1 = 0$, the operators $H^\delta_{\edge}$ commute with translations in  $\Z\tv_1$.

Hamiltonians of the type $H^\delta$ were introduced in \cite{FLW-PNAS:14,FLW-MAMS:17,FLW-2d_edge:16}.  Those of the type $\tHd$ are closely related to photonic settings studied in \cite{HR:07,RH:08,Shvets-PTI:13,LWZ:18}. In \cite{LWZ:18} a broader class of anisotropic honeycomb photonic (unperturbed) media is introduced; the analysis of the present paper extends to such operators and their perturbations.
 Both $H^\delta$ and $\tHd$ incorporate the essential features of key physical models \cite{HR:07,RH:08}. A magnetic perturbation (breaking $\CCC$-invariance) of $H^0$ was investigated in \cite{Drouot:18b,Drouot:19}. 
 
 In the remainder of the paper, we focus on edge states with energy near $E_D$. These are the  $\Ll^2_\kpar$-eigenvalues of $H_{\edge}^\delta$:
 \begin{equation}
\label{kp-evpp} 
H_{\edge}^\delta\Psi\ =\ E\  \Psi,\ \ \ \ \Psi \in \Hh^2_\kpar, \ \ \ \  \kpar \in [-\pi,\pi), \ \ \ \ E \in \GG^\delta\quad (\textrm{resp.}\ \tilde{\GG}^\delta\ ). 
\end{equation}

%
%
%\footnote{\tr{I removed \S3.5 for mathematical reasons:}
%\begin{itemize}
%\item (3.20) in kkp20 \tr{was a decomposition of the resolvent of $H^\delta|_{\Ll^2_\kpar}$ as a fiberwise integral of resolvents of  $H^\delta|_{\l^2_\bK+t\tk_2}$. though $H^\delta$ does not act on those spaces.}
%\item In kkp20, "For $z=O(1)$, sufficiently small and not eigenvalue of $H^\delta\big|_{\Ll^2_\kpar}$, $E_D+\delta z$ is $O(\delta)$-distant from the spectrum" \tr{Even if you correct to "fix $z$; if for all $\delta$ sufficiently small, $E_D + \delta z$ is not an eigenvalue of $H^\delta|_{\Ll^2_\bk}$, then $E_D + \delta z$ is $\delta$-distant
%from the spectrum", I do not see a basis for this observation. It needs that the spectrum of $H^\delta|_{\Ll^2_\bk}$ is made of points of the form $E_D + \delta \var + o(\delta)$, where $\var$ belongs to a discrete set. Of course our theorem proves this but we do not know it yet. In other scalings (e.g. semiclassical) eigenvalues can be much closer than $\delta$.}
%\end{itemize}
%}
%
%
%\footnote{\tr{I moved the section classifying edges to \S4.4. I simplified it:
%\begin{itemize}
%\item[1] I thought we should state the results before trying to explain the strategy. At this point we are lacking most of the objects -- Dirac ops etc. 
%\item[2] I incorporated the part describing AC versus ZZ to \S4.4. It is a more natural place as we already know that  Dirac points contribute to the most mass of approx. edge states. Using \eqref{slice-rep} alone allows to distinguish AC versus ZZ. The resolvent is not needed.
%\end{itemize}}
%}
%
%

\section{Multiscale analysis and effective Dirac operators}\label{eff-Dirac}
%\label{sec:4}  

We review the multi-scale construction of approximate edge states \cite{FLW-PNAS:14,FLW-2d_edge:16,FLW-2d_materials:15,FLW-MAMS:17}. 
Effective Dirac operators  emerge as determining the transverse localization of these states. Since edge states eigenvalues are poles of a resolvent they therefore also emerge in the leading order term in the resolvent expansion of $H^\delta_\edge$ near energy $E_D$; see Theorem \ref{res-exp}.  Before embarking on this expansion, we make a convenient choice of basis for the eigenspace associate with the Dirac points $(\bK,E_D)$ and  $(\bKp,E_D)$. 
Let $\bl$ be such that $\bl\cdot\tv_1=1$ and $\tk_2\cdot\bl=0$, thus 
\begin{equation}
\bl\ =\ \ktilde_1\ -\ \left[\ \frac{\ktilde_2}{|\ktilde_2|}\cdot\ktilde_1\right]\ \frac{\ktilde_2}{|\ktilde_2|}.
\label{bl-def}
\end{equation}

\begin{proposition}\label{ips}  Let $\Phi_j^\bKs$, $j=1,2$,  with $\bKs = \bK$ or $\bKp$, be given as required in hypothesis (H1); see Theorem \ref{FW12} and \eqref{vF}.  \\
(a) There exists $\theta\in\R$ such that if we redefine  $\Phi_1^\bKs$ as $e^{i\theta}\Phi_1^\bKs$, and $\Phi_2^\bKs$ as $\PPP\CCC [e^{i\theta}\Phi_1^\bKs]$, then 
we obtain an orthonormal basis for $L^2_\bKs-{\rm kernel}(H^0-E_D)$, $\{\Phi_1^\bKs,\Phi_2^\bKs\}$ with $\Phi_1^\bKs\in L^2_{\bKs,\tau}$, $\Phi_2^\bKs\in L^2_{\bKs,\bar\tau}$ and such that for all  $\br \in \R^2$:
\begin{equation}\label{lip-ip} \ 
   \blr{ (\Phi^\bKs)^\top, -2i\br\cdot\nabla\Phi^\bKs  }
    \ =\  \frac{\vF^\bKs}{|\tk_2|}\ \Big(\ \left(\ktilde_2\cdot\br\right)\  \sigma_1 + \det[\br, \ktilde_2] \  \sigma_2\Big)\ .
\end{equation} 
Here,  $\vF^\bK=\vF$ and $\vF^\bKp=-\vF$. 

\nit (b) Furthermore, if  $\bl$ is defined as in \eqref{bl-def},  then 
\begin{equation}\label{lip-ipA}
   \blr{ (\Phi^\bKs)^\top, -2i\bl\cdot\nabla\Phi^\bKs  }
    \ =\  \frac{\vF^\bKs}{|\tk_2|}\   \sigma_2,
\end{equation} 
\end{proposition}

\begin{proposition}\label{thetaKs} Let $W$ be as in $(H3)$ and $A$ be as in $(\widetilde{H3})$
 below.  For $\bKs = \bK,\ \bKp$:
\begin{align}
 &\blr{(\Phi^\bKs)^\top, W \Phi^\bKs} 
  = \var^\bKs \sigma_3, \quad \textrm{where}\quad  \var^\bK\ =\ -\var^\bKp. 
  \label{varKs}\\
&\left\langle (\Phi^\bKs)^\top, \dive\left( A \cdot \nabla \Phi^\bKs \right)\right\rangle = \tvar^\bKs \sigma_3,
 \quad \textrm{where}\quad \tvar^\bK\ =\ \tvar^\bKp.
 \label{tvar-def}
\end{align}
\end{proposition}

 \nit We remark first on the proof of Proposition \ref{thetaKs}, and then give the proof of Proposition 
 \ref{ips}.
 \medskip
 
 \nit{\it Proof of Proposition \ref{thetaKs}:} The relation \eqref{varKs} follows from $W$ being odd.
Relation \eqref{tvar-def} is proved in  \cite[Proposition 5.1 and Section 7.1]{LWZ:18}. 
\medskip

\nit{\it Proof of Proposition \ref{ips}:} We begin by proving \eqref{lip-ip} for the case $\br=\tk_2$. Start with the basis $\{\Phi_1^\bKs,\Phi_2^\bKs\}$  of
 $\textrm{kernel}(H^0-E_D)$ for which \eqref{vF} and \eqref{1nab1} hold.  Let $\ktilde_2=(\ktilde_2^{(1)},\ktilde_2^{(2)})$. 
  First consider $\bKs=\bK$. By \eqref{1nab1}, the diagonal entries of \eqref{lip-ip} vanish.
 Concerning the off-diagonal elements, we note first that
 $\blr{\Phi^\bK_2, -2i\ktilde_2\cdot\nabla \Phi^\bK_1}= \overline{\blr{\Phi^\bK_1, -2i\ktilde_2\cdot\nabla \Phi^\bK_2}}$; the matrix is Hermitian.  By \eqref{vF}, 
 $\blr{\Phi^\bK_2, -2i\ktilde_2\cdot\nabla \Phi^\bK_1}= \vF\ (\ktilde_2^{(1)}+i\ktilde_2^{(2)})$. Define 
  $\hPhi_1=\omega\Phi_1$ and $\hPhi_2=\overline\omega\Phi_2$, where $|\omega|=1$ and is to be determined.  We have
  $ \blr{\hPhi^\bK_2, -2i\ktilde_2\cdot\nabla \hPhi^\bK_1}= \vF\ \overline{\omega}^2\ (\ktilde_2^{(1)}+i\ktilde_2^{(2)})$. Choose $\omega$ so that $\overline{\omega}^2=(\ktilde_2^{(1)}-i\ktilde_2^{(2)})/|\ktilde_2|$.  Then,
 $ \blr{\hPhi^\bK_2, -2i\ktilde_2\cdot\nabla \hPhi^\bK_1}= \vF\ |\ktilde_2|$.
This implies \eqref{lip-ip} for the special case $\tk=\tk_2$ and $\bKs=\bK$. The relation for $\tk=\tk_2$ and $\bKs=\bKp$ follows by the same argument and using the second relation in \eqref{vF}.   

   To prove \eqref{lip-ip} for general $\br\in\R^2$, we
   first fix $\bKs=\bK$ and let $\omega$ and $\hPhi_j$ be chosen as above. Again the diagonal elements vanish and the matrix is Hermitian. Furthermore 
$   \left\langle \hPhi_1^\bK, -2i\br\cdot\hPhi_2^\bK\right\rangle\ =\ \overline{\omega}^2\ 
    \left\langle \Phi_1^\bK, -2i\br\cdot\Phi_2^\bK\right\rangle
$ and by \eqref{vF} we have 
$
\left\langle \hPhi_1^\bK, -2i\br\cdot\hPhi_2^\bK\right\rangle\ = \vF\ \overline{\omega}^2\ (r^{(1)}+ir^{(2)})=\ \frac{\vF}{|\tk_2|}\ (\ \ktilde_2\cdot\br\ +\ i\ \det[\ktilde_2,\br]\ )$.
Since the (2,1) entry of the matrix is the complex conjugate of this expression, we have proved that the pair $\{\hPhi^\bK_1,\hPhi^\bK_2\}$ satisfies \eqref{lip-ip} for the case $\bKs=\bK$. The case $\bKs=\bKs$ is proved the same way; we need only apply the second 
relation in \eqref{vF}. This completes the proof of \eqref{lip-ip}.

Finally, we verify \eqref{lip-ipA}. Suppose that $\bl\cdot\tv_1=1$ and $\tk_2\cdot\bl=0$. 
 Since $\bl\cdot\tv_1=1$, we have
 $\bl=\ktilde_1+\rho\ktilde_2$, for some $\rho\in\R$ (recall $\tk_2\cdot\tv_1=0$). Furthermore, 
  $\tk_2\cdot\bl=0$ means $\ktilde_2\cdot(\ktilde_1+\rho\ktilde_2)=0$ and hence
   $\rho = -(\ktilde_1\cdot\ktilde_2)/|\ktilde_2|^2$. Therefore, $\bl\ =\ \ktilde_1\ -\ [(\ktilde_2/|\ktilde_2|)\cdot\ktilde_1]\ (\ktilde_2/|\ktilde_2|)$, the projection of $\ktilde_1$ on the orthogonal complement of $\ktilde_2$.
  For this choice of $\bl$, we have $\det[\bl, \tk_2,]=\det[\ktilde_1,\ktilde_2]=1$. By the choice of $\ktilde_1$ and $\ktilde_2$ in Section \ref{Hdelta}. Thus we have 
   $
\left\langle \hPhi_1^\bK, -2i\bl\cdot\hPhi_2^\bK\right\rangle\ =  \frac{\vF}{|\tk_2|}| \sigma_2$;
 the pair $\{\hPhi^\bK_1,\hPhi^\bK_2\}$ satisfies \eqref{lip-ipA}. Finally, we drop the hats  and simply write $\Phi_j^\bKs$ instead of $\hPhi_j^\bKs$ ( $ j=1, 2$ and $\bKs=\bK, \bKp$). 
    This completes the proof of Proposition \ref{ips}.

 \subsection{Multiscale analysis}
 \label{ms-Hd}
  In this section we review the construction of approximate solutions of the eigenvalue problem  \eqref{kp-evpp} with $\kpar = \bKs \cdot \tv_1 + \delta \mu$, for small $\delta$; see \cite{FLW-2d_edge:16,FLW-2d_materials:15,LWZ:18} ($\mu = 0$) and \cite{Drouot:18b} ($\mu \neq 0$). In our discussion, we fix  $H^\delta_\edge = H^\delta$; the procedure is analogous
  for $\tHd$.  We seek a solution of  \eqref{kp-evpp} in the form $\Psi(\bx)  =  e^{i\delta\mu\bl\cdot\bx}\cdot \varphi_0(\bx) \in L^2_{\bKs \cdot \tv_1 + \delta \mu}$. We let  $\bl$ be such that  $\bl\cdot\tv_1=1$  so
   that $\varphi_0$ has fixed psedo-periodicity;  $\varphi_0 \in L^2_{\bKs \cdot \tv_1}$. Below, we shall 
   settle on $\bl$ given by the expression in \eqref{bl-def} as an optimal choice.
 The discussion below will motivate the choice of $\bl$ in the statement of Proposition \ref{ips}.
 
 Substitution into eigenvalue problem \eqref{kp-evpp} yields
\begin{equation}\label{vphi0-eq}
\left( -(\nabla+i\delta\mu\bl)^2 +  V(\bx) +\ \delta \cdot \kappa(\delta \tk_2 \cdot \bx) W(\bx) \right)\varphi_0  =  E \cdot \varphi_0,\quad 
 \varphi_0\in \Ll^2_{\bKs \cdot \tv_1}.
\end{equation}
The form of  \eqref{vphi0-eq} suggest an expansion of its solutions  $\varphi_0(\bx) = \Psi_0(\bx, \delta \tk_2 \cdot \bx)$, depending on, $\bx$, the {\it fast} scale of the periodic structure, and on $s=\delta\ktilde_2\cdot\bx$, the {\it slow} scale  of the domain wall.
The eigenvalue problem \eqref{vphi0-eq} for $\Psi_0(\bx,s)$ is 
 \begin{equations}\label{ms-evp}
 \left( -(\nabla_\bx  +  \delta\ktilde_2\D_s  + i\delta\mu\bl)^2\ +  V(\bx)\ +\ \delta\kappa(s) W(\bx)  \right) \Psi_0(\bx,s) =  E\ \Psi_0(\bx,s)\\
 \qquad \Psi_0(\bx+\tv_1,s) = e^{i\bKs \cdot \tv_1}\ \Psi_0(\bx,s),\qquad \Psi_0(\bx,s)\ \ \to\ 0\quad \textrm{as}\ \ |s|\to\infty.
 \end{equations}
We next expand $(\Psi_0,E)$ in powers of $\delta$: $\Psi_0(\bx,s;\mu)=\Psi_0^{(0)}(\bx,s;\mu)+\delta \cdot \Psi_0^{(1)}(\bx,s;\mu)+\dots$ and 
   $E(\mu)=E_D+\delta\ E_1(\mu)+\dots$, and substitute these expansions into \eqref{ms-evp}. Grouping terms according to their order in $\delta$  yields a hierarchy of equations  for $(\Psi_j(\bx,s,\mu),E_j)$,  which can be solved recursively.    The PDEs in this hierarchy are each  solved  subject to the conditions: 
   \begin{equation} \Psi_0^{(j)}(\bx+\tv_1,s)=e^{i \bKs \cdot \tv_1} \cdot \Psi_0^{(j)}(\bx,s),\quad \Psi_0^{(j)}(\bx,s)\ \ \to\ 0\quad \textrm{as}\ \ |s|\to\infty,\quad j\ge0.
   \label{ms-bc}\end{equation}

\nit  At order $\delta^0=1$, we have $\left(H^0- E_D\right)\Psi^{(0)}_0=0$. We solve it with 
  $\Psi_0^{(0)}(\bx,s)= \Phi^\bKs(\bx)^\top\ \alpha(s)$, where 
 $\alpha(s)=(\alpha_1(s),\alpha_2(s))^\top$ is to be determined such that  $|\alpha(s)|\to0$
 as $|s|\to\infty$.\\
  At order $\delta$, we obtain:
   \begin{align}
  \left(H^0 - E_D\right)\Psi_0^{(1)}(\bx,s) \ &=  2\sum_{j=1}^2\big(\ktilde_2  \D_s\alpha_j(s)  +     i\mu  \bl \alpha_j(s) \big)\cdot\nabla_\bx\Phi_j^\bKs(\bx)\nn\\
    &\quad -\kappa(s)W(\bx) \cdot \sum_{j=1}^2\Phi_j^\bKs(\bx)  \alpha_j(s)  +  E_1 \cdot \sum_{j=1}^2\Phi_j(\bx)  \alpha_j(s).
 \label{delta1}   \end{align}
Equation  \eqref{delta1} has a solution satisfying the pseudo-periodicity condition of \eqref{ms-bc} if and only if the  right hand side
is $L^2_{\bK_\star}$-orthogonal to $\Phi^\bKs_1$ and $\Phi^\bKs_2$. That is, for $m=1,2:$
\begin{align}
\sum_{j=1}^2 \blr{ \Phi_m^\bKs,-2i\ktilde_2 \nabla\Phi^\bKs_j} & \frac{1}{i}\D_s\alpha_j(s)\  +\   \mu\ \sum_{j=1}^2 \blr{ \Phi^\bKs_m,  -2i\bl\nabla\Phi^\bKs_j(\bx) } \alpha_j(s)\nn \\
&    +  \sum_{j=1}^2 \left\langle\Phi^\bKs_m,W\Phi^\bKs_j\right\rangle \ \kappa(s)\alpha_j(s)\ -\ E_1 \alpha_m(s)\ =\ 0.
\label{solvty}\end{align}
The system \eqref{solvty} can be simplified using Propositions \ref{ips} and \ref{thetaKs}.
 The first inner product in \eqref{solvty} is evaluated using \eqref{lip-ip} using $\br=\tk_2$.
To evaluate the second inner product in \eqref{solvty}, we use second part of Propositions \ref{ips} and note (from \eqref{lip-ip}) that we can eliminate the $\sigma_1$ dependence by choosing $\bl$ to satisfy, in addition to $\bl\cdot\tv_1=1$, the condition $\tk_2\cdot\bl=0$;
 see \eqref{lip-ipA}. 
The third inner product  in \eqref{solvty} is evaluated using \eqref{varKs} in Proposition \ref{thetaKs}. We summarize:

\begin{proposition}[$\CCC$ invariant, $\PPP$ breaking case]\label{H-eff}
For $H^\delta_{\rm edge}=H^\delta$, the slowly varying amplitudes $\alpha(s)=(\alpha_1(s),\alpha_2(s))^\top$ are governed by the eigenvalue problem: $ \DiKs(\mu)\ \alpha\ =\ E_1\ \alpha$, 
where $\DiKs(\mu)$  is the effective Dirac operator
\begin{equation}
\DiKs(\mu)\ =\ \vF^\bKs\ |\tk_2|\ \sigma_1\ \frac{1}{i} \frac{\D}{\D s}\ +\ \frac{\vF^\bKs}{|\tk_2|}\ \mu\  \sigma_2\ +\ \var^\bKs\ \sigma_3\ \kappa(s)\ .
\label{D-eff}
\end{equation}
Here, $\vF^\bK=\vF$ and $\vF^\bKp=-\vF$ and $\var^\bKp=-\var^\bK$.
\end{proposition}
%
%
%  \begin{equation}\label{D-eff}
%  \DiKs(\mu)  \de   \dfrac{\vF^\bKs  |\ktilde_2| \cdot \sigma_1  }{i} \frac{\D}{\D s}  +  \vF^\bKs |\bl| \mu \cdot  \sigma_2  
%  +  \var^\bKs \sigma_3 \cdot \kappa(s).
%\end{equation}

 If we now let $(E_1,\az)$ be an eigenpair of $\DiKs(\mu)$, {\it i.e.} $\DiKs(\mu) \az = E_1 \az$, $\alpha\in L^2(\R)$,  then \eqref{delta1} has a solution,  $\Psi_0^{(1)}(\bx,s)$,  with the required pseudoperiodicity. It follows that  $\Phi^\bKs(\bx)^\top\ \alpha(s) + \delta \Psi_0^{(1)}(\bx,s)$ solves \eqref{ms-evp} modulo $O(\delta^2)$.

When considering $\tHd$ instead of $H^\delta$, the perturbation of $H^0=-\Delta+V$ is the operator  $\dive(\kappa(s)  A(\bx) \cdot \nabla)$ instead of $\kappa(s)W(\bx)$. The same analysis as for $H^\delta$, now using the relation \eqref{tvar-def} of Proposition \ref{thetaKs},  yields:
\begin{proposition}[$\CCC$ breaking, $\PPP$ invariant case]\label{H-eff}
For $H^\delta_{\rm edge}=\tHd$, the slowly varying amplitudes $\alpha(s)=(\alpha_1(s),\alpha_2(s))^\top$ are governed by the eigenvalue problem: $ \tDiKs(\mu)\ \alpha\ =\ E_1\ \alpha$, 
where $\tDiKs(\mu)$  is the effective Dirac operator
\begin{equation}
\tDiKs(\mu)\ =\ \vF^\bKs\ |\tk_2|\ \sigma_1\ \frac{1}{i} \frac{\D}{\D s}\ +\ \frac{\vF^\bKs}{|\tk_2|}\ \mu\  \sigma_2\ +\ \tvar^\bKs\ \sigma_3\ \kappa(s)\ .
\label{tD-eff}
\end{equation}
Here, $\vF^\bK=\vF$ and $\vF^\bKp=-\vF$ and 
 $\tvar^\bKp=\tvar^\bK$.
\end{proposition}

%
%shows that the eigenvalue problem for $\tHd$ reduces to an eigenvalue problem for 
%  \begin{equation}
%  \tDiKs(\mu)  \de  \dfrac{\vF^\bKs\ |\ktilde_2| \cdot \sigma_1}{i} \cdot \frac{\D}{\D s}  +  \vF^\bKs |\bl| \mu \cdot \sigma_2  
%  +  \tvar^\bKs  \sigma_3 \cdot \kappa(s).
%  \label{tD-eff}\end{equation}

  \subsection{Spectra of effective Dirac operators}\label{D-eff-spec}
  
The properties of Pauli matrices (see \eqref{eq:1s}) imply the following relations among the Dirac operators and their spectra:
\begin{itemize}
 \item $\DiKp(\mu)=-\DiK(\mu)$ and hence $\sigma_{L^2}\big(\DiKp(\mu)\big)=-\sigma_{L^2}\big(\DiK(\mu)\big)$. 
 \item At $\mu=0$, $\sigma_2\DiKs(0)\sigma_2=-\DiKs(0)$ and hence 
 $\sigma_{L^2}(\DiKs(0))$ is symmetric about zero energy.
 \item $\sigma_3\tDiKp(\mu)\sigma_3= \tDiK(\mu)$ and therefore
  $\sigma_{L^2}\big(\tDiK(\mu)\big)= \sigma_{L^2}\big(\tDiKp(\mu)\big)$.
 \end{itemize}

The following result summarizes the spectral properties of $\DiKs(\mu)$ and $\tDiKs(\mu)$.
Recall first that $\theta_{\rm gap}=|\var^\bK|=|\var^\bKp|$ and  $\widetilde\theta_{\rm gap}=|\tvar^\bK|=|\tvar^\bKp|$ ; see \eqref{C1gap-d} and \eqref{Cgap-d}.
\begin{proposition}\label{cor-D-spec} Let $\bKs = \bK, \bKp$. 
\begin{enumerate}
\item The $L^2(\R)$ spectrum of $\DiKs(0)$ is real and symmetric about zero energy. Its essential spectrum in the set $\R\setminus(-\theta_{\rm gap},\theta_{\rm gap})$.  
The point spectrum of $\DiKs(0)$ contains $z_0=0$. Moreover, 
for some $N\ge0$, the point spectrum of $\DiKs(0)$ consists of $2N+1$ discrete simple eigenvalues, symmetric about zero, in the gap $(-\theta_{\rm gap},\theta_{\rm gap})$:
 \[ -\theta_{\rm gap}<-z_{-N}<\cdots<z_{-1}<z_0=0<z_{1}<\cdots<
 z_{N}<\theta_{\rm gap}.\]
\item For $\mu\in\R$,  $\DiKs(\mu)$ acting on $L^2(\R)$ has essential spectrum given by
\[ \sigma_{\rm ess}(\DiKs(\mu))\ =\ \R\setminus \big( -\theta_\gap(\mu), \theta_\gap(\mu) \big),\ 
\textrm{where}\quad 
\theta_{\rm gap}(\mu)\ =\ 
\sqrt{\theta_{\rm gap}^2 + \frac{\vF^2}{|\tk_2|^2}\ \mu^2 }.\]
Let $N$ be as in part (1). There are $2N+1$ eigenvalues $\var_{-N}^\bKs(\mu) < \dots < \var_N^\bKs(\mu)$ in the gap $\big( -\theta_\gap(\mu), \theta_\gap(\mu) \big)$.  For $j\ne0$ these eigenvalues are given by the expressions:
\begin{equation}
\var_{\pm j}^\bKs(\mu)\ =\  \pm\sqrt{ z_j^2\ +\ \frac{\vF^2}{|\tk_2|^2}\ \mu^2 }, \ \ \ \ 1\le j\le N,\label{var-j}\end{equation}
and for $j=0$ we have 
%\footnote{\tb{need to check the sign here}}
\[
\var_0^\bKs(\mu) = - \mu\ \frac{\vF}{|\tk_2|}\ \rm{\sgn}\left(\var^\bKs\right).
\]
\item The results of parts (1) and (2) apply to the operator $\tDiKs(\mu)$ where
we replace $\var_j^\bKs(\mu)$ by $\tvar_j^\bKs(\mu)$,  $\theta_{\rm gap}$ by $\widetilde\theta_{\rm gap}$
 and $\theta_{\rm gap}(\mu)$ by $\widetilde\theta_{\rm gap}(\mu)$.
\end{enumerate}
\end{proposition}

\nit Part (1) of Proposition \ref{cor-D-spec} was proved in \cite{DFW:18}; 
  Part (2) appears in \cite{Bal:17,Drouot:18b,Faure:19} and is proved using part (1) in \cite[Lemma 3.1]{Drouot:18b}. 
 
 %%%
 \subsection{From approximate edge states to genuine edge states }\label{app2gen}
 In this section we use Theorem \ref{res-exp} together with the construction of approximate edge states to obtain a complete characterization of all point spectrum in the gap about energy $E_D$. 
 
 Consider the Hamiltonian $H^\delta$ and the associated effective Dirac operators 
 $\DiKs(\mu)$,\ $\bKs=\bK, \bKp$. (The discussion of this subsection applies as well to $\tHd$.)
Fix $\mu\in\R$,  a positive integer $N$ as in Proposition \ref{cor-D-spec} and any integer $j$ with $|j|\le N$.  Finally, fix $\bKs=\bK$ or $\bKp$.  Let $\left(\var^\bKs_j(\mu),\az^{\bKs,j}(\cdot;\mu)\right)$ denote the simple $L^2(\R)$- eigenpair of $\DiKs(\mu)$ given by Proposition \ref{cor-D-spec}:
\begin{equation}\nn
\Big(  \DiKs(\mu) -\var^\bKs_j(\mu)  \Big) \alpha^{\bKs,j}(s;\mu)\ =\ 0, \ \ \ \ \alpha^{\bKs,j}(s;\mu) \in L^2(\R,\C^2).
\end{equation}
 The multiple scale expansion approach of \S\ref{ms-Hd} can be continued to arbitrary fixed order in $\delta$ and yields, for any fixed $M\ge1$ and $\delta$ sufficiently small,   a construction of an $\mathcal{O}(\delta^M)$ approximate $\Ll^2_{\bKs\cdot\tv_1+\delta\mu}$ eigenpair $\big(E_{j,M}^{\bKs,\delta}(\mu),\Psi_{j,M}^{\bKs,\delta}(\bx,s;\mu)\big)$  of \eqref{kp-evpp} with 
 \begin{equation}\nn
  \left(H_{\edge}^\delta-E_{j,M}^{\bKs,\delta}(\mu)\right)\Psi_{j,M}^{\bKs,\delta}(\bx,\delta\tk_2\cdot\bx;\mu) \ = \ \OO_{\Hh^s_{\bKs\tv_1+\delta\mu}}(\delta^M)\ .
\end{equation}
For $j=-N,\dots,N$, these approximate   eigenpairs of  \eqref{kp-evpp} have expansions in powers of $\delta$: 
\begin{align}
\Psi_{j,M}^{\bKs,\delta}(\bx,s;\mu)\  &=\   \Phi^\bKs (\bx)^\top\ \alpha^{\bKs,j}(s;\mu)\  +\  \delta^2\ \psi_{2}^{\bKs,j} (\bx,s;\mu) + \dots + \delta^M\ \psi_{M}^{\bKs,j} (\bx,s;\mu) \nn\\
E_{j,M}^{\bKs,\delta}(\mu)\ &=\ E_D\ +\ \delta\  \var^\bKs_j(\mu)\ +\ \delta^2\ e_{2}^{\bKs,j}(\mu) \ +\  \dots\ +\ \delta^M\ e_{M}^{\bKs,j}(\mu).\label{quasi-m}
\end{align}
Since the eigenvalues $\var^\bKs_j(\mu)$  of the effective operator $\DiKs(\mu)$ are distinct,
the approximate $\Ll^2_{\bKs\cdot\tv_1+\delta\mu}$ eigenvalues are  $\mathcal{O}(\delta)-$ separated. 

Fix an arbitrary $\mu_0>0$ and let $|\mu|\le\mu_0$. Basic general properties of self-adjoint operators 
together with precise information about the location of approximate eigenvalues of $H^\delta$ within the spectral gap  about $E_D$ enable us to conclude the existence of genuine eigenvalues of  $H^\delta$  in this gap; see  \cite[\S3 and Appendix]{DFW:18}. 

From the multiple scale procedure above one can conclude 
for the case of zigzag type edges the existence of $(2N+1)$ $\Ll^2_{2\pi/3+\delta\mu}-$ eigenvalues and $(2N+1)$ $\Ll^2_{-2\pi/3+\delta\mu}-$ eigenvalues of $H^\delta$ which are located within the order $\mathcal{O}(\delta)$ gap in the essential spectrum, 
 $(a_\delta(\mu),b_\delta(\mu))$, about $E_D$  \cite{FLW-2d_edge:16}.
 The resolvent expansion \eqref{zz-res} of Theorem \ref{res-exp}
  ensures, for $\delta$ sufficiently small,  that these  eigenvalues are simple and that they are the only eigenvalues in 
   this spectral gap, located an arbitrarily small fixed distance from the boundary of the gap;   see \cite{Drouot:18b} and Corollary \ref{KKp-eigsZZ}.  This is key to the topological arguments of \cite{Drouot:18b,Drouot:19} and to the  perspective outlined in \S\ref{topology}.
   
The  situation is different for the case of armchair-type edges. In this case,  we have  $\bK\cdot\tv_1=\bKp\cdot\tv_1=0$.  Hence, via the multiple  scale expansion procedure we produce
 $2\times(2N+1)$ $\Ll^2_{0+\delta\mu}-$  approximate (and then genuine) eigenpairs of $H^\delta$: $(2N+1)$ are generated by $\sigma_{\rm pp}(\DiK(\mu))$ and $(2N+1)$ by $\sigma_{\rm pp}(\DiKp(\mu))$. Since the $2N$ elements of  $\sigma_{\rm pp}(\DiK(\mu))$ and  $\sigma_{\rm pp}(\DiKp(\mu))$, corresponding to $j\ne0$, are equal (see \eqref{var-j}),  the corresponding $4N$ branches of eigenvalues of $H^\delta$ acting in $\Ll^2_{0+\delta\mu}-$ are degenerate 
   through order $\delta$. (We are as yet unable to detect a splitting at finite order.)  Therefore, in Theorem \ref{res-exp}, we give an expansion of  the  rank-two projector associated with each pair eigenvalues of $H^\delta$ generated by the pairs of approximate eigenvalues 
    $\{E_{j,M}^{\bK,\delta}(\mu),E_{j,M}^{\bKp,\delta}(\mu)\}$, $0<|j|\le N$; see Corollary \ref{KKp-eigsAC}.
     The topological consequences nonetheless persist.
\subsection{Classification of rational edges}\label{ZZ-AC}
Theorem \ref{res-exp} and Corollaries \ref{KKp-eigsZZ}-\ref{KKp-eigsAC} show
the dependence of spectral properties of $H^\delta_\edge$ on whether the edge is of zigzag or armchair type.
The purpose of this section is to motivate this classification of rational edges.
 
Consider solutions of the eigenvalue problem \eqref{kp-evpp}. Since an eigenfunction $\Psi$ is in $\Ll^2_\kpar$, it has a decomposition \eqref{slice-rep}:
 \begin{equation}\label{slice-rep1}
 \Psi(\bx) = \frac{1}{2\pi}\int_{-\pi}^\pi \widetilde{\Psi}(\kpar\ktilde_1+t\ktilde_2,\bx) dt.
 \end{equation}
 
 From the discussion of \S\ref{app2gen} an eigenfunction $\Psi$ with energy near $E_D$  has an approximation to arbitrary order in $\delta$ of multiscale type. Its dominant spectral contributions
  come from Floquet-Bloch modes with quasi-momenta near high symmetry quasi-momenta. Hence, the dominant spectral contributions to the representation \eqref{slice-rep1} is from those values of $t\in[-\pi,\pi]$ for which 
 $\kpar\tk_1+t\tk_2=(\bKs\cdot\tv_1+\delta\mu)\tk_1+t\tk_2$ is near the sublattice $\bK+\Lambda^*$ or near the sublattice
  $\bKp+\Lambda^*$.

% 
%In addition, the multiscale analysis of \S\ref{ms-Hd} indicates that approximate edge states are concentrated spectrally near Dirac point momenta. We expect edge states to be similarly concentrated. This means that the dominant part in the RHS of \eqref{slice-rep1} arises when $\kpar \tk_1 + \tk_2 \R$ passes near $\{\bK,\bKp\} \mod \Lambda$. 

Thus, we distinguish two cases where, for $|t|\le\pi$:
\begin{itemize}
\item[(a)] $\kpar \tk_1 + t \cdot \tk_2$ passes near $\bK + \Lambda^*$ or near $\bKp + \Lambda^*$, but not near both, or 
\item[(b)] $\kpar \tk_1 + t \cdot \tk_2$ passes near both $\bK+ \Lambda^*$ and $\bKp+ \Lambda^*$. 
\end{itemize}
These cases are in fact characterized by a simple arithmetic relation on the relatively prime integers, $a_1, b_1$, which define the edge direction: $\tv_1=a_1\tv_1+b_1\tv_2$.
To derive this relation, assume that (a) or (b) holds. Taking the scalar product of vectors on the segment $\kpar \tk_1 + t \cdot \tk_2$ with $\tv_1$, we deduce that $\kpar$ is near $\bK \cdot \tv_1 = 2\pi(a_1-b_1)/3\mod 2\pi$ or $\bKp \cdot \tv_1= - 2\pi(a_1-b_1)/3$ modulo  $2\pi$. Since $a_1$ and $b_1$ are integers, $\kpar$ is near $-2\pi/3$, $0$ or $+2\pi/3$ modulo  $2\pi$. Now (b) holds if and only if  $\bK \cdot \tv_1\approx\bKp \cdot \tv_1$ modulo $2\pi$ and hence 
$a_1\approx b_1$ mod $3$.  Again since $a_1$ and $b_1$ are integers, case (b) implies  $a_1=b_1$ modulo $3$. Thus, we classify rational edges according to whether or not $a_1$ and $b_1$ are equivalent modulo $3$:

\begin{definition}\label{KKp-def} Let $\tv_1=a_1\bv_1+b_1\bv_2$ where $a_1$ and $b_1$ are relatively prime. 
\begin{itemize}
\item[(a)] We say that $\R\vtilde_1$ is a zigzag-type edge if $a_1\neq b_1$ mod $3$. In this case,
 $\{ \bK \cdot \tv_1, \bKp \cdot \tv_1 \} = \{-2\pi/3,2\pi/3\} \mod 2\pi$.
\item[(b)] We say that $\R\vtilde_1$ is a armchair-type edge if $a_1 = b_1$ mod $3$. In this case, $\{ \bK \cdot \tv_1, \bKp \cdot \tv_1 \} = \{0\} \mod 2\pi$.
\end{itemize}
\end{definition}
\nit This  terminology is motivated by the most commonly studied cases in the chemistry and physics literature: the armchair edge -- where $\vtilde_1=\bv_1+\bv_2$ ($a_1=b_1=1$); and the zigzag edge -- where $\vtilde_1=\bv_1$ ($a_1=1, b_1=0$). 
  
\section{Main results: resolvent expansion and edge states}\label{main-results}

This section contains our main results:
\begin{itemize}
\item Theorem \ref{res-exp}: the resolvent expansion of $H_{\edge}^\delta$ acting on $\Ll^2_\kpar$.
\item Corollaries \ref{KKp-eigsZZ} and \ref{KKp-eigsAC}: identifications of all possible edge states, whose energies are in the spectral gap about the Dirac point energy, $E_D$.
\end{itemize}

\subsection{Resolvent expansion for $H_{\edge}^\delta$}\label{sec:5a}
We first provide some setup for the resolvent expansion. From the $L^2_\bKs$- kernel of $H^0-E_D$ (see \S\ref{HLP} and Proposition \ref{ips}), 
we form $\C^4-$ and $\C^2-$ valued Floquet-Bloch modes:
\begin{equation}\label{eq:3g}
\Phi = \matrice{\Phi^\bK  \\ \Phi^\bKp } \in C^\infty(\R^2,\C^4), \ \ \ \ \text{where } \
\Phi^\bKs = 
\matrice{\Phi_1^\bKs \\ \Phi_2^\bKs } \in C^\infty(\R^2,\C^2).
\end{equation}
Zigzag- versus armchair-type edges are defined in Definition \ref{KKp-def}. If $\tv_1$ is a zigzag-type edge then $\{ \bK\cdot \tv_1, \bKp \cdot \tv_1\} = \{ -2\pi/3,2\pi/3 \} \mod 2\pi$; thus
 $\Ll^2_{\bK\cdot\tv_1}\ne \Ll^2_{\bKp\cdot\tv_1}$.  If $\tv_1$ is an armchair-type edge then $\{ \bK\cdot \tv_1, \bKp \cdot \tv_1\} = \{0 \} \mod 2\pi$; thus
 $\Ll^2_{\bK\cdot\tv_1}= \Ll^2_{\bKp\cdot\tv_1}=\Ll^2_{0}$. 
 
Corresponding to the two edge-types,  we define the (averaging) operator \begin{equations}\label{TT}
\TT \ : \  \Dom(\TT) \rightarrow L^2(\R^2,\C^4), \ \ \ \ 
\TT u(t) = \int_{\R/\Z} \left( \ove{\Phi} u\right)(s\tv_1+t\tv_2) ds
\ \ \ \
 \text{where} 
 \\
\Dom(\TT) \de \systeme{
\Ll^2_{2\pi/3} \oplus  \Ll^2_{-2\pi/3}  &  \text{if $\tv_1$ is a zigzag-type edge;}
\\ 
 \Ll_0^2  &  \text{if $\tv_1$ is an armchair-type edge.}}
\end{equations}
 Its adjoint, $\TT^* :  L^2(\R,\C^4) \rightarrow \Dom(\TT)$, is the restriction operator  given by 
 \begin{equation} \TT^* v(\bx) = \Phi(\bx)^\top v(\tk_2\cdot \bx).\label{TT*}\end{equation} 
 $\TT$ and $\TT^*$ are bounded linear operators.
 
We also introduce unitary dilations on $L^2(\R)$: 
\begin{equation}
U_\delta g(s) \de \delta^{-1/2} g\left(\delta^{-1} s \right)\quad {\rm and}\quad U^*_\delta g(s) \de \delta^{1/2} g\left(\delta s \right)
.\label{dilate}
\end{equation}

Finally, we form a $4 \times 4$ Dirac operator using $\DiK(\mu)$ and $\DiKp(\mu)$, arising in \S\ref{ms-Hd}:
\begin{equation}
\Dirac(\mu) \de 
\matrice{\DiK(\mu) & 0 \\ 0 & \DiKp(\mu)} \ : \ H^1(\R,\C^4) \rightarrow L^2(\R,\C^4).
\label{Di-mu}\end{equation}
The operator $\Dirac(\mu)$, acting in $L^2(\R)$, has spectrum equal to $\sigma_{L^2}\big(\DiK(\mu)\big)\cup\sigma_{L^2}\big(\DiKp(\mu)\big)$. In particular, it has a gap $\big( -\theta_\gap(\mu), \theta_\gap(\mu) \big)$ in its essential spectrum; see Proposition \ref{cor-D-spec}.

The mathematical core of this paper is the following:

\begin{theorem}[Resolvent expansion]\label{res-exp}
 Let $\tv_1$ denote a rational edge, {\it i.e.} $\tv_1=a_1\bv_1+b_1\bv_2$, with $a_1$ and $b_1$ relatively prime integers. Assume $\rm{(H1)-(H3)}$ of \S\ref{HLP}. Fix $\epsilon, \ \mu_0 > 0$ and $\bKs \in \{\bK,\bKp\}$. There exists $\delta_0 > 0$ such that if
\begin{equation}
0<\delta<\delta_0, \ \ \ |\mu|<\mu_0,\ \ \ 
|z|<\var_\gap(\mu)-\epsilon, \ \ \ {\rm dist}\Big( z ,\ \sigma_{L^2}\big( \Dirac(\mu) \big) \Big) \ge\epsilon
\end{equation} 
then $H^\delta - E_D - \delta z:\Hh^2_{\bKs \cdot \tv_1 + \delta \mu}\to\Ll^2_{\bKs \cdot \tv_1 + \delta \mu}$  is invertible.  Moreover, the following expansions hold according
 to whether $\R\tv_1$ is a zigzag-type or armchair-type edge:
\begin{itemize}
\item \underline{\it Zigzag case:}\ When $\tv_1$ is a zigzag-type edge:
\begin{equations}\label{zz-res}
\left( H^\delta -E_D-\delta z  \right)^{-1}\Big|_{\Ll^2_{2\pi/3+\delta\mu}\oplus\ \Ll^2_{-2\pi/3+
\delta\mu} } 
\\
 =   \
 \frac{1}{\delta}\cdot  \left( U_\delta \TT e^{-i\mu \lr{\bl,\bx}} \right)^* \cdot \left(\Dirac(\mu)- z\right)^{-1}
 \cdot  U_\delta \TT e^{-i\mu \lr{\bl,\bx}}+  \OO_{\Ll^2_{2\pi/3+\delta\mu}\oplus  \Ll^2_{-2\pi/3+
\delta\mu}}\left(\delta^{-2/3}\right).
\end{equations}
\item \underline{\it Armchair case:} When $\tv_1$ is an armchair-type edge:
\begin{equations}\label{ac-res}
\left( H^\delta -E_D-\delta z  \right)^{-1}\Big|_{\Ll^2_{\delta\mu} }  \ = \
 \frac{1}{\delta}\cdot \left( U_\delta \TT e^{-i\mu \lr{\bl,\bx}} \right)^* \cdot \left(\Dirac(\mu)- z \right)^{-1}
 \cdot U_\delta \TT e^{-i\mu \lr{\bl,\bx}} +  \OO_{\Ll^2_{\delta\mu}}\left(\delta^{-2/3}\right).
\end{equations} 
\end{itemize}
The analogous statements hold for $\tHd$, with $\Dirac(\mu)$ replaced by $\widetilde\Dirac(\mu)$ etc.
\end{theorem}
%
%\footnote{\tr{There were mistakes in the statement: the error term cannot map $\Ll^2 \to \Hh^2$ as $\left( \Dirac-z \right)^{-1}$ does not range to $H^2$ (one cannot even hope for $\Ll^2 \rightarrow \Hh^1$ as we need to loose some additional regularity to make the homogenization arguments work); and some exponential factors were missing.}}
%

The expansion for zigzag edges appears in \cite[Theorem 1 and 3]{Drouot:18b}. 
The current work extends the resolvent expansion to the more subtle case of armchair edges.
The current work combines of techniques from \cite{FLW-2d_edge:16} and \cite{Drouot:18b} to provide a unified treatment of all rational edges by a more direct and transparent strategy  .
Note that the expansion \eqref{zz-res} combines expansions of $(H^\delta-E_D-\delta z)^{-1} $ in the two spaces  $\Ll^2_{2\pi/3+\delta\mu}$ and $\Ll^2_{-2\pi/3+\delta\mu}$, respectively in terms of the two effective Dirac resolvents $\left(\DiK(\mu)- z \right){}^{-1}$ and $\left(\DiKp(\mu)- z \right){}^{-1}$. The expansion \eqref{ac-res} is an expansion of $(H^\delta-E_D-\delta z)^{-1} $ in the single space  $\Ll^2_{0+\delta\mu}$ in terms of the block-diagonal resolvent 
$\left(\Dirac(\mu)- z \right){}^{-1}$.

We next discuss a key consequence of the resolvent expansion. Edge states energies are $\Ll^2_{\bKs\cdot\tv_1+\delta\mu}$-eigenvalues of $H_{\edge}^\delta$.  The method of \cite{FLW-2d_edge:16} shows that each of the $2N+1$ point eigenvalues of  $\Di^\bKs(\mu)$ generates a point eigenvalue of $H_{\edge}^\delta$ in the spectral gap about $E_D$, $|\mu|\le\mu_0$ and $\delta$ sufficiently small.   Leading order expressions for the eigenvectors were also  constructed.  These eigenvalues are poles of the resolvent of $H^\delta$ in the spectral gap. Using Theorem \ref{res-exp} together with the arguments of \cite{Drouot:18b,DFW:18}, we can, for $\delta$ small:
\begin{itemize}
\item[(a)] locate these eigenvalues/poles to arbitrary finite order in $\delta$;
\item[(b)] show that all corresponding eigenvectors have a multiple scale structure;
\item[(c)] expand the edge-state  eigenprojectors to arbitrary finite order in $\delta$; and 
\item[(d)] explain the simulations of edge state curves displayed in Figures \ref{CnoP} and \ref{PnoC}; see also  \cite{FLW-2d_materials:15} and aspects of the experimental study \cite{Rechtsman-etal:18}.
\end{itemize}
The results are detailed in the following two corollaries to Theorem \ref{res-exp}.
For zigzag-type edges, we recover \cite[Corollary 4]{Drouot:18b}:
%\footnote{\tg{Splitting the corollary in AC versus ZZ does not make the paper longer. I think it is better this way (and I  am guessing that it is also your prefered way)}}
%\footnote{\tr{The corollary had typos. }}

\begin{corollary}
\label{KKp-eigsZZ} Let $\tv_1$ be a zigzag-type edge and consider the setting  of Theorem \ref{res-exp}. Fix $\epsilon$ arbitrarily small and positive. For $\delta$ sufficiently small, the operator  $H^\delta$ acting in $\Ll^2_{\bKs\cdot \tv_1 + \delta\mu}$  has precisely $2N+1$ eigenvalues in the gap
\begin{equation}\label{eq:3f}
\Big( E_D - \delta \big(\theta_\gap(\mu) - \epsilon\big),  E_D + \big(\theta_\gap(\mu)-\epsilon\big) \Big).
\end{equation}
in its essential spectrum.
The associated eigenpairs $\big(E_{j}^{\bKs,\delta}(\mu), \Psi_{j}^{\bKs,\delta}(\cdot,\cdot;\mu)\big)$, $j =-N, \dots, N$ can be expanded to arbitrary finite order in $\delta$: 
\begin{equation}
E_{j}^{\bKs,\delta}(\mu) \ = \ E^{\bKs,\delta}_{j,M}(\mu) + \ O\left(\delta^{M+1}\right),\ \ \ 
\Psi_{j}^{\bKs,\delta}(\bx;\mu) \ = \  \Psi_{j,M}^{\bKs,\delta}(\bx, \delta \tk_2 \cdot \bx;\mu) \ + \   O_{\Hh^s_{\bKs \cdot \tv_1 + \delta \mu}}(\delta^{M+1}), \label{Hedge-eig}
\end{equation}
where the expansion of  $(E_{j,M}^{\bKs,\delta}(\mu), \Psi_{j,M}^{\bKs,\delta})$ is displayed in \eqref{quasi-m}.
Analogous statements hold for $\tHd$.
\end{corollary}

The next result describes  the spectrum of $H^\delta_\edge$ for armchair-type edges.  

\begin{corollary}\label{KKp-eigsAC} Fix $\epsilon>0$ and let $\tv_1$ be an armchair-type edge. Then,  $H^\delta$ has precisely $4N+2$ eigenvalues in the $\Ll^2_{\delta\mu}$-essential spectral gap 
\begin{equation}\label{eq:3f}
\Big( E_D - \delta \big(\theta_\gap(\mu) - \epsilon\big),  E_D + \big(\theta_\gap(\mu)-\epsilon\big) \Big).
\end{equation}
Moreover:
\begin{itemize}
\item The eigenvalues -- denoted $\{ E_{j}^{\bKs,\delta}(\mu) \big\}$ for $j =-N, \dots, N$ and $\bKs = \bK, \ \bKp$ -- are of multiplicity at most $2$ and may be expanded to arbitrary finite order in $\delta$
\begin{equation}
E_{j}^{\bKs,\delta}(\mu) \ = \ E^{\bKs,\delta}_{j,M}(\mu) \ + \ O\left(\delta^{M+1}\right);\ \ \textrm{see \eqref{quasi-m}}
\label{Hedge-eig}
\end{equation}
\item For each $j \in [-N,N]$, the rank-two eigenprojectors $\Pi_{j,\mu}^\delta$ associated with the pair of eigenvalues $\{E_{j}^{\bK,\delta}(\mu),E_{j}^{\bKp,\delta}(\mu)\}$ have expansions, to arbitrary finite order in $\delta$, in terms of the projectors:
%
%\footnote{\tb{check this statement}}
%
\begin{equation}
 \Psi_{j,M}^{\bKs,\delta}(\mu) \otimes \Psi_{j,M}^{\bKs,\delta}(\mu) \ + \ \OO_{\Hh^s_{\delta \mu}}\left(\delta^{M+1}\right),\quad \bKs=\bK, \bKp; \ 
\textrm{see \eqref{quasi-m}.} \label{Hedge-eig}
\end{equation}
\end{itemize}
The analogous statements hold for $\tHd$.
\end{corollary}
\nit In Corollary \ref{KKp-eigsAC}, the two-dimensional eigenprojectors for the eigenvalue pairs 
 $E_{j}^{\bK,\delta}(\mu)$, $E_{j}^{\bKp,\delta}(\mu)$ is expanded, since 
we only know, from perturbation theory,  that the splitting is $\mathcal{O}(\delta^3)$ or higher.

\subsection{Global character of edge state curves}\label{topology2}
 Figures \ref{CnoP} and \ref{PnoC} display edge state curves for zigzag-type and armchair-type edges for $H^\delta$ ($\CCC-$ invariant, $\PPP-$ breaking) and for $\tHd$ ($\CCC-$ breaking, $\PPP-$ invariant).     In this section we use Theorem \ref{res-exp}, and Corollaries \ref{KKp-eigsZZ} and \ref{KKp-eigsAC} to  explain these bifurcation curves in terms of the properties of effective Dirac operators. Indeed,   these corollaries imply that all $\Ll^2_{\bKs+\delta\mu}$ eigenvalues, $E$, 
 of $H^\delta_{\rm edge}$ in the spectral gap about $E_D$ satisfy:
$ \left(E - E_D\right)/\delta\ =\ \var + \mathcal{O}(\delta)$, 
 where $\var$ is an $L^2(\R)-$ eigenvalue of $\DiKs(\mu)$ (respectively  $\tDiKs(\mu)$), $\bKs=\bK, \bKp$. That is, a magnification of the $\Ll^2_{\bKs+\delta\mu}$ edge state eigenvalue curves of $H^\delta_{\rm edge}$ near $(E,\kpar)=(E_D,\bKs\cdot\tv_1)$ gives the eigenvalue curves of effective Dirac operators. The spectra of these effective operators are completely described in \S\ref{D-eff-spec}.
Furthermore,  that this local picture determines the essential characteristics of edge state curves over the full 
 range $0\le\kpar<2\pi$ is implied by the {\it a priori} information of \cite[Lemma 7.2]{Drouot:18b}, which implies that for $|\kpar-\bKs\cdot \tv_1|\ge C\delta$ (for some $C>0$ fixed), no edge state can have energy near $E_D$ .   We now proceed with a discussion of Figures \ref{CnoP} and \ref{PnoC};  the zigzag subcases was covered in \cite{Drouot:18b}.
 
 \subsubsection{Zigzag-type edge, $H^\delta$ ($\CCC-$ invariant, $\PPP-$ breaking),  Figure \ref{CnoP}, left panel}\label{zz-CnoP}
 
See also the schematic in Figure \ref{fig:7}.
 The effective Dirac operator $\DiK(\mu)$ determines the $\Ll^2_{2\pi/3+\delta\mu}$ spectrum of $H^\delta$ and $\DiKp(\mu)$ determines the $\Ll^2_{-2\pi/3+\delta\mu}$ spectrum of $H^\delta$.
 The simulations show a single edge state curve in the spectral gap varying linearly near $\kpar=\pm2\pi/3$. This local behavior is described by the spectra of effective Dirac operators
  (Proposition \ref{cor-D-spec}),  which have a single edge state curve ($N=0$), varying linearly with $\mu$. The slopes of the curves near $\kpar=2\pi/3$ and $\kpar=-2\pi/3$
  are equal and opposite since $\sigma_{L^2}(\DiKp)=-\sigma_{L^2}(\DiK)$; see \S\ref{D-eff-spec}. 
 
% The discrete spectrum of $\DiKs(\mu)$ is rather explit: Proposition \ref{cor-D-spec} shows that it is formed by $2N+1$ eigenvalue, prescribed by $\DiKs(0)$. This allows to describe edge state bifurcations  curves of $H^\delta$ with $\kpar$ near $\bKs \cdot \tv_1$. Since $\DiK(\mu)$ and $\DiKp(\mu)$ have opposite discrete spectra, these bifurcation curves have oppposite slopes near $\bK$ and $\bKp$. An additional result 
%This analysis yields a full description of edge states curves for TRS deformations and zigzag-type edges. Figure \ref{fig:7} is a pictorial representation of the argument when $\DiK(\mu)$ has one eigenvalue ($N=1$); see also \cite[Figure 6 and 7]{Drouot:18b}. This explains the numerical plots of Figure \ref{CnoP}(a).

 \subsubsection{Armchair-type edge, $H^\delta$ ($\CCC-$ invariant, $\PPP-$ breaking),  Figure \ref{CnoP}, right panel}\label{ac-CnoP} See also the schematic in Figure \ref{fig:8}. For armchair-type edges
 $\bK\cdot\tv_1=\bKp\cdot\tv_1=0$, and hence $\Ll^2_{\bK\cdot\tv_1}=\Ll^2_{\bKp\cdot\tv_1}=\Ll^2_{0}$. Theorem \ref{res-exp} implies that the resolvent and character of the spectrum in the spectral gap about energy $E_D$
  is determined by the block-diagonal Dirac operator $\Dirac(\mu)$. This effective operator has a two-fold degenerate eigenvalue for $\mu=0$ (corresponding to $\kpar=0$) --- a multiplicity one eigenvalue contributed by each of the two blocks
   $\DiK(0)$ and $\DiKp(0)$ --- each of which departs from zero linearly in $\mu$ with equal and opposite slope; $\sigma_{L^2}(\DiKp)=-\sigma_{L^2}(\DiK)$, Proposition \ref{cor-D-spec} and Figure \ref{fig:8}, left panel.

Two scenarios are possible for $\delta\ne0$:
\begin{itemize}
\item[(i)] either the two perturbed eigenvalue curves of $H^\delta$ cross near $\kpar = 0$, or
\item[(ii)] they split and $H^\delta$ acquires a full $L^2$-gap near energy $E_D$. 
\end{itemize}
While  the numerical simulations plotted in Figure \ref{PnoC}(b) favor (ii), our results do not preclude  either possibility. Analytic calculations show that the splitting is at most $\mathcal{O}(\delta^3)$ and we conjecture that it is $O(\delta^\infty)$.

 \subsubsection{Zigzag-type edge, $\tHd$ ($\CCC-$ breaking, $\PPP-$ invariant),  Figure \ref{PnoC}, left panel}\label{zz-PnoC}
 
  The effective Dirac operator $\tDiK(\mu)$ determines the $\Ll^2_{2\pi/3+\delta\mu}$ spectrum of $\tHd$ and $\tDiKp(\mu)$ determines the $\Ll^2_{-2\pi/3+\delta\mu}$ spectrum of $\tHd$.
  In contrast to the $\CCC-$ invariant, $\PPP-$ breaking case (\S\ref{zz-CnoP}), the  protected eigenvalue curves of $\tDiK(\mu)$ and $\tDiKp(\mu)$, which pass through zero energy, have the same slope since $\sigma_{L^2}\big(\tDiK(\mu)\big)=\sigma_{L^2}\big(\tDiKp(\mu)\big)$; see \S\ref{D-eff-spec}. This accounts for the behavior in Figure \ref{PnoC}, left panel.
 \subsubsection{Armchair-type edge, $\tHd$ ($\CCC-$ breaking, $\PPP-$ invariant),  Figure \ref{PnoC}, right panel}\label{ac-CnoP}
  See also the schematic in Figure \ref{fig:9}. In contrast to 
  the $\CCC-$ invariant, $\PPP-$ breaking case (\S\ref{ac-CnoP}), the  protected eigenvalue curves of $\tDiK(\mu)$ and $\tDiKp(\mu)$ through  zero energy are identical. As in \S\ref{ac-CnoP}, for $\delta\ne0$ the splitting is at high order as reflected in the right panel of  Figure \ref{PnoC}; see also Figure \ref{fig:9}.
  
% For armchair-type edges,  $\bK \cdot \tv_1 = \bKp \cdot \tv_1 = 0$. Thus it suffices to superpose zooms over the eigenvalue curves of $\DiK(\mu)$  and $\DiKp(\mu)$ to obtain those of $H^\delta_{\edge}$ near $\kpar = 0$. See Figures \ref{fig:8} and \ref{fig:9}. When preserving TRS, the eigenvalues curves of $\Dirac(\mu)$ cross at $(0,0)$. 
%

%Corollaries \ref{KKp-eigsZZ} and \ref{KKp-eigsAC} imply that the $L^2$-eigenvalue curves of $\DiKs(\mu)$ are blow-up of the $\Ll^2_\kpar$-eigenvalues curves of $H^\delta$. The energies / momenta are in correspondence via respective maps
%\begin{equation}
%E \mapsto \dfrac{E-E_D}{\delta} + O(1), \ \ \  \ \kpar \mapsto \dfrac{\kpar-\bKs \cdot \tv_1}{\delta} + O(1).
%\end{equation}
%In other words, a zoom on the eigenvalue curves of $H^\delta$ within a $\delta$-window of $(E_D,\bKs \cdot \tv_1)$ produces eigenvalue curves of $\DiKs$ modulo $O(\delta)$. 

 \begin{figure}
          \includegraphics{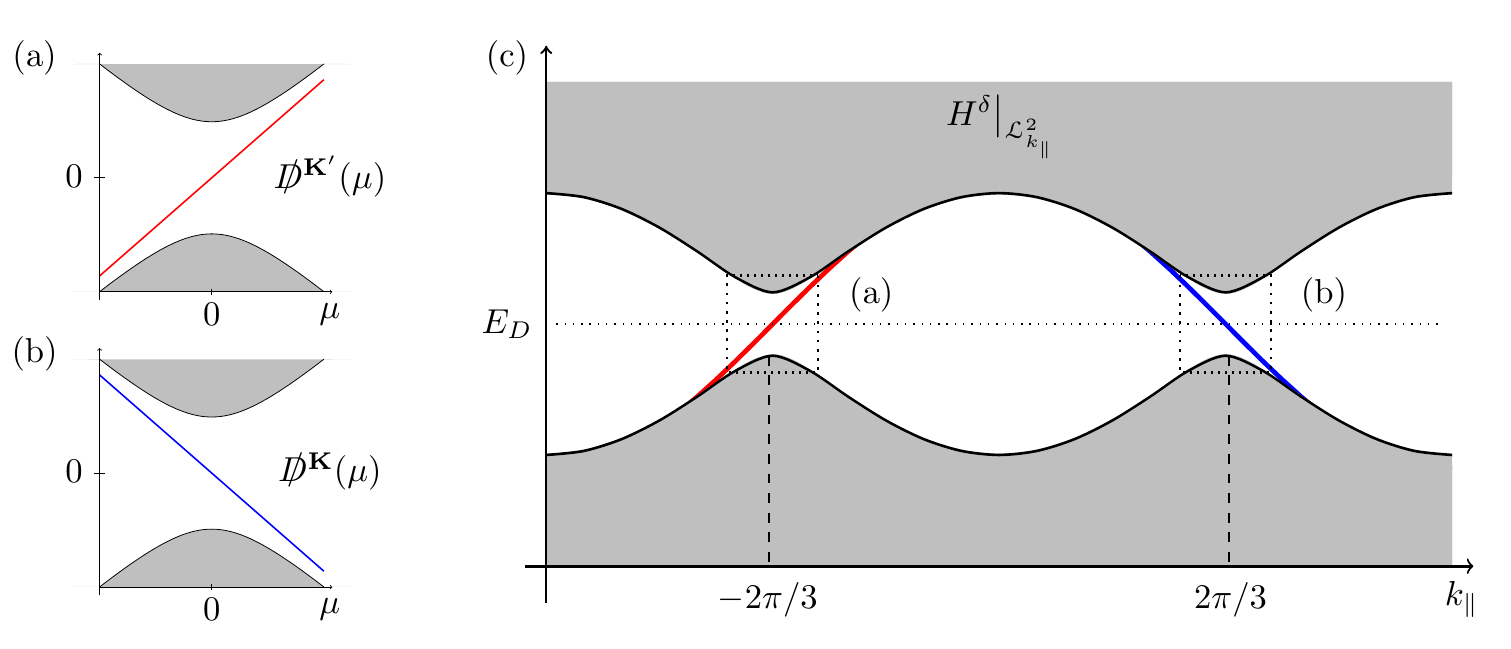}   
    \caption{{\sl Zigzag-type edge, $H^\delta$}: Spectrum of (a) $\DiKp(\mu)$ (red) and (b) $\DiK(\mu)$ (blue) as functions of $\mu$; and (c) $\Ll^2_\kpar$-spectrum of $H^\delta$ as a function of $\kpar$. Zooming at scale $\delta^{-1}$ on $(c)$ near $(-2\pi/3,E_D)$ (resp. near $(2\pi/3,E_D)$) produces (a) (resp. (b)).}
   \label{fig:7}
\end{figure}

\begin{figure}
 %       \includegraphics[height=1.5in]{dirac}   
 %     \includegraphics[height=1.5in]{dirac2}
 %  \\ 
      \includegraphics{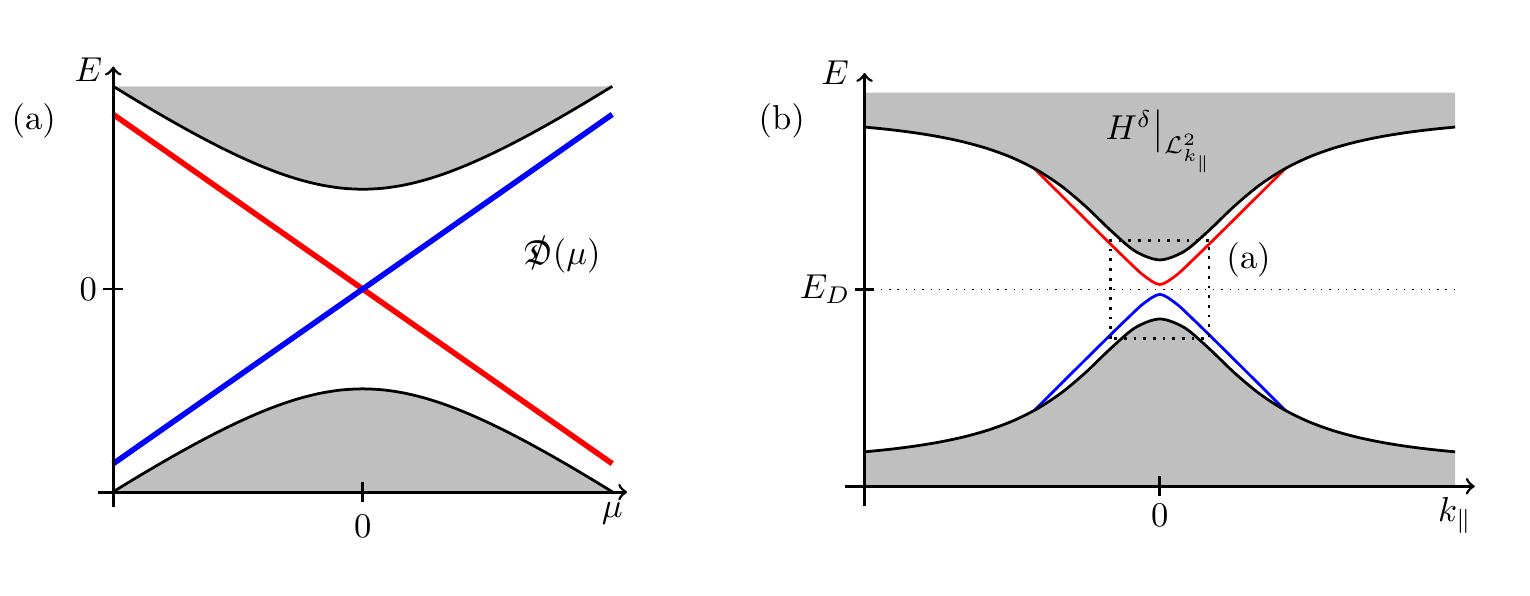}
    \caption{{\sl Armchair-type edge, $H^\delta$:} (a) Spectrum of  $\mathfrak{\Dirac}(\mu)$ as a function of $\mu$ for armchair-type edges with TRS:  superposition of the left pannels of Figure \ref{fig:7} $\mathcal{C}$-invariant case. (b) $\Ll^2_\kpar$-spectrum of $H^\delta$ as a function of $\kpar$, with anticipated splitting of edge state curves.}
   \label{fig:8}
\end{figure}

\begin{figure}
        \includegraphics{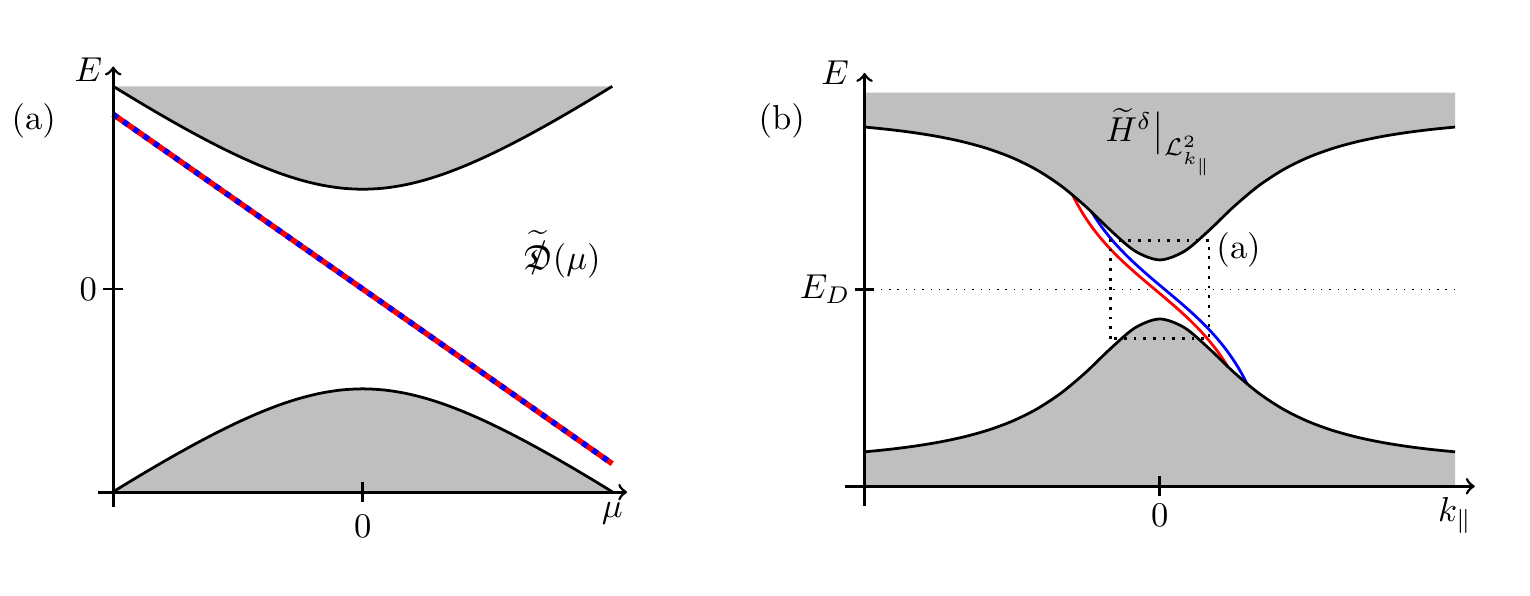}
   \caption{{\sl Armchair-type edge, $\tHd$:} (a) Spectrum of  $\widetilde{\Dirac}(\mu)$ as a function of $\mu$ for armchair-type edges with broken TRS:  two superposed copies of Figure \ref{fig:7}(b). (b) $\Ll^2_\kpar$-spectrum of $\tHd$ as a function of $\kpar$.}
   \label{fig:9}
\end{figure}

\subsection{A remark on the Valley Hall effect}\label{vhall}

Using our results, one can construct fully localized edge wave-packets  spectrally concentrated near $\bK$ or $\bKp$. In condensed matter physics, $\bK$ and $\bKp$ are referred to as 
{\it valley} degrees of freedom. 
The line-defect orientation has physical implications for energy propagation along the edge.
For $\CC-$ (time-reversal) invariant systems and a zigzag-type edge, wave-packets concentrated near $\bK$ or $\bKp$ 
 travel in opposite directions since the {\it valley indices} $\text{sgn}\big(\var^\bK\big)$ and  $\text{sgn}\big(\var^\bKp\big)$ have opposite sign; see Proposition \ref{H-eff} .
Therefore, in a non-magnetic honeycomb system, there is quantum Hall-like energy transport
 whose direction is controlled by exciting the appropriate high-symmetry sublattice. This is closely related to the {\it Valley Hall Effect}
 \cite{Mak-etal:14,QAP18}. This effect is not present for armchair-type edges.

\subsection{Topological considerations}\label{sec:5.3} We provide a topological perspective of the bifurcation curves using arguments from \cite{Drouot:18b}.
 The spectral flow of $H_{\rm edge}^\delta$ counts the signed number of $\Ll^2_\kpar$-eigenvalues
  of $H_{\rm edge}^\delta$ intersecting $E_D$ as $\kpar$ runs through $[0,2\pi]$;  see \S\ref{topology}. 
  
  The operator  $H^\delta$ breaks $\PP$ but not $\CCC$. The $\CCC$-invariance implies that edge state curves are symmetric about $\kpar = 0$. It follows that the spectral flow is equal to $0$. This is seen in the simulation results displayed in Figure \ref{CnoP}. 
  By adding suitably designed (compact) perturbations, one can deform $H^\delta$ to an operator with no edge states;   the edge modes in $\CCC-$ invariant systems are not topologically protected. 
  \footnote{We believe such non-protected states can be continuously deformed away through a family of spatially localized perturbations of $H^\delta_\edge$; this however is not immediately implied by general topological arguments.}

%   \footnote{\tg{Some comments regarding dynamical implications were moved to \S1.5. I think they are too important to only appear page 28. We could also repeat them here.}}
%\footnote{ \tg{No longer needed here:} Indeed, by the spectral localization analysis presented in \cite[Lemma 7.2]{Drouot:18b} shows
%   that for some $C>0$,  if $|\kpar-\widehat{\kpar}|\ge C\delta$ and  $\Ll^2_\kpar-\sigma(H^\delta) \cap \mathcal{G}_\delta\ =\ \{\ \}$; 
%    $(\ H^\delta-E\ )^{-1}$ acting in $\Ll^2_\kpar$ is invertible. 
%On the other hand,  the edge states energies, $E(\kpar)$,  for  $|\kpar-\widehat{\kpar}|\le C\delta$ 
%are derived from the scaled eigenvalue curves of
%    $\Dirac(\mu)$, whose  spectrum consists of two linear curves with opposite slopes. Hence,
%    the spectral flow is zero; see Figure \ref{fig:7} in the zigzag case and Figure \ref{fig:8} in the armchair case. (More generally,
%    recall $\sigma_{L^2}(\DiKp(\mu))=-\sigma_{L^2}(\DiK(\mu))$.)
%     This is a consequence
%     of the identity $\var^\bKp=-\var^\bK$ (Proposition \ref{ips}), which is a consequence of $\CCC$-invariance.}
     
The operator $\tHd$  breaks $\CCC$, but not $\PPP$. The spectral flow is here equal to $\pm2$: see the explanations  of \S\ref{topology2}, the rigorous discussion in  \cite{Drouot:19} and the numerical simulations displayed in Figure \ref{PnoC}. The edge states curves cannot be continuously removed through compact operator perturbations: edge states of $\tHd$ are topologically protected. 

In the context of our operators, $H^\delta_{\rm edge}$, see  \cite{Drouot:19} for relations between spectral flow and the topology of the bulk operator band spectrum (bulk-edge correspondence), as well as extensions beyond gap-preserving perturbations.

\section{The resolvent expansion}\label{resolvent}

In this section we prove Theorem \ref{res-exp}. Symmetry properties of the Hamiltonian $H^\delta_{\rm edge}$ (whether $\PPP$- or $\CCC$- breaking)  play a role in the global character of edge state curves (\S\ref{topology2}-\ref{sec:5.3} ). 
However symmetry does not play a role in technical details of the resolvent expansion for energies near $E_D$.  

Our method of proof unifies the approaches of  \cite{FLW-2d_edge:16,FLW-MAMS:17} and \cite{Drouot:18b} to obtain the different resolvent expansions for zigzag-type and armchair-type edge orientations. For ease of presentation, we focus on the particular case of armchair-type edge states of the $\CCC-$ invariant Hamiltonian $H_{\edge}^\delta=H^\delta$. This case presents new hurdles due to the coupling of spectral components of the two high-symmetry quasi-momentum sublattices $\bK+\Lambda^*$ and $\bKp+\Lambda^*$. 

For armchair-type edges ($\R\tv_1$ with $\bK\cdot\tv_1=\bKp\cdot\tv_1=0$), the multiscale analysis of \S\ref{eff-Dirac} indicates that  the 
 point spectrum of $H^\delta_{\rm edge}$ acting in $\Ll^2_{\delta\mu}$, with energies uniformly  within spectral gap about $E_D$, is determined by the block-diagonal effective Dirac operator, 
 $\Dirac(\mu):H^1(\R;\C^4)\to L^2(\R;\C^4)$;  one block determining the $\bK$ sublattice contribution 
 and the other block determining the $\bKp$ sublattice contribution.

The operator, $\Dirac(\mu)$, has degenerate (multiplicity two) point spectrum associated
with $\DiK(\mu)$ and $\DiKp(\mu)$; see \S\ref{D-eff-spec}. Furthermore, the spectral components of the two high-symmetry ($\bK$ and $\bKp$) quasi-momentum lattices are coupled by the line-defect perturbation. Our analysis will show that this coupling is non-resonant. Therefore, the resolvent of  $H^\delta_{\rm edge}$ acting in $\Ll^2_{\delta\mu}$ is, to dominant order as $\delta\downarrow0$, given in terms of the resolvent of  $\Dirac(\mu)$ acting in $L^2(\R;\C^4)$.
   We carry out  the expansion in $\Ll^2_0=L^2(\R^2/\Z\tv_1)$. This approach also gives the result for $\kpar = \delta \mu$ stated in Theorem \ref{res-exp}.

We will proceed along three steps:
\begin{itemize}
\item[1.] The multiscale construction of \S\ref{eff-Dirac}  \cite{FLW-2d_edge:16,FLW-MAMS:17} suggests that edge states bifurcate from the Dirac point  $(E_D,\bK_\star)$; the approximate edge mode
is spectrally concentrated near the energy / quasimomentum pair $(E_D,\bK_\star)$.  Thus  in \S\ref{sec:4.1} we introduce orthogonal projectors $\Pi_\near$ and $\Pi_\far = \Id - \Pi_\near$ that localize in energy and quasi-momentum near and far from the Dirac points. Thus,
\begin{equation}\label{eq:1a}
\Ll^2_0 = \HH_\near \oplus \HH_\far, \ \ \  \HH_\near = \operatorname{Range}(\Pi_\near), \ \ \ \HH_\far = \operatorname{Range}(\Pi_\far).
\end{equation}
\item[2.]   In \S\ref{sec:4.2}, we express $H^\delta$ on $\Ll^2_0$ in terms of its action on the orthogonal summands $\HH_\near$ and $\HH_\far$  in \eqref{eq:1a}:
\begin{equation}\label{eq:1t}
\left. \left(H^\delta - E_D - \delta z\right) \right|_{\Ll^2_0} = \matrice{H^\delta_\near - E_D-\delta z & 0
\\
0 & H^0_\far - E_D - \delta z} + \OO_{\Ll^2_0}(\delta),
\end{equation}
where $H^\delta_\near=\Pi_\near\ H^\delta\ \Pi_\near $ and $H^\delta_\far = \Pi_\far\ H^\delta\ \Pi_\far$. 
By  the spectral no-fold condition (H2) (see \eqref{H2}) the operator $H^0_\far$  has spectrum which is bounded away from $E_D$. Hence, the invertibility of \eqref{eq:1t} is controlled by that of $H^\delta_\near - E_D-\delta z$ on $\HH_\near$. 
\item[3.]  The core of the proof of Theorem \ref{res-exp} is in \S\ref{sec:4.3}-\ref{sec:4.4}. There,   $H^\delta_\near$ is related to the Dirac operators $\DiKs$, first for $H^\delta_\near$ itself (Proposition \ref{HDi}), then at  the level of the resolvent (Proposition \ref{prop:3}). A Schur complement  / Lyapunov-Schmidt reduction strategy applied to \eqref{eq:1t} yields $\Ll^2_0$-invertibility of $H^\delta - E_D - \delta z$ and  an expansion of the resolvent provided $z$ is bounded away from  $\sigma_{L^2}(\DiK)\cup\sigma_{L^2}(\DiKp)$. \end{itemize}

\subsection{Projector near Dirac points}\label{sec:4.1} Since we have assumed that $\R\tv_1$ is an armchair-type edge,$\bK+t\R$ intersects both $\bK$ and the sublattice $\bKp+\Lambda^*$.  By $\Lambda^*$ periodicity of the Floquet-Bloch modes
 we may restrict attention to the segment (1D Brillouin zone) 
 $t\in[-\pi,\pi]\mapsto\bK + t\tk_2$; see \eqref{slice-rep}.
  By the no-fold condition (H2) (see \eqref{H2}), along this quasi-momentum segment $H_\bk$ is ``gapped '' away from $\bk=\bK$ and $\bk=\bKp$.
  
   In this section, we construct an operator $\Pi_\near$ that projects simultaneously near energy $E_D$ and quasimomenta $\bK$ and $\bKp$ along $\bK + t\tk_2$, $|t|\le\pi$. We first construct operators $\Pi_\near^\bKs$ that projects near $\bKs \in \{\bK,\bKp\}$ along $\bK + t\tk_2$, $|t|\le\pi$. 

We recall that $H^0_\bk$ is the operator $H^0=-\Delta+V$ acting on $L^2_\bk$. This operator is unitarily equivalent to $H^0(\bk) = e^{-i\bk \cdot \bx} \cdot H^0 \cdot e^{i\bk \cdot \bx}$ acting on $L^2_{\bf 0}=L^2(\R^2/\Lambda)$. Since $H^0(\bk)$ varies smoothly with $\bk$ and $E_D$ is an eigenvalue of $H^0_\bKs$ of multiplicity precisely two, there exist $\varepsilon, \eta > 0$ such that
\begin{equation}
|\bk-\bKs| \leq \eta \ \ \Rightarrow \ \ \text{ $H^0_\bKs$ has no eigenvalues but $E_D$ in $[E_D-\varepsilon,E_D+\varepsilon]$. };
\label{EDonly}\end{equation}
see, for example, \cite[Section 8]{Kato}. For fixed $\bk$, satisfying $ |\bk-\bKs| \leq \eta$, we may define the orthogonal projector
\begin{equation}\label{eq:2a}
\Pi_\bk^0 = \dfrac{1}{2\pi i}\int_{|z-E_D|=\varepsilon} \big(z-H^0_\bk\big)^{-1} dz  : \ L^2_\bk \rightarrow L^2_\bk\ .
\end{equation}
Note by\eqref{EDonly} that
\begin{equation}
\label{PiKs}
\Pi^0_\bKs= \Phi^\bKs_1\otimes\Phi^\bKs_1+\Phi^\bKs_2\otimes\Phi^\bKs_2\ \ 
{\rm and}\ \ \Pi^0_\bKs(\bx,\by)=\Phi^\bKs(\bx)^\top\Phi^\bKs(\by).
\end{equation}

   We introduce $\Pi_\near^\bKs$ (see \cite{FLW-2d_edge:16}), a direct integral of operators $\Pi^0_\bk$ of quasi-momenta near $\bKs$:
\begin{equations}\label{eq:2c}
\Pi_\near^\bKs \de \dfrac{1}{2\pi}\int_{[-\pi,\pi]}^\oplus \chi \left(\delta^{-2/3} t\right) \cdot \Pi^0_{\bKs+ t \tk_2} \ dt :  \ \Ll^2_0 \rightarrow \Ll^2_0.
\end{equations}
Here, $\chi$ is the characteristic function of the interval $[-1,1]$.  $\Pi_\near^\bKs$ is well-defined for $\delta^{2/3} \leq \eta$. Note that we are suppressing the dependence of 
$\Pi_\near^\bKs $ on $\delta$.

The operator $\Pi_\near^\bKs$  projects to components with quasi-momenta in  $\bK + t\tk_2$, $|t|\le\pi$, and at most $\delta^{2/3}$ distant from $\bKs$.
Since for armchair-type edges,  the segment 
 $\bK + t\tk_2$, $|t|\le\pi$, intersects both $\bK+\Lambda$ and $\bKp+\Lambda^*$,
 along this segment the operator $H_\bk^0$ is ``gapped''  only for 
quasimomenta bounded away from $\bK$ and $\bKp$. Thus we introduce a projector which 
excludes spectral components which are near either $\bK+\Lambda^*$ or $\bKp+\Lambda^*$
\begin{equation}
\Pi_\near \de \Pi_\near^\bK + \Pi_\near^\bKp.
\end{equation}
Components in the range of  $\Pi_\near^\bK$ and range of $\Pi_\near^\bKp$ are coupled by 
the defect perturbation. We show however that such coupling is non-resonant and
 prove that its effect is negligible;  see Proposition \ref{lem:3e}.

We next provide an expansion of $\Pi_\near^\bKs$.
\begin{proposition}\label{lem:2m} As $\delta \rightarrow 0$,
\begin{equations}
\Pi_\near^\bKs = \Pi_{\near}^{\bKs,(0)} + \OO_{\Ll^2_0}(\delta^{2/3}), \ \ \ \  \textrm{where:}
\\
 \Pi_{\near}^{\bKs,(0)} \de \dfrac{1}{2\pi} \int_{[-\pi,\pi]}^\oplus \chi(\delta^{-2/3} t)\  \ e^{i t \tk_2\cdot \bx}\ \Pi^0_\bKs\ e^{-i t \tk_2\cdot \bx} \  dt,
 \label{PiKs0}\end{equations}
 where $\Pi^0_\bKs : \ L^2_\bKs \rightarrow L^2_\bKs$ is the projector to $\C \Phi_1^\bKs \oplus \C \Phi_2^\bKs$; see \eqref{PiKs}.
\end{proposition}

\begin{proof}
1. By \eqref{eq:2a}, the projector
\begin{equation}\label{eq:2i}
N(\bk) \de e^{-i(\bk-\bKs)\cdot \bx}\cdot \Pi^0_\bk\cdot e^{i(\bk-\bKs)\cdot\bx}\ :\ L^2_\bKs\to L^2_\bKs, 
\end{equation}
 associated to the operator $H^0(\bk-\bKs) = e^{-i(\bk-\bKs)\cdot \bx} \cdot H^0 \cdot e^{-i(\bk-\bKs)\cdot \bx}$,   varies smoothly with $\bk$ \cite[Section 8]{Kato} . Hence, $N(\bk)$ has a Taylor expansion:
\begin{equation}\label{eq:2f}
N(\bk) = 
N(\bKs) + \OO_{L^2_\bKs}(\bk-\bKs) = \Pi^0_\bKs + \OO_{L^2_\bKs}(\bk-\bKs).
\end{equation}
The leading order term $\Pi^0_\bKs$  is the projection to $\C \Phi_1^\bKs \oplus \C \Phi_2^\bKs$
 displayed in \eqref{PiKs}.

2. Using \eqref{eq:2i} and \eqref{eq:2f}, with $\bk=\bK_\star + t\ktilde_2$, in \eqref{eq:2c} we get
\begin{equations}
\Pi_\near^\bKs = \dfrac{1}{2\pi}\int_{[-\pi,\pi]}^\oplus  \chi(\delta^{-2/3} t) \Pi^0_{\bKs+t \tk_2} \ dt 
= 
 \dfrac{1}{2\pi} \int_{[-\pi,\pi]}^\oplus   \chi(\delta^{-2/3} t) \cdot e^{i t \tk_2\cdot \bx} N(\bKs+t \tk_2) e^{-i t \tk_2\cdot \bx} \cdot dt
 \\
=\dfrac{1}{2\pi} \int_{[-\pi,\pi]}^\oplus  \chi(\delta^{-2/3} t) \cdot e^{i t \tk_2\cdot \bx} \Pi^0_\bKs e^{-i t \tk_2\cdot \bx}  \cdot dt + \dfrac{1}{2\pi} \int_{[-\pi,\pi]}^\oplus \chi(\delta^{-2/3} t) \cdot  \OO_{L^2_{\bKs+t\tk_2}}(t) dt,
\end{equations}
The first term in this expansion is precisely $\Pi_{\near}^{\bKs,(0)}$, displayed in \eqref{PiKs0}. The second term is $\OO_{\Ll^2_0}(\delta^{2/3})$ because $\chi(\delta^{-2/3} t)$ is supported in $|t| \le\delta^{2/3}$ and of \eqref{eq:2p}.  This completes the proof of Proposition \ref{lem:2m}.
\end{proof}
Next, we give an explicit expression for $\Pi_{\near}^{\bKs,(0)}$, the leading term 
 in \eqref{PiKs0}.
\begin{proposition}\label{lem:3d} We have :
\begin{equation}
\Pi_{\near}^{\bKs,(0)} = \TT_\bKs^* \chi(\delta^{-2/3} D_s) \TT_\bKs, \label{TTchi}
\end{equation}
Here,  $\TT_\bKs : L^2(\Sigma,\C^2) \rightarrow L^2(\R)$ and its adjoint $\TT_\bKs^*:L^2(\R)\to L^2(\Sigma,\C^2)$ ($\Sigma=\R^2/\R\tv_1$) are given by
\begin{equation}
\TT_\bKs u(t) = \int_{\R/\Z} \left( \ove{\Phi^\bKs} u\right)(s\tv_1+t\tv_2) ds, \ \ \ \ \TT_\bKs^* v(\bx) = \Phi^\bKs(\bx)^\top v(\tk_2\cdot \bx).
\label{TTKs}\end{equation}
\end{proposition}

\begin{proof} 1. We use \eqref{eq:9i} to express the kernel of $\Pi_{\near}^{\bKs,(0)}$:
\begin{equations}
\Pi_{\near}^{\bKs,(0)}(\bx,\by) = \dfrac{1}{2\pi} \int_{[-\pi,\pi]} \chi(\delta^{-2/3} t) \ e^{i t \tk_2\cdot \bx}\ \Pi^0_\bKs(\bx,\by)\ e^{-i t \tk_2\cdot \by} \ dt.
\end{equations}
Since  $\Pi^0_\bKs(\bx,\by)=\Phi^\bKs(\bx)^\top\ \ove{\Phi^\bKs(\by)}$ (see  \eqref{PiKs}), 
 the kernel of $\Pi_{\near}^{\bKs,(0)}$ as an operator acting on $L^2(\Sigma)$, $\Sigma=\R^2/\Z\vtilde_1$, is given by:
\begin{equations} \label{Pi0xy-z}
\Pi_{\near}^{\bKs,(0)}(\bx,\by) =    \dfrac{\Phi^\bKs(\bx)^\top}{2\pi} \int_{[-\pi,\pi]}\ \chi(\delta^{-2/3} \omega)\ e^{i \omega \tk_2\cdot (\bx-\by)} d\omega\ \ove{\Phi^\bKs(\by)}.
\end{equations}

2. On the other hand,
 \begin{equation}
 \left(\TT_\bKs^*\chi(\delta^{-2/3}D)\TT_\bKs\right)u(\bx) =
 \dfrac{\Phi^\bKs(\bx)^\top}{2\pi} \int_{\R}\chi(\delta^{-2/3}\omega) 
 \int_{\R } e^{i\omega(\ktilde_2\cdot\bx-t_2^\prime)}\ \left(\TT_\bKs u\right)(t_2^\prime) dt_2^\prime  d\omega.
  \label{Pi0xy-a}
  \end{equation}
  Writing out the  $dt_2^\prime$ integral we have:
 \[  \int_{\R}\ e^{i\omega(\ktilde_2\cdot\bx-t_2^\prime)}  \left(\TT_\bKs u\right)(t_2^\prime)  dt_2^\prime
   =  \int_{\R}dt_2^\prime  \int_0^1 dt_1^\prime  e^{i\omega(\ktilde_2\cdot\bx-t_2^\prime)} 
  \left( \overline{\Phi^\bKs}u\right)(t_1^\prime\tv_1+t_2^\prime\tv_2).
 \]
  Set $\by=t_1^\prime\tv_1+t_2^\prime\tv_2$. Hence $dt_1^\prime dt_2^\prime=
  |\tv_1\wedge\tv_2|\ d\by\ =\ d\by$ and $t_2^\prime=\tk_2\cdot\by$. Therefore,
   \begin{equation} 
   \int_{\R}  e^{i\omega(\ktilde_2\cdot\bx-t_2^\prime)} \left(\TT_\bKs u\right)(t_2^\prime)  dt_2^\prime
=  \int_\Sigma  e^{i\omega\tk_2\cdot(\bx-\by)}\cdot \overline{\Phi^\bKs(\by)} u(\by)  d\by.
  \label{Pi0xy-b}\end{equation}
Substitution of \eqref{Pi0xy-b} into \eqref{Pi0xy-a} yields:
 \begin{equation}
 \left(\TT_\bKs^*\ \chi(\delta^{-2/3}D)\ \TT_\bKs\right)u(\bx)  =  
   \dfrac{\Phi^\bKs(\bx)^\top}{2\pi}  \int_{\R_\omega}\chi(\delta^{-2/3}\omega)  
 \int_\Sigma\ e^{i\omega\tk_2\cdot(\bx-\by)}\ \overline{\Phi^\bKs(\by)}  u(\by)  d\by\ .
 \nn\end{equation}
 Hence,   for $\delta$ small:
 \begin{align}
 \left(\TT_\bKs^*\ \chi(\delta^{-2/3}D)\ \TT_\bKs\right)(\bx,\by)  &=  
  \dfrac{\Phi^\bKs(\bx)^\top}{2\pi}   \int_{\R}\chi(\delta^{-2/3}\omega) 
  e^{i\omega\tk_2\cdot(\bx-\by)} \cdot \overline{\Phi^\bKs(\by)} d\w\nn\\
  &=    \dfrac{\Phi^\bKs(\bx)^\top}{2\pi}  \int_{[-\pi,\pi]}\chi(\delta^{-2/3}\omega)\ 
  e^{i\omega\tk_2\cdot(\bx-\by)} \cdot \overline{\Phi^\bKs(\by)} d\w
  \label{Pi0xy-c}\end{align}
  Comparing \eqref{Pi0xy-c} with \eqref{Pi0xy-z} completes the proof of Proposition \ref{lem:3d}.
 \end{proof}

\subsection{Decomposition into near and far components}\label{sec:4.2}  Using the projections  $\Pi_\near$ and $\Pi_\far = \Id - \Pi_\near$, the space $\Ll^2_0$ splits into near and far quasi-momentum components:
\begin{equation}
\Ll^2_0 = \HH_\near \oplus \HH_\far, \ \ \ \ \HH_\near \de \Pi_\near\left( \Ll^2_0\right), \ \ \ \ \HH_\far \de \Pi_\far \left( \Ll^2_0 \right).
\end{equation}
We show here that  the invertibility of $H^\delta-E_D-\delta z$ on $\Ll^2_0$ reduces to that of 
\begin{equation}
\Pi_\near \left( H^\delta-E_D-\delta z \right) \Pi_\near : \ \HH_\near \rightarrow \HH_\near.
\end{equation}

We view $H^\delta - E_D -\delta z$  as a matrix operator acting the direct summands $\HH_\near $ and  $\HH_\far$ of $\Ll^2_0$:
\begin{equations}\label{eq:3o}
H^\delta =
 \matrice{H^\delta_\near- E_D -\delta z & \Pi_\near H^\delta \Pi_\far \\ \Pi_\far H^\delta \Pi_\near & H^\delta_\far - E_D -\delta z },  \ \ \ \  \textrm{where} \\ 
  H^\delta_\near \de \Pi_\near H^\delta \Pi_\near,\qquad H^\delta_\far \de \Pi_\far H^\delta \Pi_\far.
\end{equations}
We next simplify the off-diagonal components of \eqref{eq:3o}. Note that  $\Pi_\near$ is a direct integral of spectral projections associated to $H^0_\bk$ and therefore commutes with $H^0$. Moreover, $\Pi_\near \Pi_\far = 0$ and 
\[ H^\delta = H^0 + \delta \kappa_\delta W,\quad \kappa_\delta(\bx)=\kappa(\delta\tk_2\cdot\bx).\] We deduce that
$\Pi_\near H^\delta \Pi_\far = \delta \cdot \Pi_\near \cdot \kappa_\delta W \cdot \Pi_\far$ and 
similarly, $\Pi_\far \left(H^\delta- E_D -\delta z\right) \Pi_\near = \delta \cdot \Pi_\far  \cdot \kappa_\delta W \cdot  \Pi_\near$.
Therefore \eqref{eq:3o} becomes
\begin{equation}\label{Hnf}
H^\delta =
 \matrice{ H^\delta_\near- E_D -\delta z & \delta \cdot \Pi_\near  \cdot \kappa_\delta W \cdot  \Pi_\far \\ 
 \delta \cdot \Pi_\far  \cdot \kappa_\delta W \cdot  \Pi_\near & H^\delta_\far- E_D -\delta z}.
\end{equation}
We construct the resolvent $(H^\delta_\near- E_D -\delta z)^{-1}$ using the Schur complement:
\begin{lemma}\label{schur}
Let
$ \mathscr{M} = \matrice{A&B\\ C&D } $ be such that 
$D$ and $\E \de A - BD^{-1}C$ are invertible.\\  Then,  $\mathscr{M}$ is invertible and 
\begin{equation}
\mathscr{M}^{-1}=
 \matrice{  \E^{-1} & -\E^{-1}BD^{-1} \\
 - D^{-1}C\E^{-1} & D^{-1}+D^{-1}C\E^{-1}BD^{-1}}.
 \end{equation}
\end{lemma}
\nit We shall apply Lemma \ref{schur} with 
\begin{equations}
A = H^\delta_\near- E_D -\delta z,\quad B= \delta \cdot \Pi_\near  \cdot \kappa_\delta W \cdot  \Pi_\far\\
C =  \delta \cdot \Pi_\far  \cdot \kappa_\delta W \cdot  \Pi_\near,\quad D= H^\delta_\far- E_D -\delta z\ .
\end{equations}
 We first study the invertibility of $D=H^\delta_\far- E_D -\delta z$ on $\HH_\far$. We note that
\begin{equation}
H^\delta_\far- E_D -\delta z = H^0_\far-E_D + \delta \cdot \Pi_\far (\kappa_\delta W - z) \Pi_\far = H^0_\far - E_D + \OO_{\HH_\far}(\delta).
\end{equation}
Because of (H1) and (H3), we have for some $c_1 > 0$
\begin{equation}
\dist\left( \sigma_{\HH_\far}\left(H^0\right), E_D \right) \geq c_1 \delta^{2/3}.
\end{equation}
Therefore the operator $H^0-E_D : \HH_\far \rightarrow \HH_\far$ is invertible, and its inverse is $\OO_{\HH_\far}(\delta^{-2/3})$. 
A Neumann series argument implies that $D$ is invertible   on $\HH_\far$ with
$D^{-1} = \OO_{\HH_\far}(\delta^{-2/3})$.

We turn to the invertibility of $\E = A - BD^{-1}C : \ \HH_\near \rightarrow \HH_\near$. We have
\begin{equations}
\E = \left(H^\delta_\near- E_D -\delta z\right) - \delta^2  \cdot \Pi_\near \kappa_\delta W \cdot  \left(H^\delta_\far- E_D -\delta z\right)^{-1} \cdot \kappa_\delta W \Pi_\near .\label{Eexp}
\end{equations}
Since $D^{-1} = \OO_{\HH_\far}(\delta^{-2/3})$, the correction term in \eqref{Eexp} is of size $\OO_{\HH_\near}(\delta^{2-2/3})=\OO_{\HH_\near}(\delta^{4/3})$ and therefore
\begin{equation}\label{eq:3s}
\E= H^\delta_\near- E_D -\delta z + \OO_{\HH_\near}(\delta^{4/3}).
\end{equation}
To study the invertibility of $\E$, we analyze the invertibility of the leading order term in \eqref{eq:3s}. In \S\ref{sec:4.3} and \S\ref{sec:4.4} we will prove the two following propositions,
 expansions of $H^\delta_\near- E_D -\delta z $ and its inverse:

\begin{proposition}\label{HDi} %\label{prop:5} 
For $z$ varying in compact subsets of $\C$, as operators on $\Ll^2_0$,
\begin{align}\label{eq:4y}
&H^\delta_\near - E_D- \delta z 
\\
&= \delta \cdot (U_\delta \TT)^* \cdot  \chi(\delta^{1/3} D_s)
\matrice{\DiK -z & 0 \\ 0 & \DiKp - z}   \chi(\delta^{1/3} D_s)\cdot U_\delta \TT + \OO_{\Ll^2_0}\left(\delta^{4/3}\right).\nn
\end{align}
\end{proposition}

\begin{proposition}\label{prop:3} 
Fix $\epsi > 0$. There exists $\delta_0 > 0$ such that  for
\begin{equation}
\delta \in (0,\delta_0), \ \ |z| \leq \var_D - \epsi, \ \ \dist(\sigma_{L^2}(\DiK),z) \geq \epsi,
\end{equation}
the operator $H^\delta_\near - E_D- \delta z : \HH_\near \rightarrow \HH_\near$ is invertible; and as operators on $\HH_\near$,
\begin{equations}\label{eq:4g}
\left(H^\delta_\near - E_D- \delta z\right)^{-1} 
 = \dfrac{1}{\delta} \cdot (U_\delta \TT)^* \cdot \matrice{\DiK -z & 0 \\ 0 & \DiKp - z}^{-1}  \cdot U_\delta \TT  + \OO_{\HH_\near}\left(\delta^{-2/3}\right).
\end{equations}
\end{proposition}

\begin{proof}[Proof of Theorem \ref{res-exp} assuming Propositions \ref{HDi} and \ref{prop:3}] 1. Because of Proposition \ref{prop:3}, in the range specified by Theorem \ref{res-exp}, $H^\delta_\near - E_D- \delta z$ is invertible on $\HH_\near$ with norm $\OO_{\HH_\near}(\delta^{-1})$. A Neumann series argument then shows that $\E$ -- given in \eqref{eq:3s} -- is invertible on $\HH_\near$ with $\E^{-1} = O_{\HH_\near}(\delta^{-1})$. This implies
\begin{equations}\label{eq:2q}
\E^{-1} = \left( H^\delta_\near- E_D -\delta z + \OO_{\HH_\near}(\delta^{4/3}) \right)^{-1} 
\\
= \left(H^\delta_\near - E_D- \delta z\right)^{-1} + \OO_{\HH_\near}(\delta^{-2/3}).
\end{equations}

2. Because of $\E^{-1} = \OO_{\HH_\near}(\delta^{-1})$, $B = \OO_{\HH_\far \rightarrow \HH_\near}(\delta)$, $C = \OO_{\HH_\near \rightarrow \HH_\far}(\delta)$ and $D^{-1} = \OO_{\HH_\far}(\delta^{-2/3})$, Lemma \ref{schur} yields 
\begin{equations}\label{eq:1e}
\left( H^\delta-E_D-\delta z \right)^{-1} = \matrice{\E^{-1} & 0 \\ 0 & 0} + \OO_{\Ll^2_0}(\delta^{-2/3}),
\end{equations}
where the leading order term in \eqref{eq:1e} is $\OO_{\Ll^2_0}(\delta^{-1})$. We then substitute \eqref{eq:2q} into \eqref{eq:1e}. The leading order 
term is given by that of \eqref{eq:4g}. This completes the proof of  \eqref{ac-res}, the resolvent expansion $H^\delta$ acting in $\Ll^2_0$ (armchair-type edges). As noted at the start of the proof, the arguments apply $H^\delta$ acting in $\Ll^2_{\delta\mu}$ for all $|\mu|<\mu_0$, with $\delta$ sufficiently small. This completes the proof of Theorem \ref{res-exp} assuming 
 Propositions \ref{HDi} and \ref{prop:3}.
\end{proof}

In the following two sections we turn to the proofs of Propositions \ref{HDi} and \ref{prop:3}.

\subsection{Proof of Proposition \ref{HDi}}\label{sec:4.3}
Since 
\begin{equation}\label{eq:1f}
H^\delta_\near = H^0_\near + \delta \cdot \Pi_\near \cdot \kappa_\delta W \cdot \Pi_\near,\quad
 \kappa_\delta(\bx)\ =\ \kappa(\delta\tk_2\cdot\bx).
\end{equation}
to prove Proposition \ref{HDi}, we must expand terms arising from 
$H^0_\near$,  which are
dominated by the conical-crossings of Dirac points, and 
 those arising from $\Pi_\near  \kappa_\delta W \Pi_\near$, which arise from the domain wall. We treat these terms separately in the following two subsections.

\subsubsection{Expansion of $H^0_{\rm near}-E_D-\delta z$; contributions from the conical crossing}

\begin{proposition}\label{lem:3b} As $\delta \rightarrow 0$, uniformly for $z$ in compact subsets of $\C$,
 \begin{align}\label{eq:2s}
H^0_{\rm near}-E_D-\delta z\ &=\ \Pi^\bKs_\near \left( H^0-E_D - \delta z \right) \Pi^\bKs_\near
\\
&=
\TT_\bKs^* \chi(\delta^{-2/3} D_s) \cdot    \big( \vF^\bKs \sigma_1 D_s - \delta z\big)  \cdot \chi(\delta^{-2/3} D_s) \TT_\bKs + \OO_{\Ll^2_0}(\delta^{4/3}).
\nn
\end{align}
where $D_s=-i\D_s$ and the operators $\TT_\bKs$ and $\TT_\bKs^*$ are defined in \eqref{TTKs}. 
\end{proposition}

\begin{proof}[Proof of Proposition \ref{lem:3b}] 1. Since $\Pi_\near^\bKs$ is given by \eqref{eq:2c}, 
\begin{align}\label{eq:2r}
 & \Pi_\near^\bKs \left(H^0-E_D-\delta z\right) \Pi_\near^\bKs\\
\qquad &= \dfrac{1}{2\pi}\int_{[-\pi,\pi]}^\oplus \chi \left(\delta^{-2/3} \tau\right) \cdot \Pi^0_{\bKs+ \tau \tk_2} \left( H^0_{\bKs+ \tau \tk_2}-E_D - \delta z\right) \Pi^0_{\bKs+ \tau \tk_2} \cdot \chi \left(\delta^{-2/3} \tau\right) d\tau.
\nn\end{align}
Below we estimate the integrand of the direct integral in \eqref{eq:2r}.

2. Recall from \eqref{eq:2i} that $\Pi_\bk^0=e^{i(\bk-\bKs)\cdot\bx} N(\bk) e^{-i(\bk-\bKs)\cdot\bx}$, where $N(\bk):L^2_{\bKs}\to L^2_{\bKs}$ varies smoothly for $\bk$ near $\bKs$ .
 For such $\bk$, 
\begin{equations}
\Pi_\bk^0 \left( H_\bk^0-E_D \right) \Pi_\bk^0 = e^{i(\bk-\bKs)\cdot\bx} M(\bk) e^{-i(\bk-\bKs)\cdot\bx}, 
\\
\text{where}\ \  M(\bk)
 % \de e^{-i(\bk-\bKs) \bx} \Pi_\bk \left( H_\bk^0-E_D \right) \Pi_\bk e^{i(\bk-\bKs) \bx} 
 %\\ 
= N(\bk) \left( -(\nabla+i(\bk-\bKs))^2+V-E_D \right) N(\bk).
\end{equations} 
 is a smoothly varying family of operators acting on $L^2_\bKs$.

We next use the expansion of $N(\bk)$ in \eqref{eq:2f} to expand $M(\bk)$  about $\bk=\bKs$ up to quadratic corrections in $\bk - \bKs$. The leading order term vanishes because $M(\bKs) = 0$. We obtain
\begin{equations}
M(\bk) = \Pi_\bKs \cdot \big( -2i (\bk-\bKs) \nabla \big) \cdot \Pi_\bKs 
\\
+ \Pi_\bKs  \left( H^0_\bKs-E_D \right) \cdot \OO_{L^2_\bKs}(\bk-\bKs) + \OO_{L^2_\bKs}(\bk-\bKs) \cdot \left( H^0_\bKs-E_D \right) \Pi_\bKs   
 + \OO_{L^2_\bKs}(\bk-\bKs)^2.
\end{equations}
The second and third terms vanish because $\Pi_\bKs  \left( H^0_\bKs-E_D \right) =  \left( H^0_\bKs-E_D \right) \Pi_\bKs = 0$. We conclude that
\begin{equation}\label{eq:2d}
M(\bk) = \Pi_\bKs \cdot \big( -2i(\bk-\bKs) \cdot \nabla \big) \cdot \Pi_\bKs + \OO_{L^2_\bKs}\left(|\bk-\bKs|^2\right).
\end{equation}

3. We expand the integrand of \eqref{eq:2r} using \eqref{eq:2d} and \eqref{eq:2f}:
\begin{align}
&\Pi^0_{\bKs+ \tau \tk_2} \left( H^0_{\bKs+ \tau \tk_2}-E_D - \delta z\right) \Pi^0_{\bKs+ \tau \tk_2} 
 \nn \\
&\quad = e^{i \tau\ktilde_2\cdot\bx} \big( M(\bKs+ \tau \tk_2) -\delta z  \cdot N(\bKs+\tau\tk_2) \big) e^{-i \tau\ktilde_2\cdot\bx}
\nn \\
&\quad = e^{i \tau\ktilde_2\cdot\bx} \Pi^0_\bKs \big( -2i \tau \ktilde_2\cdot\nabla - \delta z \big)  \Pi^0_\bKs e^{-i \tau\ktilde_2\cdot\bx} \ + \ \OO_{L^2_{\bKs+\tau\ktilde_2}}(\tau^2+ z\delta \tau)
\label{eq:2dd}
\end{align}
Substitution of this expansion into the direct integral \eqref{eq:2r}, we obtain
\begin{equations}
H^0_\near-E_D -\delta z= \Pi_\near \left( H^0-E_D-\delta z \right) \Pi_\near 
\\
= \dfrac{1}{2\pi} \int_{[-\pi,\pi]}^\oplus \chi \left(\delta^{-2/3} \tau\right) \cdot e^{i \tau\ktilde_2\cdot\bx} \Pi^0_\bKs \left( -2i \tau \ktilde_2\cdot\nabla - \delta z \right)  \Pi^0_\bKs e^{-i \tau\ktilde_2\cdot\bx} \cdot \chi \left(\delta^{-2/3} \tau\right) d\tau
\end{equations}
\begin{equations}\label{eq:2e}
+ \int_{[-\pi,\pi]}^\oplus \chi \left(\delta^{-2/3} \tau\right) \cdot
\OO_{L^2_{\bKs+\tau\ktilde_2}}(\tau^2+ z\delta \tau) \chi \left(\delta^{-2/3} \tau\right) d\tau
\end{equations}
Since $\chi\left(\delta^{-2/3} \tau\right) = 0$ for $|\tau| \geq \delta^{2/3}$, the remainder term in \eqref{eq:2e}  is a direct integral of operators that are on  $\OO_{L^2_{\bKs+\tau \tk_2}}(\delta^{4/3})$. 
Therefore it is $\OO_{\Ll^2_0}(\delta^{4/3})$.

Next we focus on the leading order term in the expansion \eqref{eq:2e}. The operator $\Pi_\bKs  \cdot \big(-2 i\tk_2\cdot \nabla \big) \cdot \Pi_\bKs$ acts on $\C\Phi_1^\bKs\oplus\C\Phi_2^\bKs$. By
 Proposition \ref{ips}, with respect to this basis, $\Pi_\bKs  \cdot (-2 i\tk_2 \nabla ) \cdot \Pi_\bKs$ has the matrix representation:
\begin{equation}
\blr{\Phi^\bKs, -2 i \tk_2 \cdot  \nabla (\Phi^\bKs)^\top} = \vF^\bKs \sigma_1\ ,
\end{equation}
where $\vF^\bKs=\vF$ for $\bKs=\bK$ and $\vF^\bKs=-\vF$ for $\bKs=\bKp$.
 Thus
 $\Pi_\bKs  \cdot \big( -2 i\tk_2 \nabla \big) \cdot \Pi_\bKs$, as an operator acting on $L^2_\bKs$, has the 
 kernel
\begin{equation}
(\bx,\by) \in \R^2/\Lambda \times \R^2/\Lambda \mapsto \Phi^\bKs(\bx)^\top \cdot
 \vF^\bKs \sigma_1 \cdot \ove{\Phi^\bKs(\by)}.
\end{equation}
We deduce that the leading order term in \eqref{eq:2e} has kernel
\begin{equation}\label{eq:3u}
\dfrac{\Phi^\bKs(\bx)^\top}{2\pi}\int_{[-\pi,\pi]}^\oplus \chi^2 \left(\delta^{-2/3} \tau\right)  \big(\vF^\bKs \sigma_1 \tau-\delta z\big)  e^{i \tau \tk_2 (\bx- \by)} d\tau  \cdot \ove{\Phi^\bKs(\by)}
\end{equation}
where $(\bx,\by) \in \left(\R^2/\Z \tv_1\right)^2$.
As in the proof of Proposition \ref{lem:3d}, we see that \eqref{eq:3u} is the kernel of the operator
\begin{equations}
 \TT_\bKs^* \chi(\delta^{-2/3} D_s) \cdot    \big( \vF^\bKs \sigma_1 D_s - \delta z\big)  \cdot \chi(\delta^{-2/3} D_s) \TT_\bKs.
\end{equations}
This completes the proof of Proposition \ref{lem:3b}.
\end{proof}

Because $H^0-E_D - \delta z$ commutes with $\Pi_\near^\bKs$, and $\Pi_\near^\bK \Pi_\near^\bKp = \Pi_\near^\bKp \Pi_\near^\bK = 0$ as long as $\delta$ is sufficiently small, we deduce that
\begin{equation}
\Pi_\near^\bK \left(H^0-E_D - \delta z\right)\Pi_\near^\bKp = \Pi_\near^\bKp \left(H^0-E_D - \delta z\right)\Pi_\near^\bK = 0.
\end{equation}
We sum the expansions over $\bKs \in \{\bK,\bKp\}$, recalling that $\Pi_\near = \Pi_\near^\bK + \Pi_\near^\bKp$ and $\TT = \TT_\bK\oplus \TT_\bKp$.  
Together with Proposition \ref{lem:3b}, we conclude that as $\delta \rightarrow 0$, uniformly for $z$ in compact subsets of $\C$:  $H^0_\near-E_D - \delta  z$ equals
\begin{align}\label{eq:3y}
&H^0_\near-E_D - \delta  z\\ 
&\quad =\ \TT^* \chi(\delta^{-2/3} D_s) \cdot    \matrice{ \vF^\bKs \sigma_1 D_s - \delta z & 0 \\ 0 & \vF^\bKs \sigma_1 D_s - \delta z}  \cdot \chi(\delta^{-2/3} D_s) \TT + \OO_{\Ll^2_0}(\delta^{4/3}).
\nn
\end{align}

We note from \eqref{dilate} the following identities:
\begin{equations}\label{eq:9g}
D_s = \delta \cdot U_\delta^*\ D_s\ U_\delta; \ \ \ \ \chi(\delta^{-2/3} D_s)\ U_\delta^* = U_\delta^* \ \chi(\delta^{1/3} D_s) ; \ \ \\ U_\delta\ \chi(\delta^{-2/3} D_s)  =  \chi(\delta^{1/3} D_s)\ U_\delta.
\end{equations}
Therefore, we obtain from \eqref{eq:3y} that $H^0_\near-E_D - \delta z$
 \begin{equations}\label{eq:3yy}
 =
\delta \cdot \TT^* \chi(\delta^{-2/3} D_s)\ U_\delta^* \cdot  
  \matrice{ \vF^\bKs \sigma_1 D_s -  z & 0 \\ 0 & \vF^\bKs \sigma_1 D_s -  z}  \cdot U_\delta^*\ \chi(\delta^{-2/3} D_s)\ \TT 
  \ +\  \OO_{\Ll^2_0}(\delta^{4/3})
\\
=
\delta \cdot (U_\delta \TT)^*\ \chi(\delta^{1/3} D_s) \cdot    \matrice{ \vF^\bKs \sigma_1 D_s -  z & 0 \\ 0 & \vF^\bKs \sigma_1 D_s -  z}  \cdot \chi(\delta^{1/3} D_s)\ (U_\delta \TT)\ +\ \OO_{\Ll^2_0}(\delta^{4/3}).
\end{equations}
Thus, we have  proved \eqref{eq:4y} for $H^0_\near=H^\delta_\near-\delta\Pi_\near\ \cdot \kappa_\delta W \cdot\ \Pi_\near $.

\subsubsection{Contribution of the domain wall; proof of  Proposition \ref{HDi}}

We must expand the expression $\delta\Pi_\near\ \cdot \kappa_\delta W \cdot\ \Pi_\near $,
where $\kappa_\delta(\bx)=\kappa(\delta\tk_2\cdot\bx)$, for $\delta$ small. 

\begin{proposition}\label{lem:3c}  As $\delta \rightarrow 0$, uniformly for $z$ in compact subsets of $\C$,
\begin{equation}\label{eq:2t}
\Pi_\near^\bKs \cdot \kappa_\delta W \cdot \Pi_\near^\bKs = \TT_\bKs^* \cdot  \chi(\delta^{-2/3} D_s) \cdot \var^\bKs \sigma_3 \kappa(\delta \cdot) \cdot  \chi(\delta^{-2/3} D_s) \cdot \TT_\bKs
 + \OO_{\Ll^2_0}\left(\delta^{2/3}\right).
\end{equation}
\end{proposition}

\begin{proof} 1. First,  Proposition \ref{lem:2m} and  $\kappa \in L^\infty$ imply that
\begin{equations}\label{eq:5a}
\Pi_\near^\bKs \cdot \kappa_\delta W \cdot \Pi_\near^\bKs = \Pi_\near^{\bKs,(0)} \cdot \kappa_\delta W \cdot \Pi_\near^{\bKs,(0)}\ +\ \OO_{\Ll^2_0}\left(\delta^{2/3}\right).
\end{equations}
Furthermore, by Proposition \ref{lem:3d}, the leading term  in \eqref{eq:5b} is 
\begin{equation}
\Pi_\near^{\bKs,(0)} \cdot \kappa_\delta W \cdot \Pi_\near^{\bKs,(0)}  = \TT_\bKs^* \chi(\delta^{-2/3}D_s) \cdot \TT_{\bKs}  \kappa_\delta W  \TT_\bKs^* \cdot \chi(\delta^{-2/3}D_s) \TT_{\bKs}.
\label{5ab}\end{equation}
 Using the definitions of $\TT_{\bKs}$ and its adjoint in \eqref{TTKs} , we find that the inner expression in \eqref{5ab}, $\TT_{\bKs}  \kappa_\delta W  \TT_\bKs^*$, is a multiplication operator:
\begin{equations}
 \left(\TT_{\bKs}\ \kappa_\delta W\ \TT_\bKs^*\right) u(t) = \kappa(\delta t)\ F(t)\ u(t), \ \ \ \ \text{where}
\\
F(s) \de \int_{\R/\Z}\ \ove{\Phi^\bKs(t\tv_1 + s\tv_2)} \ W(t\tv_1 + s\tv_2) \Phi^\bKs(t\tv_1 + s\tv_2)^\top\ dt.
\end{equations}
We deduce the following expression for the dominant term in \eqref{eq:5a}:
\begin{equation}
\Pi_\near^{\bKs,(0)} \cdot \kappa_\delta W\ \Pi_\near^{\bKs,(0)}  = \TT_\bKs^* \chi(\delta^{-2/3}D_s) \cdot \kappa_\delta F \cdot \chi(\delta^{-2/3}D_s) \TT_{\bKs}.
\end{equation}

\nit 2. $F$ is smooth and one-periodic with an absolutely convergent  Fourier series:
\begin{equation}\label{eq:5d}
F(s) = \sum_{m \in 2\pi \Z}   \hF_m\cdot e^{im s}, \ \ \ \ \hF_m = \int_0^1 e^{-im sÕ}\ F(s') ds'.
\end{equation}
We deduce that
\begin{equations}\label{eq:5b}
\Pi_\near^{\bKs,(0)} \cdot \kappa_\delta W \cdot \Pi_\near^{\bKs,(0)} = \sum_{m \in 2\pi \Z}  \TT_\bKs^*  \hF_m \cdot \chi(\delta^{-2/3}D_s)   \cdot \kappa(\delta \cdot) e^{im \cdot} \cdot \chi(\delta^{-2/3}D_s) \TT_{\bKs}
\\
= \sum_{m \in 2\pi \Z}    \TT_\bKs^*  \hF_m \cdot \chi(\delta^{-2/3}D_s)    \kappa(\delta \cdot)   \chi\big(\delta^{-2/3}(D_s-m)\big) \cdot e^{im \cdot} \TT_{\bKs}.
\end{equations}
This sum may be expressed as the  $m=0$ term plus a sum over $m\ne0$ terms.
We next bound this sum
from above and show that it is negligible for $\delta$ small.

3. Fix $m \neq 0$ and consider the operator $\chi\big(\delta^{-2/3}D_s\big) \cdot \kappa(\delta \cdot)  \cdot \chi\big(\delta^{-2/3}(D_s-m)\big)$, appearing in \eqref{eq:5b}. The function $\xi \mapsto \chi\big(\delta^{-2/3}(\xi-m)\big)$ is the indicator function of the interval $[m-\delta^{2/3},m+\delta^{2/3}]$. For all $\delta\le\delta_0$ ($\delta_0$ sufficiently small), none of these sets intersects the interval $[-1,1]$. Thus we may rewrite
 $\chi\big(\delta^{-2/3}(\xi-m)\big)$ as:
\begin{equation}
\chi\big(\delta^{-2/3}(\xi-m)\big) = \xi \cdot \dfrac{\chi\big(\delta^{-2/3}(\xi-m)\big)}{\xi} \de \xi \cdot \psi_{\delta,m}(\xi),\quad\textrm{where}\quad  |\psi_{\delta,m}(\xi)| \leq 1.
\end{equation}
Hence, $\chi\big(\delta^{-2/3}(D_s-m)\big)=D_s\ \psi_{\delta,m}(D_s)$. Substitution of this expression and commuting $D_s$ through $\kappa$ gives:
\begin{equations}
\chi\big(\delta^{-2/3}D_s\big) \cdot \kappa(\delta \cdot)  \cdot \chi\big(\delta^{-2/3}(D_s-m)\big)
=
\chi\big(\delta^{-2/3}D_s\big) \cdot \kappa(\delta \cdot)  D_s \cdot \psi_{\delta,m}(D_s)
\\
=
\chi\big(\delta^{-2/3}D_s\big)D_s\cdot \kappa(\delta \cdot)  \cdot  \psi_{\delta,m}(D_s) - \chi\big(\delta^{-2/3}D_s\big) \cdot [D_s,\kappa(\delta,\cdot)] \cdot  \psi_{\delta,m}(D_s).
\end{equations}
The first term involves the multiplier $\chi(\delta^{-2/3} D_s)D_s$, which is $\OO_{L^2}(\delta^{2/3})$ and $\psi_{\delta,m}(D_s)$, which is bounded on $L^2$ with norm $\le1$. Therefore this term is
 $\OO_{L^2}(\delta^{2/3})$, uniformly in $m \neq 0$. 
 The second term involves the commutator $[D_s,\kappa(\delta,\cdot)]$, which is $\OO_{L^2}(\delta)$ because $\kappa' \in L^\infty$.  Hence the second term satisfies the operator bound $\OO_{L^2}(\delta)$, uniformly in $m \neq 0$. We conclude 
 \begin{equation}\label{eq:5c}
\chi\big(\delta^{-2/3}D_s\big) \cdot \kappa(\delta \cdot)  \cdot \chi\big(\delta^{-2/3}(D_s-m)\big) = \OO_{L^2}(\delta^{2/3}).
\end{equation}
Summing over $m \neq 0$ and using that $\sum_m |F_m|$ is finite, because $F$ is smooth, 
\begin{equations}
\Bigg\| \sum_{m \in 2\pi \Z/\setminus\{0\}}\     \TT_\bKs^*  \hF_m \cdot \chi(\delta^{-2/3}D_s)    \kappa(\delta \cdot)   \chi\big(\delta^{-2/3}(D_s-m)\big) \cdot e^{im \cdot} \TT_{\bKs} \Bigg\|_{\Ll^2_0} 
\\
\leq C \delta^{2/3} \sum_{m \in 2\pi \Z/\setminus\{0\}}  \big| \hF_m \big| = O(\delta^{2/3}).
\end{equations}

4. It follows that the dominant contribution is from the $m=0$ term:
\begin{equation}
\Pi_\near^\bKs \cdot \kappa_\delta W \cdot \Pi_\near^\bKs =   \TT_\bKs^* \cdot \chi(\delta^{-2/3}D_s)  \cdot \hF_0\  \kappa(\delta \cdot) \cdot  \chi\big(\delta^{-2/3}D_s\big) \cdot  \TT_{\bKs} + \OO_{\Ll^2_0}(\delta^{2/3}).
\end{equation}
Finally, we observe from  \eqref{eq:5d} that
\begin{equation}
\hF_0 = \int_{\R^2/\Lambda} \ove{\Phi^\bKs(\bx)} \cdot W(\bx) \Phi^\bKs(\bx)^\top d\bx = \var^\bKs \sigma_3,
\end{equation}
where we substituted $\bx = s\tv_1 + t\tv_2$ and used the definition of $\var^\bKs$. The proof of Proposition \ref{lem:3c} is complete. 
\end{proof}

Proposition \ref{lem:3c} extracts the dominant term arising from  $\bK$-to-$\bK$ and $\bKp$-to-$\bKp$ quasimomentum coupling due to the domain wall perturbation $\kappa_\delta W$. Likewise, we must study $\bK$-to-$\bKp$ coupling via $\kappa_\delta W$. The next proposition shows that this interaction is negligible. This has the important  consequence that the effective Dirac operator is block-diagonal.

\begin{proposition}\label{lem:3e} As $\delta \rightarrow 0$, uniformly for $z$ in compact subsets of $\C$,
\begin{equation}\label{eq:3v}
\Pi_\near^\bK\cdot \kappa_\delta W \cdot  \Pi_\near^\bKp= \OO_{\Ll^2_0}\left(\delta^{2/3}\right), \ \ \ \ \Pi_\near^\bKp \cdot  \kappa_\delta W \cdot  \Pi_\near^\bK= \OO_{\Ll^2_0}\left(\delta^{2/3}\right).
\end{equation}
\end{proposition}

\begin{proof} The proof follows a similar strategy to that of Proposition \ref{lem:3c}. We estimate
 $\Pi_\near^\bK \cdot \kappa_\delta W \cdot  \Pi_\near^\bKp$; the adjoint bound gives the estimate for $\Pi_\near^\bKp\cdot  \kappa_\delta W\cdot  \Pi_\near^\bK$.

1. In analogy with Step 1 in the proof of Proposition \ref{lem:3c},
\begin{equations}\label{eq:5f}
\Pi_\near^\bK \cdot \kappa_\delta W \cdot \Pi_\near^\bKp =\ \Pi_\near^{\bK,(0)} \cdot \kappa_\delta W \cdot \Pi_\near^{\bKp,(0)}\ +\ \OO_{\Ll^2_0}\left( \delta^{2/3}\right) \\ 
= \TT_\bK^* \chi(\delta^{-2/3}D_s)   \cdot \kappa(\delta \cdot) G \cdot \chi(\delta^{-2/3}D_s) \TT_{\bKp} + \OO_{\Ll^2_0}\left( \delta^{2/3}\right),
\end{equations}
where the function $G$ is given by
\begin{equation}\label{eq:5g}
G(s) = \int_{\R/\Z} \ove{\Phi^\bK(t\tv_1 + s \tv_2)}  W(t\tv_1 + s\tv_2) \Phi^\bKp(t\tv_1 + s\tv_2)^\top dt.
\end{equation}

2. In contrast with $F(t)$ in \eqref{eq:5d},  the function $G(t)$ is not one-periodic; we have instead
$G(t+1) = e^{i(\bK-\bKp)\cdot\tv_2} \cdot G(t)$.
Therefore, we can write
\begin{equation}\label{eq:5i}
G(s) = \sum_{m \in S}  \hG_m\cdot e^{im s}, \ \ 
 S=(\bK-\bKp)\cdot\tv_2 + 2\pi\Z,\ \textrm{and}\ \  \hG_m = \int_0^1 e^{ -im s'}\cdot G(s') ds' . 
\end{equation}
Because $G$ is smooth, this series converges absolutely. The analog of \eqref{eq:5b} is 
\begin{equations}\label{eq:5j}
\Pi_\near^{\bK,(0)} \cdot \kappa_\delta W \cdot \Pi_\near^{\bKp,(0)} 
= \sum_{m \in S}   \TT_\bK^*  \hG_m \cdot \chi(\delta^{-2/3}D_s)    \kappa(\delta \cdot)   \chi\big(\delta^{-2/3}(D_s-m)\big) \cdot e^{im \cdot} \TT_{\bKp}.
\end{equations}

We claim that $\dist(0,S) = 2\pi/3$. Indeed, $(\bK-\bKp)\cdot\tv_2$ is either $-2\pi/3, 0$ or $2\pi/3$ modulo $2\pi$. If it is equal to $0$, then we would have
\begin{equation}
(\bK-\bKp)\cdot\tv_2 = 0 \mod 2\pi, \ \ \ \ (\bK-\bKp)\cdot\tv_1 = 0 \mod 2\pi
\end{equation}
(the second equality holds because $\R \tv_1$ is an armchair edge). It would follow that $\bK-\bKp \in 2\pi \Z\tk_1 \oplus 2\pi \Z \tk_2 = \Lambda^*$, which is not possible because $\bK \neq \bKp \mod \Lambda^*$.

3. Since $0\notin S$, arguments analogous to those in Step 3 of Proposition \ref{lem:3c} yield:
\begin{equations}
\left\|\Pi_\near^{\bK,(0)} \cdot \kappa_\delta W \cdot \Pi_\near^{\bKp,(0)} \right\|_{\Ll^2_0} 
\leq C \delta^{2/3} \sum_{m \in S}  \big| \hG_m \big| = O(\delta^{2/3}).
\end{equations}
We conclude that
$\Pi_\near^{\bK,(0)} \cdot \kappa_\delta W \cdot \Pi_\near^{\bKp,(0)} = \OO_{\Ll^2_0}(\delta^{2/3})$. Together with \eqref{eq:5f}, this completes the proof of Proposition \ref{lem:3e}.
We remark that the arguments in Steps 2. and 3. above, are an alternative to the Poisson summation arguments employed in \cite{FLW-2d_edge:16,FLW-MAMS:17}. \end{proof}

Recall that $\Pi_\near  = \Pi_\near^\bK + \Pi_\near^\bKp$. We sum the expansions of  Propositions \ref{lem:3c} and bounds of Proposition \ref{lem:3e} over $\bKs \in \{\bK,\bKp\}$, and then multiply by $\delta$, to deduce:
\begin{align}\label{eq:4z}
&\Pi_\near \cdot \delta \kappa_\delta W \cdot \Pi_\near\\
&\  =  \TT^* \cdot  \chi(\delta^{-2/3} D_s)   \cdot \delta  \kappa_\delta  \matrice{\var^\bK \sigma_3  & 0 \\ 0 & \var^\bKp \sigma_3} \cdot  \chi(\delta^{-2/3} D_s) \cdot \TT + \OO_{\Ll^2_0}\left(\delta^{4/3}\right)
\nn \\
&\  =\delta \cdot \TT^* \cdot  \chi(\delta^{-2/3} D_s) U_\delta^*  \cdot    \kappa  \matrice{\var^\bK \sigma_3  & 0 \\ 0 & \var^\bKp \sigma_3}\cdot   U_\delta \chi(\delta^{-2/3} D_s) \cdot \TT + \OO_{\Ll^2_0}\left(\delta^{4/3}\right) 
 \nn\\
&\ =\delta \cdot (U_\delta \TT)^* \cdot  \chi(\delta^{1/3} D_s)  \cdot    \kappa  \matrice{\var^\bK \sigma_3  & 0 \\ 0 & \var^\bKp \sigma_3}\cdot    \chi(\delta^{1/3} D_s) \cdot U_\delta \TT + \OO_{\Ll^2_0}\left(\delta^{4/3}\right)  
\nn\end{align}
Here, we used that $\TT = \TT_\bK \oplus \TT_\bKp : \Ll^2_0 \rightarrow L^2(\R,\C^4)$ (see \eqref{TT}) and the relations \eqref{eq:9g}. The expression \eqref{eq:4z} exhibits the domain-wall contribution in the RHS of \eqref{eq:4y}. 

Summing \eqref{eq:3yy} and \eqref{eq:4z} yields \eqref{eq:4y}:
\begin{equations}\label{eq:3l}
H^\delta_\near - E_D- \delta z = H^0_\near - E_D- \delta z + \Pi_\near \cdot \delta \kappa_\delta W \cdot \Pi_\near
\\
=
 \delta\cdot (U_\delta \TT)^* \cdot  \chi(\delta^{1/3} D_s) \cdot \matrice{\DiK -z & 0 \\ 0 & \DiKp - z}  \cdot \chi(\delta^{1/3} D_s) \cdot U_\delta\TT + \OO_{\Ll^2_0}\left(\delta^{4/3}\right).
\end{equations}
This completes the proof of Proposition \ref{HDi}.

\subsection{Resolvent expansion}\label{sec:4.4}

In this section, we use the expansion \eqref{eq:3l} for $H^\delta_\near - E_D- \delta z$ to obtain Proposition \ref{prop:3}, the corresponding expansion for the resolvent. We recall that at the end of  \S\ref{sec:4.2} we showed how to deduce Theorem \ref{res-exp} from  Proposition \ref{prop:3}.

\begin{proof}[Proof of Proposition \ref{prop:3}] By  Proposition \ref{HDi}:
\[ \frac{1}{\delta}\left(H_{\rm near}^\delta-E_D-\delta z\right)\ =\ L_\delta(z)\ +\  \OO_{\Ll^2_0}(\delta^{1/3}),\] where
\begin{equation}\label{Ldel}
L^\delta(z) \de (U_\delta \TT)^* \cdot \chi(\delta^{1/3} D_s) \ \left( \Dirac  -z\right) \chi(\delta^{1/3} D_s)\  (U_\delta \TT)  \ \text{ and } \ \  \Dirac  = \matrice{\DiK  & 0 \\ 0 & \DiKp }.
\end{equation}
Therefore, 
\begin{equation}
  \left(H_{\rm near}^\delta-E_D-\delta z\right)\ \frac{1}{\delta}\ (U_\delta\TT)^*\ \left(\ \Dirac   -z\ \right)^{-1}\ (U_\delta\TT)\ =\ A_\delta(z)\ +\  \OO_{\Ll^2_0}(\delta^{1/3}),\label{op-prod}
  \end{equation}
where $A^\delta(z) $ is the operator product:
 \begin{align}\label{eq:4h}
 &A^\delta(z) \de L^\delta(z) \cdot
 \left( (U_\delta\TT)^* \ \left( \Dirac-z \right)^{-1} \ (U_\delta \TT )\ \right)\\
&\  = (U_\delta \TT)^* \cdot \chi(\delta^{1/3} D_s)  \left( \Dirac-z \right)  \chi(\delta^{1/3} D_s) \cdot (U_\delta   \TT)\  (U_\delta\TT )^* \cdot \left(\ \Dirac   -z\ \right)^{-1}  \cdot (U_\delta \TT). 
\nn \end{align}

In Steps 1-3 below, we prove that $A^\delta(z) = \Pi_\near+\OO_{\Ll^2_0}(\delta^{-1/3})$. In Step 4-5, we conclude that $H^\delta_\near - E_D- \delta z$ is invertible for $\delta$ sufficiently small and $z\notin\sigma_{L^2}(\Dirac)=\sigma_{L^2}(\DiK)\cup\sigma_{L^2}(\DiKp)$, and we finally prove the expansion \eqref{eq:4g}. 

1. We focus on the operator $U_\delta \TT \TT^* U_\delta^*$ appearing in \eqref{eq:4h}. 
Using  \eqref{TT} and \eqref{TT*} we see that $\TT \TT^* - \Id$ is a multiplication operator:
\begin{equation}
\left(\TT \TT^*\ -\ \Id\right) u(s) = f(s) \cdot u(s), \ \ \ \ f(s) \de \int_{\R/\Z}  \ove{\Phi^\bKs(t\tv_1+s\tv_2)} \ \Phi^\bKs(t\tv_1+s\tv_2)^\top dt - 1.
\end{equation} 
Hence, $U_\delta \TT \TT^* U_\delta^* - \Id$ is the operator: multiplication by $f(\delta^{-1} \cdot)$\ .

The function $f$ is one-periodic with null average. Therefore, we may write $D_s F = f$, where  $F$ is one-periodic. Moreover, for all functions $u, g \in H^1$,
\begin{equation}
\blr{(U_\delta f) u, g}_{L^2} = \blr{f, U_\delta^* (g\ove{u})}_{L^2}  = \blr{F, D_s U_\delta^* (g\ove{u})}_{L^2} .
\end{equation}
Using the product rule for derivatives, we deduce that
\begin{equation}
\left| \blr{(U_\delta f) u, g}_{L^2} \right| \leq C\delta |F|_{L^\infty} |g|_{H^1} |u|_{H^1}. 
\end{equation}
By duality,  the multiplication operator by $u\mapsto (U_\delta f)u$ is bounded from $H^1$ to $H^{-1}$, with norm at most $C\delta |F|_{L^\infty}$. Therefore, we have shown
\begin{equation}\label{eq:4b}
U_\delta \TT \TT^* U_\delta^* - \Id = \OO_{H^1 \rightarrow H^{-1}}(\delta).
\end{equation}

2. We shall make use of the following bounds on the operators $\chi(\delta^{1/3} D_s)$ and $\Id-\chi(\delta^{1/3} D_s)$ between Sobolev spaces:
\begin{align}
\big\| \chi(\delta^{1/3}D_s) \big\|_{H^{-2}\to L^2} \le  \sup_{|\xi|\le\delta^{-1/3}} \left(1+|\xi|^2\right)^{-1} = O\left( \delta^{-2/3} \right);   \label{eq:1x-a}\\
\big\| \chi(\delta^{1/3}D_s)-\Id \big\|_{H^1 \to L^2} \le  \sup_{|\xi|\ge\delta^{-1/3}} (1+|\xi|^2)^{-1/2} = O \left( \delta^{1/3}\right). \label{eq:1x-b}
\end{align}

The operator $\chi(\delta^{1/3}D_s)$ commutes with $D_s$. Thus it is bounded from $H^{-1}$ to itself and satisfies the bound   $\big\| \chi(\delta^{1/3}D_s) \big\|_{H^{-1}} = 1$.  By \eqref{eq:1x-a}, it is also bounded from $H^{-2}$ to $L^2$ with   bound $O(\delta^{-2/3})$. Hence, $\chi(\delta^{1/3}D_s) \left( \Dirac-z \right)\chi(\delta^{1/3}D_s) = \OO_{H^{-1} \rightarrow L^2}(\delta^{-2/3})$. 
Combined with \eqref{eq:4b}, we get
\begin{equation}
\chi(\delta^{1/3}D_s) \left( \Dirac-z \right) \chi(\delta^{1/3}D_s) \cdot \left(U_\delta \TT \TT^* U_\delta^* - \Id\right) = \OO_{H^1 \rightarrow L^2}(\delta^{-2/3}\cdot\delta)=\OO_{H^1 \rightarrow L^2}(\delta^{1/3}).
\end{equation}
Since $z$ is at fixed distance from the spectrum of $\Dirac  $ and $\kappa, \kappa' \in L^\infty$, the operator 
$\left( \Dirac-z \right)^{-1}$ maps $L^2$ to $H^1$ with bounded norm. We deduce that
\begin{equation}\label{eq:4a}
\chi(\delta^{1/3}D_s) \left( \Dirac-z \right) \chi(\delta^{1/3}D_s) \cdot \left(U_\delta \TT \TT^* U_\delta^* - \Id\right) \cdot \left( \Dirac-z \right)^{-1} = \OO_{L^2}(\delta^{1/3}).
\end{equation}

3. We write $(U_\delta \TT)\  (U_\delta\TT) ^* = \Id + \left(U_\delta \TT \TT^* U_\delta^*-\Id\right)$ in \eqref{eq:4h} and apply \eqref{eq:4a}:
\begin{equation}
A^\delta(z)  = (U_\delta\TT)^*  \cdot  \chi(\delta^{1/3}D_s) \left( \Dirac-z \right) \chi(\delta^{1/3}D_s) \left( \Dirac-z \right)^{-1} \cdot U_\delta\TT + \OO_{\Ll^2_0}(\delta^{1/3}).
\end{equation}
Furthermore, inserting $\chi = 1+(\chi-1)$ and using the definition of $U_\delta$ we have
\begin{align}
A^\delta(z)  &= (U_\delta\TT)^*  \cdot  \chi(\delta^{1/3}D_s) \left( \Dirac-z \right) \ \left( \Dirac-z \right)^{-1} \cdot U_\delta\TT\nn\\
&\  \  +\  (U_\delta\TT)^*  \cdot  \chi(\delta^{1/3}D_s) \left( \Dirac-z \right) \ (\chi(\delta^{1/3}D_s)-1)\ \left( \Dirac-z \right)^{-1}\ (U_\delta\TT)\ +\ \OO_{\Ll^2_0}(\delta^{1/3}).
\nn\\
&\de\ \TT^*\ \chi(\delta^{-2/3}D_s)\ \TT\ +\ B_\delta(z)\ +\ \OO_{\Ll^2_0}(\delta^{1/3})
\label{B-def}\end{align}

For the first term on the right hand side of \eqref{B-def}  we have by Proposition \ref{PiKs0}:
\[\TT^*\ \chi(\delta^{-2/3}D_s)\ \TT\ =\ 
\Pi_\near^{(0)}\ =\ \Pi_{\rm near}+\OO(\delta^{2/3})\ .
%\TT^*\ U_\delta^*\chi(\delta^{1/3} D_s)\ U_\delta\ \TT
\]
 We bound the operator $B_\delta(z):L^2\to L^2$ 
 as follows. The operator $B_\delta(z)$ is a conjugation by $U_\delta\TT$ of the operator
 \begin{align}
& \chi(\delta^{1/3}D_s) \left( \Dirac-z \right) \ (\chi(\delta^{1/3}D_s)-1)\ \left( \Dirac-z \right)^{-1}\nn
\\ 
&=\ \chi(\delta^{1/3}D_s)\ (\chi(\delta^{1/3}D_s)-1)\  \left( \Dirac-z \right) \ \left( \Dirac-z \right)^{-1}
\nn\\
&\quad  +\ \chi(\delta^{1/3}D_s)\ [\left( \Dirac-z \right), (\chi(\delta^{1/3}D_s)-1)]\ \left( \Dirac-z \right)^{-1}\nn\\
&=\ \OO\left(\ \|\ \chi(\delta^{1/3}D_s)\ \kappa\ \left(1-\chi(\delta^{1/3}D_s)\right)\ \|_{H^1\to L^2} \  \| \left( \Dirac-z \right)^{-1}\|_{L^2\to H^1}\ \right)\ =\ \OO_{L^2}(\delta^{1/3})
\label{Az} \end{align}
To obtain the second to last equality in \eqref{Az}, we have used that $\chi\ (1-\chi)=0$ and that the operator $\Dirac$ is equal to matrix-valued function of $D_s$ plus multiplication operator equal to:   $\kappa$ times a constant matrix. The bound in \eqref{Az} then follows from \eqref{eq:1x-b} and the assumption $\kappa\in L^\infty$. It follows then, that 
$B_\delta(z)=\OO_{\Ll^2_0}(\delta^{1/3})$. Furthermore,
$A^\delta(z) \ =\ \Pi_{\rm near}\ +\ \OO_{\Ll^2_0}(\delta^{1/3})$ and by \eqref{op-prod}
\begin{equation}
\left( H_\near^\delta-E_D-\delta z \right) \cdot \left( \dfrac{1}{\delta} (U_\delta\TT)^* \cdot \left( \Dirac-z \right)^{-1}  \cdot U_\delta \TT\right) = \Pi_\near + O_{\Ll^2_0}(\delta^{1/3}).
\end{equation} %\Dirac  -z\

4. 
Using the hypothesis that $\dist\left(z,\sigma_{L^2}\left(\Dirac \right)\right) \geq \epsi$, a Neumann series argument implies that $H_\near^\delta-E_D-\delta z$ is invertible on $\HH_\near$ with
\begin{equation}\label{eq:4i}
\left(H_\near^\delta-E_D-\delta z \right)^{-1} =
\dfrac{1}{\delta} \cdot \Pi_\near (U_\delta\TT)^* \cdot \left( \Dirac-z \right)^{-1}  \cdot \left(U_\delta \TT\right) \Pi_\near + \OO_{\Ll^2_0}(\delta^{-2/3}).
\end{equation}

5. To conclude the proof of Proposition \ref{prop:3}, we must replace  $\Pi_\near$ in \eqref{eq:4i}
 by the identity operator. Proposition \ref{lem:2m} and the identity $\chi(\delta^{-2/3} D_s)  =  U_\delta^*\chi(\delta^{1/3} D_s)U_\delta$ imply
\begin{equations}
\Pi_\near=\Pi_\near^{(0)}+\OO_{\Ll^2_0}(\delta^{2/3})=
(U_\delta\TT)^*\chi(\delta^{1/3} D_s)U_\delta \TT + \OO_{\Ll^2_0}(\delta^{2/3}).
\end{equations}
Using \eqref{eq:4b} and $\chi(\delta^{1/3}D_s) = \OO_{L^2 \to H^1}(\delta^{-1/3})$, we deduce that
\begin{equations}\label{eq:4d}
 \left(U_\delta \TT\right) \Pi_\near = \left(U_\delta \TT\right)  (U_\delta \TT)^*  \cdot \chi(\delta^{1/3}D_s) U_\delta\TT  + \OO_{\Ll^2_0}(\delta^{2/3}) 
 \\ =
\chi(\delta^{1/3}D_s)  \left(U_\delta \TT\right) + \OO_{\Ll^2_0 \rightarrow \Hh^1_0}(\delta^{2/3}).
\end{equations}
The dual bound to \eqref{eq:4d} is 
\begin{equation}\label{eq:1c}
\Pi_\near (U_\delta \TT)^*   = 
(U_\delta \TT)^* \chi(\delta^{1/3}D_s)  + \OO_{\Hh_0^1 \rightarrow \Ll^2_0}(\delta^{2/3}).
\end{equation}
Acting on  \eqref{eq:4d} with $\Pi_\near (U_\delta \TT)^* \left( \Dirac-z \right)^{-1}$ and using \eqref{eq:1c}, we deduce that
\begin{equations}
\Pi_\near (U_\delta \TT)^* \cdot \left( \Dirac-z \right)^{-1} \cdot \left(U_\delta \TT\right) \Pi_\near \\
 = (U_\delta \TT)^* \cdot \chi(\delta^{1/3}D_s)\  \left( \Dirac-z \right)^{-1}\ \chi(\delta^{1/3}D_s) \cdot  \left(U_\delta \TT\right) \ + \ \OO_{\Ll^2_0}(\delta^{2/3}).
\label{eq:4e}
\end{equations}

We eliminated  $\Pi_\near$ from \eqref{eq:4i} at the expense of introducing $\chi(\delta^{1/3}D_s)$.
We now replace 
the operator  $\chi(\delta^{1/3}D_s)$ by the identity -- with an error $\OO_{\Ll^2_0}(\delta^{2/3})$.
Observe that
\begin{equations}\label{chDich}
\chi(\delta^{1/3}D_s)  \left( \Dirac-z \right)^{-1} \chi(\delta^{1/3}D_s) - \left( \Dirac-z \right)^{-1}\\
  = \left(\chi(\delta^{1/3}D_s) -1 \right)  \left( \Dirac-z \right)^{-1} \chi(\delta^{1/3}D_s) + 
  \left( \Dirac-z \right)^{-1} \left( \chi(\delta^{1/3}D_s) - 1 \right).
\end{equations}
We bound the first term in \eqref{chDich} by studying the operator prefactor $\left(\chi(\delta^{1/3}D_s) -1 \right)  \left( \Dirac-z \right)^{-1}$ in $L^2$; the  factor $\chi(\delta^{1/3}D_s)$ is clearly bounded on $L^2$. Since $\left( \Dirac-z \right)^{-1}$ is bounded from $L^2$ to $H^1$ and $\chi(\delta^{1/3}D_s) - \Id = \OO_{H^1 \to L^2}(\delta^{1/3})$ (see \eqref{eq:1x-b}), 
 the product is  $\OO_{L^2}(\delta^{1/3})$ and so the first term in  \eqref{chDich} is  $\OO_{L^2}(\delta^{1/3})$. The second term in \eqref{chDich} is the adjoint of the operator prefactor just studied. It is therefore  $\OO_{\Ll^2_0}(\delta^{1/3})$. We conclude that
\begin{equations}
\chi(\delta^{1/3}D_s)   \left( \Dirac-z \right)^{-1} \chi(\delta^{1/3}D_s) = \left( \Dirac-z \right)^{-1} + \OO_{\Ll^2_0}(\delta^{1/3})
\label{chDch1}\end{equations}
Substituting \eqref{chDch1} into \eqref{eq:4e} yields 
\begin{align}\nn
\Pi_\near (U_\delta \TT)^* \cdot \left( \Dirac-z \right)^{-1} \cdot \left(U_\delta \TT\right) \Pi_\near= (U_\delta \TT)^*\cdot \left( \Dirac-z \right)^{-1} \cdot  \left(U_\delta \TT\right) +  \OO_{\Ll^2_0}(\delta^{1/3}).
\label{eq:4ee}
\end{align}
Substituting back into \eqref{eq:4i} we get
\begin{equation}\label{eq:4ii}
\left(H_\near^\delta-E_D-\delta z \right)^{-1} =
\dfrac{1}{\delta} \cdot  (U_\delta \TT)^*\ \cdot \left( \Dirac-z \right)^{-1}\  \left(U_\delta\TT\right)  + \OO_{\Ll^2_0}(\delta^{-2/3}).
\end{equation}
The proof of Proposition \ref{prop:3} is now complete.
\end{proof}

\bibliographystyle{plain}
\bibliography{kkp25-autonum}

\begin{bibdiv}
\begin{biblist}

\bib{Aea:18}{article}{
      author={Ammari, H.},
      author={Fitzpatrick, B.},
      author={Lee, H.},
      author={Hiltunen, E.~Orvehed},
      author={Yu, S.},
       title={Honeycomb-lattice {M}innaert bubbles},
        date={2018},
     journal={Preprint, arXiv:1811.03905},
}

\bib{avron-simon:78}{article}{
      author={Avron, J.E.},
      author={Simon, B.},
       title={Analytic properties of band functions},
        date={1978},
     journal={Annals of Physics},
      volume={110},
       pages={85\ndash 101},
}

\bib{B19a}{article}{
      author={Bal, G.},
       title={Continuous bulk and interface description of topological
  insulators},
        date={2019},
     journal={J. Math. Phys.},
      volume={\textbf{60}},
      number={8},
       pages={20pp},
}

\bib{B19c}{article}{
      author={Bal, G.},
       title={Topological invariants for interface modes},
        date={2019},
     journal={Preprint, arXiv:1906.08345},
}

\bib{B19b}{article}{
      author={Bal, G.},
       title={Topological protection of perturbed edge states},
        date={2019},
     journal={Commun. Math. Sci.},
      volume={\textbf{17}},
      number={1},
       pages={193\ndash 225},
}

\bib{Bal:17}{article}{
      author={Bal, G.},
       title={Topological protection of perturbed edge states},
        date={arXiv:1709.00605},
}

\bib{berkolaiko-comech:18}{article}{
      author={Berkolaiko, G.},
      author={Comech, A.},
       title={Symmetry and {D}irac points in graphene spectrum},
        date={2018},
     journal={Journal of Spectral Theory},
      volume={8},
      number={3},
       pages={1099\ndash 1147},
}

\bib{BKR17}{article}{
      author={Bourne, C.},
      author={Kellendonk, J.},
      author={Rennie, A.},
       title={The {K}-theoretic bulk-edge correspondence for topological
  insulators},
        date={2017},
     journal={Ann. Henri Poincar\'e},
      volume={\textbf{18}},
      number={5},
       pages={1833\ndash 1866},
}

\bib{BR18}{article}{
      author={Bourne, C.},
      author={Rennie, A.},
       title={Chern numbers, localisation and the bulk-edge correspondence for
  continuous models of topological phases.},
        date={2018},
     journal={Math. Phys. Anal. Geom.},
      volume={\textbf{21}},
      number={3},
}

\bib{B19}{article}{
      author={Braverman, M.},
       title={Spectral flows of {T}oeplitz operators and bulk-edge
  correspondence.},
        date={2019},
     journal={Lett. Math. Phys.},
      volume={\textbf{109}},
      number={10},
       pages={2271\ndash 2289},
}

\bib{Marquardt:17}{article}{
      author={C.~Brendel, O.~Painter, V.~Peano},
      author={Marquardt, F.},
       title={Snowflake topological insulator for sound waves},
        date={2017},
     journal={https:},
        ISSN={/arxiv.or},
}

\bib{C91}{article}{
      author={de~Verdiere, Y.~Colin},
       title={Sur les singularites de {V}an {H}ove generiques},
        date={1991},
     journal={Memoires de la S. M. F. serie 2},
      volume={\textbf{46}},
       pages={99\ndash 109},
}

\bib{Drouot:19}{article}{
      author={Drouot, A.},
       title={The bulk-edge correspondence for continuous honeycomb lattices},
        date={2019},
     journal={Comm. Partial Differential Equations},
      volume={44},
      number={12},
       pages={1406\ndash 1430},
}

\bib{Drouot:18b}{article}{
      author={Drouot, A.},
       title={Characterization of edge states in perturbed honeycomb
  structures},
        date={2019},
     journal={Pure Appl. Anal.},
      volume={1},
      number={3},
       pages={385\ndash 445},
}

\bib{D:19}{article}{
      author={Drouot, A.},
       title={Microlocal analysis of the bulk-edge correspondence},
        date={2019},
     journal={Preprint, arXiv:1909.10474},
}

\bib{Drouot:18a}{article}{
      author={Drouot, A.},
       title={The bulk-edge correspondence for continuous dislocated systems},
        date={https://arxiv.org/abs/1810.10603},
}

\bib{DFW:18}{article}{
      author={Drouot, A.},
      author={Fefferman, C.~L.},
      author={Weinstein, M.~I.},
       title={Defect modes for dislocated periodic media},
        date={https://arxiv.org/abs/1810.05875},
}

\bib{EG02}{article}{
      author={Elbau, P.},
      author={Graf, G.~M.},
       title={Equality of bulk and edge {H}all conductances revisted},
        date={2002},
     journal={Comm. Math. Phys.},
      volume={\textbf{229}},
       pages={415\ndash 432},
}

\bib{EGS05}{article}{
      author={Elgart, A.},
      author={Graf, G.~M.},
      author={Schenker, J.~H.},
       title={Equality of the bulk and the edge {H}all conductances in a
  mobility gap},
        date={2005},
     journal={Comm. Math. Phys.},
      volume={\textbf{259}},
       pages={185\ndash 221},
}

\bib{Faure:19}{article}{
      author={Faure, F.},
       title={Manifestation of the topological index formula in quantum waves
  and geophysical waves},
        date={arXiv:1901.10592},
}

\bib{FLW-PNAS:14}{article}{
      author={Fefferman, C.~L.},
      author={Lee-Thorp, J.~P.},
      author={Weinstein, M.~I.},
       title={Topologically protected states in one-dimensional continuous
  systems and {D}irac points},
        date={2014},
     journal={Proceedings of the National Academy of Sciences},
       pages={07391},
}

\bib{FLW-2d_materials:15}{article}{
      author={Fefferman, C.~L.},
      author={Lee-Thorp, J.~P.},
      author={Weinstein, M.~I.},
       title={Bifurcations of edge states -- topologically protected and
  non-protected -- in continuous 2d honeycomb structures},
        date={2016},
     journal={2D Mater.},
      volume={3},
       pages={014008},
}

\bib{FLW-2d_edge:16}{article}{
      author={Fefferman, C.~L.},
      author={Lee-Thorp, J.~P.},
      author={Weinstein, M.~I.},
       title={Edge states in honeycomb structures},
        date={2016},
     journal={Annals of PDE},
      volume={2},
      number={12},
}

\bib{FLW-MAMS:17}{article}{
      author={Fefferman, C.~L.},
      author={Lee-Thorp, J.~P.},
      author={Weinstein, M.~I.},
       title={Topologically protected states in one-dimensional systems},
        date={2017},
     journal={Memoirs of the American Mathematical Society},
      volume={247},
      number={1173},
}

\bib{FW:12}{article}{
      author={Fefferman, C.~L.},
      author={Weinstein, M.~I.},
       title={Honeycomb lattice potentials and {D}irac points},
        date={2012},
     journal={J. Amer. Math. Soc.},
      volume={25},
      number={4},
       pages={1169\ndash 1220},
}

\bib{FW:14}{article}{
      author={Fefferman, C.~L.},
      author={Weinstein, M.~I.},
       title={Wave packets in honeycomb lattice structures and two-dimensional
  {D}irac equations},
        date={2014},
     journal={Commun. Math. Phys.},
      volume={326},
       pages={251\ndash 286},
         url={DOI: 10.1007/s00220-013-1847-2},
}

\bib{FLW-CPAM:17}{article}{
      author={Fefferman, C.L.},
      author={Lee-Thorp, J.P.},
      author={Weinstein, M.I.},
       title={Honeycomb {S}chroedinger operators in the strong-binding regime},
        date={2018},
     journal={Comm. Pure Appl. Math.},
      volume={71},
      number={6},
}

\bib{G09}{article}{
      author={Grushin, V.~V.},
       title={Multiparameter perturbation theory of {F}redholm operators
  applied to {B}loch functions},
        date={2009},
     journal={Mathematical Notes},
      volume={\textbf{86}},
      number={6},
       pages={767\ndash 774},
}

\bib{GYZ:19}{article}{
      author={Guo, H.},
      author={Yang, X.},
      author={Zhu, Y.},
       title={Bloch theory based gradient recovery method for computing
  topological edge modes in photonic graphene},
        date={2019},
     journal={J. Comp. Phys.},
      volume={379},
       pages={403\ndash 420},
}

\bib{HR:07}{article}{
      author={Haldane, F. D.~M.},
      author={Raghu, S.},
       title={Possible realization of directional optical waveguides in
  photonic crystals with broken time-reversal symmetry},
        date={2008},
     journal={Phys. Rev. Lett.},
      volume={100},
      number={1},
       pages={013904},
}

\bib{H93}{article}{
      author={Hatsugai, Y.},
       title={The {C}hern number and edge states in the integer quantum hall
  effect},
        date={1993},
     journal={Phys. Rev. Lett.},
      volume={\textbf{71}},
       pages={3697\ndash 3700},
}

\bib{Kato}{article}{
      author={Kato, T.},
       title={Perturbation theory for linear operators},
        date={1995},
     journal={Classics in Mathematics. Springer-Verlag, Berlin},
}

\bib{KRS02}{article}{
      author={Kellendonk, J.},
      author={Richter, T.},
      author={Schulz-Baldes, H.},
       title={Edge current channels and {C}hern numbers in the integer quantum
  {H}all effect.},
        date={2002},
     journal={Rev. Math. Phys.},
      volume={\textbf{14}},
      number={1},
       pages={87\ndash 119},
}

\bib{KS04b}{article}{
      author={Kellendonk, J.},
      author={Schulz-Baldes, H.},
       title={Boundary maps for crossed products with an application to the
  quantum {H}all effect},
        date={2004},
     journal={Comm. Math. Phys.},
      volume={\textbf{3}},
      number={3},
       pages={611\ndash 637},
}

\bib{KS04a}{article}{
      author={Kellendonk, J.},
      author={Schulz-Baldes, H.},
       title={Quantization of edge currents for continuous magnetic operators},
        date={2004},
     journal={J. Funct. Anal.},
      volume={\textbf{209}},
      number={2},
       pages={388\ndash 413},
}

\bib{Shvets-PTI:13}{article}{
      author={Khanikaev, A.~B.},
      author={Mousavi, S.~H.},
      author={Tse, W.-K.},
      author={Kargarian, M.},
      author={MacDonald, A.~H.},
      author={Shvets, G.},
       title={Photonic topological insulators},
        date={2013},
     journal={Nature Materials},
      volume={12},
       pages={233\ndash 239},
}

\bib{Kuchment:16}{article}{
      author={Kuchment, P.A.},
       title={An overview of periodic elliptic operators},
        date={2016},
     journal={Bull. Amer. Math. Soc.},
      volume={53},
      number={3},
       pages={343\ndash 414},
}

\bib{LWZ:18}{article}{
      author={Lee-Thorp, J.~P.},
      author={Weinstein, M.~I.},
      author={Zhu, Y.},
       title={Elliptic operators with honeycomb symmetry; {D}irac points, edge
  states and applications to photonic graphene},
        date={2018},
     journal={Arch. Rational Mech. Anal.},
}

\bib{Mak-etal:14}{article}{
      author={Mak, K.F.},
      author={McGill, K.L.},
      author={Park, J.},
      author={McEuen, P.L.},
       title={The valley hall effect in $mos_2$ transistors},
        date={2014},
     journal={Science},
      volume={344},
      number={6191},
       pages={1489\ndash 1492},
}

\bib{irvine:15}{article}{
      author={Nash, L.~M.},
      author={Kleckner, D.},
      author={Read, A.},
      author={Vitelli, V.},
      author={Turner, A.~M.},
      author={Irvine, W. T.~M.},
       title={Topological mechanics of gyroscopic materials},
        date={2015},
     journal={Proceedings of the National Academy of Sciences},
      volume={112},
      number={47},
       pages={14495\ndash 14500},
}

\bib{RMP-Graphene:09}{article}{
      author={Neto, A. H.~Castro},
      author={Guinea, F.},
      author={Peres, N. M.~R.},
      author={Novoselov, K.~S.},
      author={Geim, A.~K.},
       title={The electronic properties of graphene},
        date={2009},
     journal={Reviews of Modern Physics},
      volume={81},
       pages={109\ndash 162},
}

\bib{Rechtsman-etal:18}{article}{
      author={Noh, J.},
      author={Huang, S.},
      author={Chen, K.~P.},
      author={Rechtsman, M.~C.},
       title={Observation of photonic topological valley {H}all edge states},
        date={2018},
     journal={Phys. Rev. Lett.},
      volume={120},
       pages={063902},
}

\bib{ozawa_etal:18}{article}{
      author={Ozawa, T.},
      author={Price, H.~M.},
      author={Amo, A.},
      author={Goldman, N.},
      author={Hafezi, M.},
      author={Lu, L.},
      author={Rechtsman, M.},
      author={Schuster, D.},
      author={Simon, J.},
      author={Zilberberg, O.},
      author={Carusotto, I.},
       title={Topological photonics},
        date={preprint},
}

\bib{QAP18}{article}{
      author={Qian, K.},
      author={Apigo, D.~J.},
      author={Prodan, C.},
      author={Barlas, Y.},
      author={Prodan, E.},
       title={Topology of the valley-chern effect.},
        date={2018},
     journal={Phys. Rev. B},
      volume={\textbf{98}},
      number={10},
       pages={155138},
}

\bib{RH:08}{article}{
      author={Raghu, S.},
      author={Haldane, F. D.~M.},
       title={Analogs of quantum-{H}all-effect edge states in photonic
  crystals},
        date={2008},
     journal={Phys. Rev. A},
      volume={78},
      number={3},
       pages={033834},
}

\bib{RS4}{book}{
      author={Reed, M.},
      author={B., Simon},
       title={Methods of modern mathematical physics: Analysis of operators,
  volume iv},
   publisher={Academic Press},
        date={1978},
        ISBN={9780125850049},
}

\bib{artificial-graphene:11}{article}{
      author={Singha, A.},
      author={Gibertini, M.},
      author={Karmakar, B.},
      author={Yuan, S.},
      author={Polini, M.},
      author={Vignale, G.},
      author={Kastnelson, M.~I.},
      author={Pinczuk, A.},
      author={Pfeiffer, L.~N.},
      author={West, K.~W.},
      author={Pellegrini, V.},
       title={Two-dimensional {M}ott-{H}ubbard electrons in an artificial
  honeycomb lattice},
        date={2011},
     journal={Science},
      volume={332},
       pages={1176},
}

\bib{T14}{article}{
      author={Taarabt, A.},
       title={Equality of bulk and edge {H}all conductances for continuous
  magnetic random {S}chr\"odinger operators},
        date={2014},
     journal={Preprint, arXiv:1403.7767},
}

\bib{Soljacic-etal:08}{article}{
      author={Wang, Z.},
      author={Chong, Y.~D.},
      author={Joannopoulos, J.~D.},
      author={Soljacic, M.},
       title={Reflection-free one-way edge modes in a gyromagnetic photonic
  crystal},
        date={2008},
     journal={Phys. Rev. Lett.},
      volume={100},
       pages={013905},
}

\bib{Waterstraat:16}{article}{
      author={Waterstraat, N.},
       title={Fredholm operators and spectral flow.},
        date={arXiv:1603.02009},
}

\end{biblist}
\end{bibdiv}
\end{document}